\documentclass[a4paper]{article}
\usepackage[affil-it]{authblk}
\usepackage{latexsym}
\usepackage{xspace}
\usepackage{graphics}

\usepackage[vlined,ruled,linesnumbered]{algorithm2e}
\usepackage{algorithmic}

\usepackage{xcolor}

\usepackage[colorlinks=false, pdfborder={0 0 0},
            pdftitle={Approximate Tree Decompositions of Planar Graphs in Linear Time},pdfauthor={Frank Kammer and Torsten Tholey},pdfcreator={Frank Kammer and Torsten Tholey},
           pdfkeywords={tree decomposition, planar graph}]{hyperref}

\usepackage{amsfonts,graphicx,amsmath}
\usepackage{amssymb}
\usepackage{amsthm}

\usepackage{epsfig,pstricks,pst-node}
\usepackage{pst-all}

\def\markatright#1{\leavevmode\unskip\nobreak\quad\hspace*{\fill}{#1}}

\makeatletter
\def\@maketitle{%
  \newpage
  \null
  \vskip 2em%
  \begin{center}%
  \let \footnote \thanks
    {\Large\bfseries \@title \par}%
    \vskip 1.5em%
    {\normalsize
      \lineskip .5em%
      \begin{tabular}[t]{c}%
        \@author
      \end{tabular}\par}%
    \vskip 1em%
    {\normalsize \@date}%
  \end{center}%
  \par
  \vskip 1.5em}
\makeatother

\newtheorem {theorem}{Theorem}[section]
\newtheorem{definition}[theorem]{Definition}{\bf}{\it}
\newtheorem{lemma}[theorem]{Lemma}{\bf}{\it}

\newtheorem{corollary}[theorem]{Corollary}{\bf}{\it}
\newtheorem{observation}[theorem]{Observation}

\def\cost{%
  weight%
}

\def\markatright#1{\leavevmode\unskip\nobreak\quad\hspace*{\fill}{#1}}

\newcounter{FitemZ}

\newcommand{\Nat}{I\!\!N}

\newcommand{\down}{\downarrow}
\newcommand{\downvertex}{\!\!\downarrow}

\definecolor{myredblue}{rgb}{0.5,0,0.5}

\newcommand\spp[2]{$(#1,#2)$}

\def\markatright#1{\leavevmode\unskip\nobreak\quad\hspace*{\fill}{#1}}

\newcommand\hide[1]{}

\begin{document}

\title{Approximate Tree Decompositions \\
  of Planar Graphs in Linear Time}

\author{Frank Kammer$^*$}
\author{Torsten Tholey%
  \thanks{E-mail address: \texttt{\{kammer,tholey\}@informatik.uni-augsburg.de}}}
\affil{Institut f\"ur Informatik, Universit\"at Augsburg, Germany}

\date{\vspace{-5ex}}

\maketitle
\begin{abstract}
Many algorithms have been developed for $\mathrm{NP}$-hard
problems on graphs with small treewidth $k$. 
For example, all problems that are expressible in linear 
extended monadic second order
can be solved in linear time on graphs of bounded treewidth.
It turns out that
the bottleneck of many 
algorithms for $\mathrm{NP}$-hard problems
is the computation of a tree
decomposition of width $O(k)$.
In particular, by the bidimensional theory,
there are many linear extended monadic second order problems that 
can be solved on $n$-vertex planar graphs with treewidth $k$ in
a time linear in $n$ and subexponential in $k$ if a 
tree decomposition of width $O(k)$ can be found in such a time.

We present the first algorithm that, on $n$-vertex planar graphs with
treewidth $k$, finds a tree decomposition of width $O(k)$ in such a time.
In more detail, our algorithm has a running time of $O(n k^{2}\log k)$.
We show 
the result as a special case of a result concerning
so-called weighted 
treewidth of weighted graphs.


\end{abstract}

\section{Introduction}\label{sec:intro}
\newcommand{\PN}{\mathcal{N}}

The treewidth, extensively studied by Robertson and
Seymour~\cite{RobS84}, is one of
the basic parameters in graph theory. Intuitively,
the treewidth measures the similarity of a graph to a tree by means of a
so-called
tree decomposition.
A tree decomposition of
width $r$---defined precisely in the beginning of
Section~\ref{sec:ide}---is a decomposition of a graph $G$ into small
subgraphs 
part of a so-called bag
such that each bag
contains at most 
$r+1$ vertices and such that
the bags
are
connected by a tree-like structure.
The treewidth $\mathrm{tw}(G)$ of a graph $G$ is the smallest $r$ for which
$G$ has a tree decomposition of width~$r$.

Often, $\mathrm{NP}$-hard
problems are solved on graphs $G$
with
small treewidth
by the following
two steps: First, compute a tree decomposition
for $G$ of 
width~$r\in \Nat$
\newpage

\noindent and second, solve the problem by using this tree decomposition.
Unfortunately, there is a trade-off
between the running times of these two steps depending on our
choice of $r$.
Very often the first step is the bottleneck.
For example,
Arnborg, Lagergren and Seese \cite{ArnLS91} showed that, for
  all problems expressible in so-called
  linear extended monadic second order (linear EMSO), the second step
  runs on $n$-vertex graphs in a time linear in~$n$.
Demaine, Fomin, Hajiaghayi, and Thilikos \cite{DemFHT05} have shown that,
for many so-called bidimensional problems that are also
linear EMSO problems,
one can find a solution of a size $\ell$ in a given $n$-vertex graph---if one
exists---in a time
linear in $n$
and subexponential in $\ell$ as follows:
First, try to find a tree decomposition for $G$ of width
$r=\hat{c}\sqrt \ell$ for some constant $\hat{c}>0$.
One can choose $\hat{c}$ such that, if
 the
algorithm
fails, then there is no solution of size at most $\ell$. Otherwise,
in a second step, use the
tree decomposition obtained to solve the problem
by well known algorithms
in $O(n c^r)=O(n c^{\hat{c}\sqrt{\ell}})$ time for some constant $c$.
Thus, it is very important to support the first step
also in such a time. This, even for planar graphs,
was not possible by previously known algorithms.

{\bf{Recent results.}}
In the following overview over related results, 
$n$ denotes
the number
of vertices and $k$
the treewidth of the graph under
consideration.
{Tree decomposition and treewidth were 
introduced by Robertson and Seymour \cite{RobS84}, which
also presented the first algorithm for the computation of
the treewidth and a tree decomposition with a running time
polynomial in $n$ and exponential in $k$~\cite{RobS86}.
There are numerous papers with improved running times, 
as, e.g., \cite{ArnCP87,BodK91,LagA91,RobS95}.}
Here we focus on algorithms with running times 
either being polynomial in both $k$
and $n$, or being subquadratic in $n$. 
Bodlaender~\cite{Bod96} has shown that
a tree decomposition
can be found
in a time linear in $n$ and exponential in $k$.
However,
the running time of Bodlaender's algorithm is
practically infeasible already for very small $k$.
The algorithm
achieving the so far smallest
approximation ratio of the treewidth among the algorithms
with a running time
polynomial in $n$ and $k$ is the algorithm of
Feige, Hajiaghayi, and Lee \cite{FeiHL08}.
It constructs a tree decomposition of width 
$O(k\sqrt{\log k})$ thereby improving the bound
$O(k\log k)$ of Amir~\cite{Ami10}. 
In particular, no algorithms with constant approximation
ratios are known that are polynomial in the number of vertices and in the
treewidth.
One of the so far most efficient practical algorithms with constant approximation
ratio was presented by Reed
in 1992 \cite{Ree92}. 
His algorithm computes
a tree decomposition of 
width $3k+2$
in
$O(f(k) \cdot n \log n)$ time
for some exponential function $f$.
More precisely, this width is obtained after slight
modifications as observed by Bodlaender~\cite{Bod93}.

Better algorithms are known for the special case
of planar graphs.
Seymour and Thomas~\cite{SeyT94}
showed that
the so-called branchwidth $\mathrm{bw}(G)$
and
a so-called branch decomposition
of width $\mathrm{bw}(G)$ for a planar
graph $G$ can be computed in
$O(n^2)$ and $O(n^4)$ time, respectively.
A minimum branch decomposition of a
graph $G$ can
be used directly---like a tree decomposition---to
support efficient algorithms.
For each graph $G$, its branchwidth $\mathrm{bw}(G)$
is closely related to its treewidth $\mathrm{tw}(G)$; in detail,
$\mathrm{bw}(G) \leq \mathrm{tw}(G)+1 \leq \max( 3/2\,\mathrm{bw}(G),2)$ \cite{RobS91}.
Gu and Tamaki \cite{GuT08} improved the running time to $O(n^3)$
for constructing a branch decomposition and thus for finding a tree decomposition of
width $O(\mathrm{tw}(G))$.
They also showed that one can compute a tree decomposition of
width $(1.5+c)\mathrm{tw}(G)$ for a planar graph $G$ in $O(n^{(c+1)/c}\log
n)$ time for each $c\ge 1$ \cite{GuT11}.
Recently, Gu and Xu~\cite{GuX14} presented an algorithm to compute a constant factor
approximation of the treewidth in time $O(n\log^4 n\log k)$.
It is open whether deciding $tw(G) \le k$
is NP-complete or polynomial time solvable for planar graphs~$G$.

{\bf{Our results.}} 
In this paper, a weighted
graph $(G,c)$ is a graph $G=(V,E)$ with a {\cost} function
$c:V \rightarrow \Nat$ ($\Nat=\{1,2,3,\ldots\}$).
Moreover, $c_{\mathrm{max}}$ always denotes the maximum weight over all vertices.
In contrast to our conference version~\cite{KamT12}, we now consider the problem of finding a 
tree decomposition on weighted planar graphs. 
Roughly speaking, the weighted treewidth of a weighted graph is defined analogously as 
the (unweighted) treewidth, but instead of counting the vertices of a bag,
we sum up the weight of the vertices in the bag.
All results shown in this
paper for weighted graphs and weighted treewidth can be applied to
unweighted graphs and unweighted treewidth by 
setting $c(v)=1$ for all vertices $v$ of $G$. 
In Section~\ref{sec:bod}, we compute tree decompositions for 
$\ell$-outerplanar graphs with an algorithm different to Bodlaender's
algorithm~\cite{Bod98}, but
we obtain the same time bound and the same treewidth.
Our algorithm reduces $\ell$-outerplanar graphs to 
weighted $\ell$-outerplanar graphs, where the unweighted version is
$1$-outerplanar.
Another application of weighted treewidth is that
it allows us to
triangulate graphs with only a little increase in the weighted treewidth, which
improves our approximation ratio by about a factor of~$4$.

Interestingly, the generalization from unweighted to weighted graphs is possible
without increasing the asymptotic running time.
Moreover, 
we slightly modify our algorithm in such a way that we can
afterwards bound 
the number of so-called 
$({\mathcal S},\varphi)$-components, which improves the running time by
a factor $k$ compared to the running time shown in our conference version.

Given a weighted planar graph $(G,c)$
with $n$ vertices and
weighted tree\-width~$k$,
our algorithm
computes a tree decomposition for $G$
in time $O(n \cdot k^2 c_{\mathrm{max}} \log k)$
such that the vertices in each bag have a total {\cost} of $O(k)$. 
This means that, for unweighted graphs with $n$ vertices and
treewidth $k$, we obtain a tree
decomposition of width $O(k)$ in $O(n \cdot k^2\log k)$ time, which is
a better running time than that of Gu and Xu~\cite{GuX14} for
all $n$-vertex planar graphs of treewidth $k=O(\log n)$. We do
not focus on graphs with a
larger treewidth since for such graphs 
it is not clear whether
the second step mentioned for solving NP-hard problems
can be solved efficiently.

Our result can be used to find a
solution of size $\ell\in\Nat$ 
for many bidimensional graph problems on planar graphs
that are
expressible in linear EMSO 
in a time linear in $n$ and subexponential in $\ell$.
Such problems are, e.g.,
  {\textsc{Minimum Dominating Set}},
  {\textsc{Minimum Maximal Matching}}, and
  {\textsc{Minimum Vertex Cover}},
which all are NP-hard on planar graphs.

In contrast to general graphs, on planar graphs many graph parameters
as branchwidth and rank-width differ only by a
constant factor from the tree\-width~\cite{RobS91,FomSD10,Oum08}. 
Thus, 
our algorithm
also computes
a
constant factor approximation for these parameters on $n$-vertex planar graphs
in a time linear in~$n$.

\section{Main Ideas}\label{sec:ide}
\newcommand{\CC}{C}
\newcommand{\DD}{D}
Before we can describe our ideas, we precisely define 
tree decompositions and (weighted)
treewidth.

\begin{definition}[tree decomposition, bag, width, {\cost} of a bag]\label{def:treeD}
\hspace{-0.3mm}A {\em tree de\-compo\-sition} for an
unweighted graph $G=(V,E)$ or
for a weighted graph $G=(V,E)$ with a {\cost} function
$c:V\rightarrow \Nat$ is a pair $(T,B)$,
where %
$T=(W,F)$ is a tree and $B$ is a function that maps each
node $w$ of~$T$ to a subset %
of~$V$---called the {\em bag}
of~$w$---such that%
\begin{enumerate}\itemsep 2pt
\item each vertex of~$G$ is contained in a bag and each edge of~$G$ is
a subset of a bag,
\item for each vertex $v\in V$, the nodes whose bags contain $v$ induce a
subtree of~$T$.%
\end{enumerate}
In addition, the {\em unweighted width} of %
$(T,B)$---or short, the {\em width} of~$(T,B)$---is
$\max_{w\in W}\{|B(w)|-1\}$
and the {\em weighted width} is
$\max_{w\in W}\{c(B(w))-1\}$ with $c(B(w))=$
$\sum_{v \in B(w)}
c(v)$. The term $c(B(w))$ is also called the {\em {\cost}} of the bag of~$w$.
The {\em (weighted) treewidth} $\mathrm{tw}(G)$ of a graph $G$ is the smallest $k$ for which
$G$ has a tree decomposition of (weighted) width $k$.
\end{definition}

   For simplification, on weighted graphs the word treewidth 
   means weighted treewidth. However, in this section and 
   in Section~\ref{sec:alg} summarizing our main results we often refer explicitly 
   to weighted or unweighted treewidth.

We next describe our ideas for the
construction
 of a tree decomposition
of small unweighted treewidth and subsequently
generalize it to weighted tree\-width. 
In the case of unweighted plane graphs 
it is very useful to model
vertices as points in a landscape 
where
we assign a {\em height} to each vertex $v$.
For the time being,
the height
can be assumed to be the
length of a shortest path from $v$ to a vertex incident to the outer face
for some given planar embedding $\varphi$
where the length of a path is the number of its vertices---a more precise
definition is given in the next paragraph.
In particular, this is of interest 
for a graph $G$ if we can bound the height of all vertices of~$G$ by $O(tw(G))$
since, as part of our 
computation of a
tree decomposition, we split $G$ in some kind similar
to cutting a round cake into slices.
More exactly,
we use paths starting in a vertex $v^*$ of largest height and
following vertices with decreasing height until reaching a
vertex adjacent to the outer face. Technically, we realize
the splitting by putting the vertices of such a path into
one bag. However, we find such a tree decomposition of width $k$ only if the height 
of~$v^*$ is at most $k+1$ since, otherwise, the vertices of such a path can
not be all part of one bag.
In
the case of weighted plane graphs we have the problem that a large
{\cost} of a vertex increases the {\cost} of a bag so much that we have to
reduce the number of additional vertices that can
be put together with this vertex in one bag of a tree decomposition.
To translate the {\cost} of a vertex into our landscape model,
we consider a vertex not to be a single point %
with a
single height in the landscape, but as a {\em cliff}
leading from a lower height to an upper height.
Thus, instead of a single height we assign a {\em height interval} to each
vertex whose length can be considered as the length of the cliff and
is the {\cost} of the vertex minus one.

A weighted planar graph $(G,c)$ with an embedding $\varphi$ is called a {\em weighted
plane graph} $(G,\varphi,c)$, which we now consider.
We now precisely define %
the height interval %
of each vertex $v$.
It is
referred
to as $h_\varphi(v)=[h^-_\varphi(v),h^+_\varphi(v)]$. 
This means
that
one end of the cliff assigned to vertex $v$ has height
$h^-_\varphi(v)$ and the other end has height $h^+_\varphi(v)$. We also call
$h^-_\varphi(v)$ the {\em lower height} of~$v$ and $h^+_\varphi(v)$ the {\em
upper height} of~$v$. If $i\in h_\varphi(v)$ for some $i\in \Nat$, 
we also say $v$ is a vertex of {\em height} $i$.
The set of all
vertices incident to the outer face is called the {\em coast}
(of the plane graph).
To define the lower and upper heights %
of the vertices, we initially define a function $\eta$ with $\eta(v)=c(v)$
for all vertices~$v$.
We now use the concept of a so-called peeling consisting of a
sequence of peeling steps.
A {\em peeling step}
decrements $\eta(v)$ by one for all vertices $v$ that are part of the coast and
subsequently removes all vertices with $\eta(v)=0$.
Let us number the peeling steps by $1,2,3,\ldots$.
After the removal of all vertices, we set $h_\varphi(v)=[i-c(v)+1,i]$ for all
vertices $v$ that are removed in the $i$th peeling step ($i\in \Nat$).
The height interval of a vertex consists exactly of the numbers of the peeling
steps the vertex is incident to the outer face 
including the peeling step that removed the vertex.
For an example %
see also Fig.~\ref{fig:outerpl}.
A weighted graph $G$ is called
{\em weighted $\ell$-outerplanar}
if there is an embedding $\varphi$ of~$G$ such that
all vertices have upper height at most $\ell$. In this case, $\varphi$ is also
called 
{\em weighted $\ell$-outerplanar}.

We also want to remark that, if we assign weight one to all vertices, the %
definitions %
of weighted treewidth %
and of 
weighted $\ell$-outerplanar
graphs
correspond
to %
usual unweighted treewidth and to %
usual unweighted
$\ell$-outerplanar graphs, respectively.

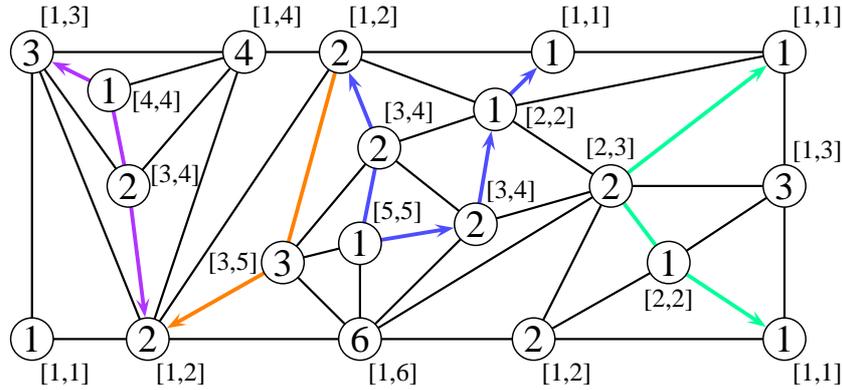
\begin{figure}[b!]
\vspace{1mm}
\begin{center}
\scalebox{1.45}{
\scalebox{1} 
{
\begin{pspicture}(0.1,-1.6625)(7.325,1.6625)
\definecolor{color24}{rgb}{0.0,1.0,0.6}
\definecolor{CadetBlue}{rgb}{0.7,0.2,1.0}
\definecolor{lightBlue}{rgb}{0.3,0.3,1.0}
\psline[linewidth=0.02cm](2.305,1.3125)(2.82,1.3125)
\psline[linewidth=0.02cm](3.18,1.3125)(4.745,1.3125)
\psline[linewidth=0.02cm](5.105,1.3125)(6.845,1.3125)
\psline[linewidth=0.02cm](7.025,1.1325)(7.025,0.2675)
\psline[linewidth=0.02cm](7.025,-0.0925)(7.025,-1.1325)
\psline[linewidth=0.02cm](6.845,-1.3125)(4.93,-1.3125)
\psline[linewidth=0.02cm](4.57,-1.3125)(3.355,-1.3125)
\psline[linewidth=0.02cm](2.995,-1.3125)(1.43,-1.3125)
\psline[linewidth=0.02cm](1.3069,-1.1417)(2.0681,1.1417)
\psline[linewidth=0.02cm](1.07,-1.3125)(0.38,-1.3125)
\psline[linewidth=0.02cm](0.2,-1.1325)(0.2,1.1325)
\psline[linewidth=0.02cm](0.38,1.3125)(1.945,1.3125)
\psline[linewidth=0.02cm](2.9002,1.1627)(1.3498,-1.1627)
\psline[linewidth=0.02cm](3.1685,1.2493)(4.2315,0.8507)
\psline[linewidth=0.02cm](3.3023,-1.1852)(4.0977,-0.3898)
\psline[linewidth=0.02cm](2.5902,-0.4742)(3.2348,0.2992)
\psline[linewidth=0.02cm](3.4906,0.3251)(4.0844,-0.1501)
\psline[linewidth=0.035cm,linecolor=lightBlue,arrowsize=0.05291667cm 2.0,arrowlength=1.4,arrowinset=0.4]{->}(3.2831,0.6046)(3.0743,1.1268)
\psline[linewidth=0.035cm,linecolor=orange,arrowsize=0.05291667cm 2.0,arrowlength=1.4,arrowinset=0.4]{->}(2.3187,-0.7018)(1.4236,-1.2133)
\psline[linewidth=0.02cm](2.6023,-0.7398)(3.0477,-1.1852)
\psline[linewidth=0.02cm](1.1831,-1.1454)(0.2669,1.1454)
\psline[linewidth=0.02cm](0.3046,1.166)(0.9704,0.234)
\psline[linewidth=0.035cm,linecolor=CadetBlue,arrowsize=0.05291667cm 2.0,arrowlength=1.4,arrowinset=0.4]{->}(1.0973,-0.0911)(1.2252,-1.114)
\psline[linewidth=0.02cm](6.8485,1.2772)(4.5765,0.8228)
\psline[linewidth=0.035cm,linecolor=color24](5.558,-0.0565)(5.867,-0.4685)
\psline[linewidth=0.035cm,linecolor=color24,arrowsize=0.05291667cm 2.0,arrowlength=1.4,arrowinset=0.4]{->}(5.5921,0.198)(6.8671,1.1897)
\psline[linewidth=0.02cm](5.3695,-0.0735)(4.8305,-1.1515)
\psline[linewidth=0.02cm](6.1248,-0.5127)(6.8752,-0.0123)
\psline[linewidth=0.035cm,linecolor=color24,arrowsize=0.05291667cm 2.0,arrowlength=1.4,arrowinset=0.4]{->}(6.1248,-0.7123)(6.8586,-1.2016)
\psline[linewidth=0.02cm](2.6496,-0.5688)(3.0004,-0.4812)
\psline[linewidth=0.035cm,linecolor=lightBlue,arrowsize=0.05291667cm 2.0,arrowlength=1.4,arrowinset=0.4]{->}(3.3526,-0.4079)(4.0277,-0.2954)
\psline[linewidth=0.02cm](3.3283,-1.2182)(5.2967,-0.0068)
\psline[linewidth=0.035cm,linecolor=lightBlue,arrowsize=0.05291667cm 2.0,arrowlength=1.4,arrowinset=0.4]{->}(4.2546,-0.0849)(4.3671,0.5902)
\psline[linewidth=0.02cm](1.1921,0.2242)(2.0079,1.1758)
\psline[linewidth=0.035cm,linecolor=orange](2.5224,-0.4388)(2.9526,1.1388)
\psline[linewidth=0.035cm,linecolor=lightBlue,arrowsize=0.05291667cm 2.0,arrowlength=1.4,arrowinset=0.4]{->}(4.5273,0.9148)(4.7836,1.1711)
\psline[linewidth=0.02cm](5.3002,0.1873)(4.5498,0.6877)
\psline[linewidth=0.02cm](1.0731,1.0119)(1.9519,1.2631)
\psline[linewidth=0.035cm,linecolor=CadetBlue](0.9353,0.786)(1.0397,0.264)
\psline[linewidth=0.035cm,linecolor=CadetBlue,arrowsize=0.05291667cm 2.0,arrowlength=1.4,arrowinset=0.4]{->}(0.739,1.043)(0.3789,1.2231)
\psline[linewidth=0.02cm](3.175,-0.6175)(3.175,-1.1325)
\psline[linewidth=0.035cm,linecolor=lightBlue](3.2103,-0.261)(3.3147,0.261)
\psline[linewidth=0.02cm](3.5208,0.4944)(4.2292,0.7306)
\psline[linewidth=0.02cm](4.3981,-0.2131)(5.2769,0.0381)
\psline[linewidth=0.02cm](5.63,0.0875)(6.845,0.0875)
\psline[linewidth=0.02cm](4.9063,-1.2232)(5.8187,-0.7018)
\pscircle[linewidth=0.0135,dimen=outer,fillstyle=solid](2.125,1.3125){0.2}
\usefont{T1}{ptm}{m}{n}
\rput(2.125,1.3125){4}
\rput(2.125,1.3125){\rput(0.3,0.32){\scalebox{0.75}{\small[1,4]}}}
\pscircle[linewidth=0.0135,dimen=outer,fillstyle=solid](3.0,1.3125){0.2}
\usefont{T1}{ptm}{m}{n}
\rput(3.0,1.3125){2}
\rput(3.0,1.3125){\rput(0.3,0.32){\scalebox{0.75}{\small[1,2]}}}
\pscircle[linewidth=0.0135,dimen=outer,fillstyle=solid](4.925,1.3125){0.2}
\usefont{T1}{ptm}{m}{n}
\rput(4.925,1.3125){1}
\rput(4.925,1.3125){\rput(0.3,0.32){\scalebox{0.75}{\small[1,1]}}}
\pscircle[linewidth=0.0135,dimen=outer,fillstyle=solid](7.025,1.3125){0.2}
\usefont{T1}{ptm}{m}{n}
\rput(7.025,1.3125){1}
\rput(7.025,1.3125){\rput(0.3,0.32){\scalebox{0.75}{\small[1,1]}}}
\pscircle[linewidth=0.0135,dimen=outer,fillstyle=solid](7.025,0.0875){0.2}
\usefont{T1}{ptm}{m}{n}
\rput(7.025,0.0875){3}
\rput(7.025,0.0875){\rput(0.3,0.32){\scalebox{0.75}{\small[1,3]}}}
\pscircle[linewidth=0.0135,dimen=outer,fillstyle=solid](7.025,-1.3125){0.2}
\usefont{T1}{ptm}{m}{n}
\rput(7.025,-1.3125){1}
\rput(7.025,-1.3125){\rput(0.3,-0.32){\scalebox{0.75}{\small[1,1]}}}
\pscircle[linewidth=0.0135,dimen=outer,fillstyle=solid](4.75,-1.3125){0.2}
\usefont{T1}{ptm}{m}{n}
\rput(4.75,-1.3125){2}
\rput(4.75,-1.3125){\rput(0.3,-0.32){\scalebox{0.75}{\small[1,2]}}}
\pscircle[linewidth=0.0135,dimen=outer,fillstyle=solid](3.175,-1.3125){0.2}
\usefont{T1}{ptm}{m}{n}
\rput(3.175,-1.3125){6}
\rput(3.175,-1.3125){\rput(0.3,-0.32){\scalebox{0.75}{\small[1,6]}}}
\pscircle[linewidth=0.0135,dimen=outer,fillstyle=solid](1.25,-1.3125){0.2}
\usefont{T1}{ptm}{m}{n}
\rput(1.25,-1.3125){2}
\rput(1.25,-1.3125){\rput(0.3,-0.32){\scalebox{0.75}{\small[1,2]}}}
\pscircle[linewidth=0.0135,dimen=outer,fillstyle=solid](4.4,0.7875){0.2}
\usefont{T1}{ptm}{m}{n}
\rput(4.4,0.7875){1}
\rput(4.4,0.7875){\rput(0.5,-0.075){\scalebox{0.75}{\small[2,2]}}}
\pscircle[linewidth=0.0135,dimen=outer,fillstyle=solid](0.2,-1.3125){0.2}
\usefont{T1}{ptm}{m}{n}
\rput(0.2,-1.3125){1}
\rput(0.2,-1.3125){\rput(0.3,-0.32){\scalebox{0.75}{\small[1,1]}}}
\pscircle[linewidth=0.0135,dimen=outer,fillstyle=solid](0.2,1.3125){0.2}
\usefont{T1}{ptm}{m}{n}
\rput(0.2,1.3125){3}
\rput(0.2,1.3125){\rput(0.3,0.32){\scalebox{0.75}{\small[1,3]}}}
\pscircle[linewidth=0.0135,dimen=outer,fillstyle=solid](4.225,-0.2625){0.2}
\usefont{T1}{ptm}{m}{n}
\rput(4.225,-0.2625){2}
\rput(4.225,-0.2625){\rput(0.32,0.3){\scalebox{0.75}{\small[3,4]}}}
\pscircle[linewidth=0.0135,dimen=outer,fillstyle=solid](2.475,-0.6125){0.2}
\usefont{T1}{ptm}{m}{n}
\rput(2.475,-0.6125){3}
\rput(2.475,-0.6125){\rput(-0.45,0.0){\scalebox{0.75}{\small[3,5]}}}
\pscircle[linewidth=0.0135,dimen=outer,fillstyle=solid](3.35,0.4375){0.2}
\usefont{T1}{ptm}{m}{n}
\rput(3.35,0.4375){2}
\rput(3.35,0.4375){\rput(0.27,0.325){\scalebox{0.75}{\small[3,4]}}}
\pscircle[linewidth=0.0135,dimen=outer,fillstyle=solid](1.075,0.0875){0.2}
\usefont{T1}{ptm}{m}{n}
\rput(1.075,0.0875){2}
\rput(1.075,0.0875){\rput(0.425,0.1){\scalebox{0.75}{\small[3,4]}}}
\pscircle[linewidth=0.0135,dimen=outer,fillstyle=solid](5.45,0.0875){0.2}
\usefont{T1}{ptm}{m}{n}
\rput(5.45,0.0875){2}
\rput(5.45,0.0875){\rput(0.0,0.342){\scalebox{0.75}{\small[2,3]}}}
\pscircle[linewidth=0.0135,dimen=outer,fillstyle=solid](5.975,-0.6125){0.2}
\usefont{T1}{ptm}{m}{n}
\rput(5.975,-0.6125){1}
\rput(5.975,-0.6125){\rput(0.0,-0.34){\scalebox{0.75}{\small[2,2]}}}
\pscircle[linewidth=0.0135,dimen=outer,fillstyle=solid](3.175,-0.4375){0.2}
\usefont{T1}{ptm}{m}{n}
\rput(3.175,-0.4375){1}
\rput(3.175,-0.4375){\rput(0.34,0.29){\scalebox{0.75}{\small[5,5]}}}
\pscircle[linewidth=0.0135,dimen=outer,fillstyle=solid](0.9,0.9625){0.2}
\usefont{T1}{ptm}{m}{n}
\rput(0.9,0.9625){1}
\rput(0.9,0.9625){\rput(0.432,-0.1){\scalebox{0.75}{\small[4,4]}}}
\end{pspicture}
}
         }%
\end{center}
\vspace{1mm}
\caption{A weighted plane graph $G$ with its vertices labeled by their weights and
the resulting height intervals written beside the vertices. Arrows indicate
a neighbor with smallest upper height. The thick edges define so-called
perfect crest separators defined in Section~\ref{sec:goodmount}. In our example
where $G$
is assumed to be a cake, the perfect crest separators are the paths used
to cut the cake into slices.}
\label{fig:outerpl}
\end{figure}

Observe that each vertex of lower height $q \ge 2$ is incident to a face
with a vertex of upper
height $q-1$  
and that the upper height of a vertex $v$ is 
the total weight of a shortest weighted path from $v$ to the coast.
For technical reasons and
to simplify our definitions, our observations and our lemmas, 
in the rest of
the paper
we 
usually consider only {\em almost triangulated graphs}, i.e.,
plane graphs in which the boundary of each inner face 
consists of exactly three vertices and edges.  
As a consequence, each vertex of lower height $q \ge 2$
is adjacent to a vertex of upper height $q-1$.
If a weighted plane graph   
$(H,\psi,c)$ of treewidth $k-1$ and maximal vertex weight
$c_{\mathrm{max}}$ is not almost triangulated, we can multiply each weight of a
vertex by $x\in \Nat$ to obtain a weighted plane graph
$(H',\psi,xc)$ of treewidth $xk-1$. Afterwards, 
it can be %
triangulated 
by simply adding a new vertex of weight $1$ into each inner face
and by connecting this vertex by edges 
with all vertices on the boundary of that
inner face. 
Let $(H'',\psi'',c'')$ be the graph obtained.
Theorem 2 in \cite{KloLL02} shows that %
a tree decomposition $(T',B')$ for $H'$ %
can be turned into a tree decomposition for  
$H''$ by adding at most $3k-2$ of the new vertices into each bag of~$(T,B)$.
Thus, we have a tree decomposition for $H''$ where 
the weight of every bag is bounded by
$xk+3k-2$. 
For some $\alpha,\beta \in \Nat$, assume that we can compute a tree decomposition
$(T'',B'')$ for $H''$ of width 
$\alpha \cdot \mathrm{tw}(H'')+\beta c'_{\mathrm{max}}-1$ where $c'_{\mathrm{max}}$ is the
maximal weight of a vertex in $H''$.
Then the size of the bags of~$(T'',B'')$ is bounded by 
 $\alpha(xk+3k-2)+\beta c'_{\mathrm{max}}$.
Removing the new vertices from $(T'',B'')$,
we get a tree decomposition for $H'$. %
If we finally take this tree decomposition as a tree decomposition
for $H$, which has the vertices of smaller weight, the weight  of
every bag is bounded by $\lfloor (\alpha (xk+3k-2)+\beta x c_{\mathrm{max}})/x \rfloor=\lfloor
\alpha k+\alpha(3k-2)/x
+\beta c_{\mathrm{max}}\rfloor$. 
This means that we can compute a tree decomposition for $H$ of width 
$(\alpha+\epsilon)k + \beta c_{\mathrm{max}} + O(1)$ if we choose $x$ large
enough.

To describe our ideas, we need some more %
definitions.
For a subgraph $G'$ of a %
weighted plane
graph $(G,\varphi,c)$,
we use $\varphi|_{G'}$ to denote
the embedding of~$G$ restricted to the vertices and edges of~$G'$.
For a graph $G=(V,E)$ and a vertex set $V'\subseteq V$,
we let $G[V']$ be the subgraph of~$G$ induced by the vertices
of~$V'$; and we define $G-V'$ to be the graph
$G[V\setminus V']$. If a graph $G$ is a subgraph of another
graph $G'$, we write $G\subseteq G'$.
Through the whole paper, path and cycles are {\em simple}, i.e., no vertex and
no edge
appears more than once in it.
For a graph $G=(V,E)$, we also say that
a vertex set $S\subseteq V$ {\em disconnects} two vertex sets $A,B\subseteq V$
{\em weakly} if no connected component of~$G-S$ contains vertices of both $A$
and $B$. $S$ {\em disconnects} $A$ and $B$ {\em strongly} if
additionally $S\cap(A\cup B)=\emptyset$ holds. If a vertex set $S$
strongly disconnects two non-empty vertex sets, we say that $S$ is a
{\em separator} (for these vertex sets). A special kind of separators
being of great significance for our paper is defined in the following
definition. 

\begin{definition}[coast separator]
A set $Y$ that strongly disconnects a vertex set
$U$ from the %
coast is called %
a {\em coast separator} (for~$U$).
\end{definition} 

As a consequence of the definition above, the vertices of a coast separator are disjoint from
the coast.
Finally, we define the {\em
weighted size of a separator}
and the {\em weighted length of a path or cycle} as
the sum over the {\cost}s of its vertices.

 As observed by Bodlaender~\cite{Bod98},
one can easily construct
a tree decomposition of width $3\ell-1$ for
an $\ell$-outerplanar unweighted graph $G=(V,E)$ in $O(\ell|V|)$ time.
In Section~\ref{sec:bod}, we show that
his algorithm can be extended to
weighted $\ell$-outerplanar
graphs.
One idea
to find a tree decomposition of weighted width $O(k)$ for
an $\omega(k)$-weighted-outer\-pla\-nar graph %
of weighted treewidth $k$
is to search for
a coast separator~$Y$ of weighted size $O(k)$ that disconnects the
vertices of large %
lower
height strongly from the coast by
applying Theorem~\ref{the:Sep} below.
For proving the theorem we use the
following two well-known observations, which follow from the
definition of a tree decomposition.

\begin{observation}\label{obs:weakSep}
Let $(T,B)$ be a tree decomposition for 
a weighted graph $(G,c)$, and let $W$ and $F$ be the set of nodes and arcs,
respectively, of~$T$.
Take
$\{w',w''\}\in F$. For each pair of subtrees %
$(W_1,F_1)$ and $(W_2,F_2)$ part of different trees in
the forest $(W,F \setminus \{w',w''\})$, $B(w')\cap B(w'')$ weakly
disconnects
$\bigcup_{w \in W_1}B(w)$
and
$\bigcup_{w \in W_2}B(w)$.
\end{observation}

\begin{observation}\label{obs:connected}
Let $(T,B)$ and $(G,c)$ be defined as in Obs.~\ref{obs:weakSep},
and let $V'$ be a subset of the vertices of~$G$ such that
$G[V']$ is connected. Then the nodes of~$T$ whose bags contain
at least one
vertex of~$V'$ induce a connected subtree of~$T$.
\end{observation}

\begin{theorem}\label{the:Sep}
Let $(G,\varphi,c)$ be a weighted plane graph of
weighted %
tree\-width $k\in \Nat$.
Moreover,
let $V_1$ and $V_2$ be connected sets of vertices of~$G$
such that
$(\min_{v\in V_2}h^-_\varphi(v))-(\max_{v\in V_1}h^+_\varphi(v)) \geq k+1$.
Then, there exists a
set $Y$
  of weighted size at most $k$ that strongly disconnects
 $V_1$ and~$V_2$.
 \end{theorem}

\begin{proof}
We exclude the case $k=1$ since it is well known that graphs of
tree\-width~$1$ are forests, i.e., 
$h^-_\varphi(v)=1$
holds for all vertices $v$.
The same is true for graphs of weighted treewidth $1$.
Consequently, no sets $V_1$ and $V_2$ with the properties described in the
theorem exist in the case of forests. %

Let $(T,B)$
be a tree decomposition of width $k$
with a smallest number of
bags containing both at least one
vertex of~$V_1$ and at least one
vertex of~$V_2$.
If there is no
such bag, then 
for each $i\in\{1,2\}$, the nodes of~$T$ whose bags contain at least
one vertex of~$V_i$  induce a subtree of~$T$ (Obs.~\ref{obs:connected}), and 
we thus can find
two closest nodes $w_1$ and $w_2$ in $T$ with
$V_1\cap B(w_1)\neq \emptyset \neq V_2\cap B(w_2)$.
For
the node $w'$ adjacent to $w_2$
on the $w_1$-$w_2$-connecting path in $T$,
the set $B(w')\cap B(w_2)$ %
is a separator of weighted size at most
$k$ for %
$V_1$ and $V_2$ (Obs.~\ref{obs:weakSep}).
Hence, let us assume that
there is at least one node $w$ in $T$ with
its bag containing
both a vertex $v_1\in V_1$ and a vertex $v_2 \in V_2$.
Since %
the {\cost} of~$B(w)$ is at most $k+1$,
for at least one number $i$
with $h^+_\varphi(v_1)< i< h^-_\varphi(v_2)$,
there is no vertex in $B(w)$ that has a height interval containing $i$.
In other words, no vertex in $B(w)$ is a vertex of height $i$.
Since
$V_2$ is connected and
$\min_{v\in V_2}h^-_\varphi(v))>i$,
there
is a connected set $Z$ that consists
exclusively  of vertices of
height $i$ with $Z$ disconnecting $V_2$ %
from all vertices
$u$ with $h_\varphi(u)^+<i$, i.e., in particular from
$V_1$.
In fact,
$Z$ can be chosen as
the set of vertices of the coast of the connected component that contains $V_2$
after $i-1$ peeling steps and is therefore connected. Thus,
the nodes of~$T$, whose bags contain at least one vertex of~$Z$, induce a subtree
$T'$ of~$T$ (Obs.~\ref{obs:connected}).
Therefore, it is possible to
root $T$ such that $w$ is a child of the root
and such that the subtree
$T_w$ of~$T$ rooted in $w$ does not contain any node of~$T'$.
We then replace $T_w$ by two copies $T_1$ and $T_2$ of~$T_w$
and similarly the
edge between the root $r$ of~$T$ and the root of~$T_w$
by two edges connecting $r$ with the root of~$T_1$ and $T_2$,
respectively.
For each node $w'$ in $T_w$ with copies $w'_1$ and $w'_2$ in $T_1$ and
$T_2$, respectively,
we
define the bag $B(w'_1)$
to
consist of those vertices of the
bag $B(w')$
that are also contained
in the connected component of~$G[V\setminus Z]$ containing
$V_2$, and $B(w'_2)$
should contain the remaining vertices
of~$B(w')$.
Since there are no edges between the vertices of the
connected component of~$G[V\setminus Z]$ containing $V_2$ and
the vertices of other connected components of
$G[V\setminus Z]$ and since the bags of~$T_w$ contain no vertex
of~$Z$,
for each edge, both of its endpoints still appear
in at least one bag
after the replacement described above.
To sum up, the
replacement
leads to a tree decomposition
of width $k$ with a lower number of bags containing both
a vertex of~$V_1$ and of~$V_2$.
Contradiction.
\end{proof}

Assume that we are given a connected weighted plane graph
$(G,\varphi,c)$ with $G=(V,E)$
and weighted treewidth $k\in\Nat$ that contains exactly one so-called crest. 
\begin{definition}[crest, upper/lower height]
For a weighted plane graph,
a maximal connected set $H$ of vertices
of the same upper height is a
{\em crest} if no vertex of~$H$ is connected to a
vertex of larger upper height.
The {\em (upper) height} of a crest is the upper height of its vertices,
and the {\em lower height} is the minimal lower height among the lower heights of its
vertices.
\end{definition}
W.l.o.g.,
$|V|>1$. Thus, every vertex is incident to an edge, which implies that 
$c_{\mathrm{max}}\le k$, i.e.,
$c(v)\le k$
for all vertices $v$
since each vertex must be contained with another vertex
of {\cost} at least~1 in a common bag of total {\cost} at most~$c(v)+1$.
We can try to
construct a tree decomposition for $G$
as follows:

We will construct a series of weighted subgraphs $(G',c')$ of~$(G,c)$,
where
the {\cost} function $c'$ %
should be implicitly defined by
the restriction of
$c$ to the vertices of the subgraph.
Initialize $G'_1=(V'_{1},E'_{1})$ with $G$.  For $i=1,2,\ldots$, as long as
$G'_i$ %
has
only one crest and has
vertices of lower height at least 
$2k+1$
(which are connected since $G'_i$ has only one crest)
apply
Theorem~\ref{the:Sep} to obtain a separator $Y_i$ of weighted size at most $k$ separating
the vertices of lower height at least $2k+1$
from %
all vertices of upper height at most $k$.
This means that we
separate the vertices of large lower height from all vertices of the coast since a vertex $v$
of the coast
can
have an upper height of at most $c(v)\le k$. Thus, $Y_i$ is a coast separator.
Then, define $G'_{i+1}=(V'_{i+1},E'_{i+1})$
as the subgraph of~$G'_i$ induced by the vertices of~$Y_i$ and
of the connected component of~$G'_i\setminus Y_i$ that contains
the crest of~$G'_i$.
Moreover, let
$G_i=G'_i[Y_i\cup(V'_i\setminus V'_{i+1})]$.
If the recursion
stops with a %
weighted $O(k)$-outerplanar graph $G'_j$ ($j\in \Nat$),
we set
$G_{j}=G'_j$
and
construct a tree decomposition for $G$ as follows:
First, compute a tree decomposition $(T_i,B_i)$ of weighted width $O(k)$
for each $G_i$
($i \in \{1,\ldots,j\}$). This is possible
since $G_i$ is weighted $O(k)$-outerplanar.
Second, set
$Y_0=Y_{j}=\emptyset$. Then, for all $i\in\{1,\ldots,j\}$,
add the vertices of $Y_i \cup Y_{i-1}$
to all bags of~$(T_i,B_i)$.
Finally, for %
all $i\in\{1,\ldots,j-1\}$,
connect an arbitrary node of~$T_i$ with an arbitrary node
of~$T_{i+1}$.
This leads to a tree decomposition for $G$ of weighted width $O(k)$.

If we are given a weighted plane  graph $(G,\varphi,c)$ with
weighted treewidth $k\in\Nat$
that has more than
one crest
of lower height at least $2k+1$,
(or if this is true for one of the subgraphs $G'_i$ defined above)
we cannot
apply Theorem~\ref{the:Sep} to find one coast separator separating
simultaneously all vertices of lower height
at least $2k+1$ from the coast
since these vertices
may not be
connected. For cutting off the vertices
of large height, one might use
several coast separators; one
for each connected
component induced by the vertices of lower height at least $2k+1$.
However, if we insert the vertices of all coast separators
into every bag of a tree decomposition for the %
graph
with the vertices of small lower
height,
this may increase the width of the
tree decomposition by more than a constant factor since 
there may be more than a constant number of coast separators.
{This is the reason why}, for some suitable linear function
$f:\Nat \rightarrow \Nat$ and some constant $q\in \Nat$, we search
for further
separators called {\em perfect crest separators} that
are disjoint from the crests and that
partition
our graph~$G$ (or~$G'_i$) into %
smaller
subgraphs---called {\em components}---such
that, %
for each component~$C$ containing a non-empty set $V'$ of vertices of
lower height %
at least $f(k)$, there is a set $Y_{\CC}$ {of vertices} with
the following properties:\label{page:ideas}%
\begin{itemize}
\item[(P1)] $Y_{\CC}$
is a coast separator for $V'$ of weighted size $O(k)$.
\item[(P2)] $Y_{\CC}$ is contained in ${\CC}$.
\item[(P3)] $Y_{\CC}$ is disjoint to the set of vertices with lower
height $\le q$.
\end{itemize}

The main idea is that---in some kind similar to the construction
above---we want to construct a tree decomposition separately for
each component and afterwards to combine these tree decompositions
to a tree decomposition of the whole graph. The
properties above should guarantee that,
for each tree decomposition computed for a component,
we
have to add the vertices of at most one coast separator
{into %
the bags of that tree decomposition}.
We next {try to guarantee}
(P1)-(P3).

By making
the components so small that each component has at most one maximal connected
set of vertices of height at least 
$f(k)$,
we can easily
find a set of coast separators satisfying (P1) %
by using Theorem~\ref{the:Sep}.
We next try to find some constraints for the perfect crest separators
such that we can also guarantee (P2).
Recall that we assume that
our graph is almost triangulated.
Thus, 
the vertices
of a %
coast separator of minimal weighted size induce a unique
cycle.
 Suppose for a moment that it is possible
to choose
each perfect crest separator %
as the vertices of a path %
with the property that,
for each pair of consecutive vertices $u$ and $v$ with $u$
before $v$,
the upper height of~$v$ is one smaller than the lower height of
$u$, i.e., $h^+_\varphi(v)=h^-_\varphi(u)-1$.
In Section~\ref{sec:goodmount}, we call such a path a {\em down path} if it also has some
additional properties.
For a down path $P$, there is no path
that connects two vertices $u$ and $v$ of~$P$ with a
strictly shorter weighted length than the subpath of~$P$ from $u$ to $v$.
Consequently,
whenever we search
for a coast separator of smallest weighted size for a connected
set $V'$ {in a component $C$}, there is no need
to consider any coast separator with a subpath $Q$
{consisting of vertices that are strongly disconnected
from the vertices of~$C$ by the vertex set of~$P$.}
{Note that it is still possible that vertices
of a coast separator belong to
a perfect
crest
separator.
To guarantee that (P2) holds,
in a more precise definition of the
components, we let
the vertices and edges
of a perfect crest
separator %
belong to two components on
`both sides' of the perfect crest separator.
Then we can observe that,} if we
choose the perfect crest separators as down paths
not containing any vertex of any crest,
each maximal connected set of vertices
of lower height at least %
$f(k)$
in a component ${\CC}$
has a coast separator $Y_C$ that is completely contained
in ${\CC}$. %

Unfortunately, the vertex set of one %
down path can
not be a separator.
For two down paths $P_1$ and $P_2$ that start in
  two
adjacent vertices, we use $P_1\circ P_2$ to denote
the concatenation of the reverse path of~$P_1$, a path $P'$,
and the path $P_2$, where $P'$ is the path induced by the
edge connecting the
first vertices of~$P_1$ and $P_2$.
The idea is to define a
perfect crest separator
as the vertex set of such a concatenation $P_1\circ P_2$.
For more information on crest separators, see
Section~\ref{sec:goodmount}.
With our new definition of a perfect crest separator
we cannot avoid in general that a coast separator
$Y_C$ for the crest of a component $C$
crosses a perfect crest separator and uses vertices
outside $C$. Thus (P2) may be violated. To
handle this problem we
define a {\em minimal coast separator}
for a connected set $S$ in a graph $G$ to be a coast separator
$Y$ for $S$ that %
has minimal weighted size
such
that among all such coast separators
the subgraph of~$G$ induced by the vertices of~$Y$ and the vertices of the
connected component \mbox{of~$G\setminus Y$} containing $S$ has a minimal
number of inner faces.
Whenever a minimal coast separator for a connected
set $S$ in a component~$C$ crosses a perfect crest separator,
the part of the minimal coast separator outside $C$ forms
a so-called pseudo shortcut. Further details on pseudo shortcuts and their
computations are described in Section~\ref{sec:interact}.
Lemma~\ref{lem:enclosev} 
shows that,
if a pseudo shortcut $P$ is part of a
minimal
coast separator~$Y_{\CC}$ for a
connected set of vertices
in a component~${\CC}$ and if it passes %
through %
another component~$C'$, $Y_{\CC}$ also
separates all vertices of lower height at least $2k+1$
in $C'$
from the
coast. Then,
{we merge}
{$C$ and $C'$}
to one super component ${\CC}^*$.
After a similar merging for each pseudo shortcut passing through
another component,
both properties (P1) and (P2) hold. %
For guaranteeing property (P3), 
vertices of height $\le c$ and, in particular, the vertices of 
the coast play a special role for several definitions, e.g.,
for the pseudo          
shortcuts or the so-called lowpoints of a perfect crest separator.

Given a perfect crest separator $X=P_1\circ P_2$,
the idea is to %
find tree decompositions $(T_1,B_1)$
and $(T_2,B_2)$ for the two components of~$G$
on `either sides' of~$X$
such that %
$T_i$, for each
$i\in\{1,2\},$ has a node $w_i$  with $B_i(w_i)$ containing all
vertices of~$P_1$ and~$P_2$. By inserting an
additional edge $\{w_1,w_2\}$ we then obtain a tree decomposition
for the whole graph.
However, in general we are given a set $\mathcal{X}$ of
perfect crest separators that splits our graph into components
for which (P1)-(P3) holds.
If we want to construct a tree decomposition for
a component, we usually have to guarantee that, for each
perfect crest separator $X\in\mathcal{X}$, there is
a bag %
containing all vertices of~$X$. Since we can use
the techniques described above to cut off the vertices of
large
height
from each component, it remains
to find such a tree decomposition for the remaining
$O(k)$-weighted-outerplanar subgraph of the component.
Because of the simple structure of our perfect crest separators
we can indeed find such a tree decomposition by
extending
{Bodlaender's algorithm \cite{Bod98}} for
$O(k)$-outerplanar
graphs. For more details see %
Section~\ref{sec:bod}.
Finally, we can iteratively connect the tree decompositions
constructed for the several components in the same
way as it is described in case of only one perfect crest separator.
For an example, see also Fig.~\ref{fig:globalExample} and
\ref{fig:globalExampleTD}---the concepts of `top
vertex' and `ridge' are defined in the next section.
Our algorithm to compute a tree decomposition is presented
in Section~\ref{sec:alg}.

\definecolor{OliveGreen}{rgb}{0.0,0.6,0.0}
 \begin{figure}[hb!]
   \vspace{2mm}
   \centering
      \input{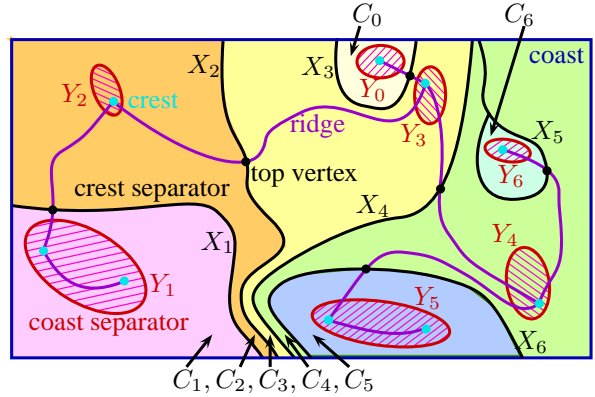}
      \caption{A weighted plane graph $(G,\varphi,c)$
           decomposed by a set ${\mathcal S}=\{X_1,\ldots,X_6\}$ of
     perfect crest separators into
       several components $C_0,\ldots,C_6$. In addition, each
       component $C_i$ has a coast separator $Y_i$ for all crests in $C_i$.}
      \label{fig:globalExample}
   \vspace{2mm}
   \end{figure}
\begin{figure}[ht!]
   \centering
      \scalebox{0.95}{\input{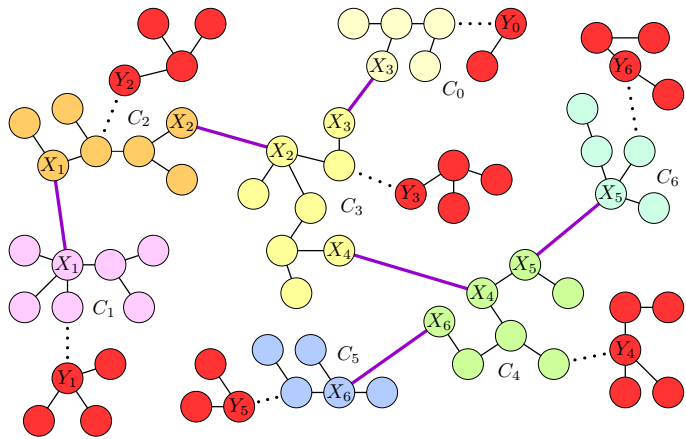}}
      \caption{A sketch of a tree decomposition for the graph $G$ of
      Fig~\ref{fig:globalExample}. In detail, each colored
       component 
       $C_i$ of
      Fig~\ref{fig:globalExample} has its own tree decomposition $(T_i,B_i)$ for
       the
       vertices of small height in $C_i$ 
       with the bags of $(T_i,B_i)$ being colored  
       with the same color than the component $C_i$.
      We implicitly assume
      that $Y_i$ is part of all bags of $(T_i,B_i)$.
      A tree decomposition for the vertices of large height in $C_i$ 
      is constructed recursively and contains a
      bag---marked with $Y_i$---containing all vertices of $Y_i$ 
      and being connected to one bag of $(T_i,B_i)$.
      For connecting the tree decompositions of different components,
      there is an edge
      connecting a node $w'$
      of $(T_i,B_i)$ with a node $w''$ of $(T_j,B_j)$  
      if and only if
      $C_i$ and $C_j$ share a certain edge (on their boundary) that 
      later is defined more precisely.
       In this case, the bags $w'$ and $w''$ contain the vertices of
       the perfect crest separators disconnecting $C_i$ and $C_j$---and
       are marked by the name of this crest separator.
       }
      \label{fig:globalExampleTD}
   \end{figure}

\section{Decomposition into Mountains}\label{sec:goodmount}

In the next three sections we let $(G,\varphi,c)$ be a weighted,
almost triangulated, and biconnected
graph with vertex set~$V$ and edge set~$E$.
If a weighted graph is not biconnected
one can easily
construct a tree decomposition for the whole graph by combining tree
decompositions for each biconnected subgraph.
Moreover,
let $n$ be the
number of vertices of~$G$ 
and take $\ell$ as the smallest number such that $\varphi$ is
weighted $\ell$-outerplanar.
Recall that, for each vertex $v$,
$h_\varphi(v)=[h^-_\varphi(v),h^+_\varphi(v)]$ is the height interval 
of~$v$ with $h^-_\varphi(v)$ and $h^+_\varphi(v)$ being called
the lower and upper height of~$v$.
Recall also that
a crest
is
a maximal connected set~$H$
of vertices of the same upper height
such that
no vertex in $H$
is connected to a vertex of larger upper height.
A weighted plane  graph with exactly one crest
is called a {\em mountain}. 
In this section we show a splitting of 
$(G,\varphi,c)$ into several mountains. Since 'certain' crests are
not of interest, we also show that, given a set of crests, a splitting of
$(G,\varphi,c)$ into several so-called components is possible such that each component contains
one crest part of the set.

As indicated in Section~\ref{sec:ide} our splitting process
makes use of so-called perfect crest separators and
down paths.
Since it is not so easy to compute
perfect crest separators,
we start to define and to consider crest separators more general
and later describe how to find perfect crest separators as
a subset of the crest separators.
First of all, we have to define down paths precisely.
Let us assume w.l.o.g.\ that the vertices
of all graphs considered in this paper are numbered with
pairwise different integers called the {\em vertex number}.
For each
vertex $u$ with a lower height $q\ge 2$, we define
the {\em down vertex} of~$u$ to be
the neighbor of~$u$ that among all neighbors
of~$u$ with upper height $q-1$ has the smallest vertex number.
We denote the down vertex of~$u$ by
$u\!\down$. The down edge of~$u$ is the
      edge $\{u,u\!\down\}$. The {\em down path} (of a vertex $v$) is a
path that (starts in $v$,) consists completely of down edges,
and ends in a vertex of the coast. 
In particular, a vertex $v$ of lower height 1 is
a down path
that only consists of itself.
Note that
every vertex has a down path.

\begin{definition}[crest separator, top edge, exterior/interior lowpoint]
\label{def:2} A {\em crest separator} in a weighted, almost triangulated,
and
biconnected graph is a tuple $X=(P_1,P_2)$ with $P_1$ being a down path starting
in some vertex $u$ and $P_2$ 
being a down path starting in a neighbor $v \neq u\!\down$ of~$u$ with
$h^+_{\varphi}(v)\le h^+_{\varphi}(u)$, where 
in the case $h^+_{\varphi}(v)= h^+_{\varphi}(u)$ the vertex number of $v$ is smaller than that of
$u$.
The edge $\{u,v\}$ is called {\em top edge} of~$X$.
The first vertex $v$
on $P_1$ that also is part of~$P_2$,
if it exists, is called
the  {\em lowpoint of}~$X$.
If $v$ belongs to the coast, we call $v$ an {\em exterior
lowpoint}, otherwise an {\em interior~lowpoint}.
\end{definition}

In the remainder
of this paper, 
we let
$\mathcal{S}(G,\varphi,c)$ be the set of all crest separators in
$(G,\varphi,c)$. 
Note that, for a crest separator $X=(P_1,P_2)$,
{the vertex set of}
$P_1 \circ P_2$ usually defines a
separator. This explains the name crest separator,
but formally a crest separator
is a tuple of paths.
Note also that a top edge is never a down edge and uniquely defines a
crest separator. Moreover, the top edges of two different crest
separators are always different.
Since an $n$-vertex planar graph has at most {$O(n)$}
edges, which can possibly be a top edge, and since, 
in a
weighted $\ell$-outerplanar
graph,
each crest separator consists of at most $2\ell$ vertices,
the next lemma holds.

\begin{lemma}\label{lem:allcs}
The set $\mathcal{S}(G,\varphi,c)$
can be constructed in
$O(\ell n)$ time.
\end{lemma}

Since crest separators are in the main focus of our paper, we use some additional terminology: Let $X=(P_1,P_2)$ be a crest separator. Then,
the
{\em top vertices} of
$X$
consist of the first vertex of~$P_1$ and the first vertex of
$P_2$. A top vertex of~$X$ is
called {\em highest} if its upper height is at least as large as the upper
height of the other top vertex of~$X$.
We write
$v\in X$ and say that {\em $v$ is a vertex of~$X$} to denote the
fact that $v$ is a vertex
 of~$P_1$ or $P_2$. The {\em border edges} of~$X$ are the
 edges of~$P_1\circ P_2$.
The {\em height} of~$X$ is the 
maximum upper height over all its vertices, which is the upper height
of the first vertex of~$P_1$.
The {\em weighted length} of~$X$ is the weighted length of 
$P_1\circ P_2$.
The {\em essential boundary of~$X$} is the
subgraph of~$G$ induced by all border edges of~$X$
that appear on exactly one of the two paths $P_1$ and $P_2$.
In particular, if $X$ has a lowpoint, the vertices of
the essential boundary consists exactly of the vertices appearing
before the lowpoint on $P_1$ or $P_2$ and of the lowpoint itself.
If $X$ has no lowpoint, the essential boundary is the
subgraph of~$G$ induced by the edges of~$P_1\circ P_2$.  
\begin{definition}[crest separator path]
For two vertices $s_1$ and $s_2$ being part of the essential boundary of
a crest separator $X$,
the {\em (long and short) crest-separator path} from~$s_1$
to~$s_2$
is the longest and shortest
path, respectively,
 from $s_1$ to $s_2$ that consists
 only of border edges of~$X$ and that does not
contain the lowpoint of~$X$ as an inner 
vertex.\end{definition} 

Note that, if neither $s_1$ nor $s_2$ is the lowpoint of 
$X$, the longest and the shortest
crest-separator path of~$X$ are the same.
If $s_1$ and $s_2$ are vertices of~$P_1 \circ P_2$,
{but not both are}
part of the essential boundary
of~$X$, the {\em
crest-separator path} from $s_1$ to $s_2$
is the
shortest path
from $s_1$ to $s_2$ consisting completely 
of edges of~$P_1\circ P_2$. 

We say that two paths $P'$ and $P''$ {\em cross},
if after merging the endpoints of the common edges of~$P'$ and
$P''$, there is a vertex $v$ with 
incident edges $e_1$ of~$P'$, $e_2$ of~$P''$, $e_3$ of~$P'$, $e_4$ of~$P''$
appearing clockwise in this order around $v$. The vertices that are merged
into $v$ are called the {\em crossing vertices of} $P'$ and $P''$.
Moreover, we also say that $P$ and $P'$ cross if
adding a new endpoint $v'$ to the outer face as well as
an edge $e$ from $v'$ to one endpoint of either $P$ or $P'$
with lower height 
$1$ 
together with an appropriate planar embedding of~$e$
in $\varphi$
makes the resulting paths
cross with respect to the definition of the previous sentence.
A crest separator 
$X=(P_1,P_2)$ and a path $P$ cross if $P_1\circ P_2$ and $P$ cross.

Each set
$\mathcal{S}$
of crest separators
splits $G$ into several subgraphs.
More precisely, for a set $\mathcal{S} \subseteq \mathcal{S}(G,\varphi,c)$,
let us define two inner faces $F$ and $F'$ of~$(G,\varphi,c)$
to be
{\em $(\mathcal{S},\varphi)$-connected}
 if there is a
list $(F_1,\ldots,F_j)$ ($j \in \Nat$) of inner faces of~$(G,\varphi,c)$
with $F_1=F$ and $F_j=F'$
such that, for each $i\in\{1,\ldots,j-1\}$, the
faces $F_i$ and $F_{i+1}$ share a common edge
not being a border edge of a crest
separator in $\mathcal{S}$.
A set $\mathcal{F}$
of inner faces of~$(G,\varphi,c)$ is
{\em $(\mathcal{S},\varphi)$-connected}
if each pair of faces in~$\mathcal{F}$
is $(\mathcal{S},\varphi)$-connected. Hence, a graph is
split by crest separators into the following
kind of subgraphs.

\begin{definition}[$(\mathcal{S},\varphi)$-component]
Let $\mathcal{S} \subseteq \mathcal{S}(G,\varphi,c)$.
For a maximal non\-empty $(\mathcal{S},\varphi)$-connected
set $\mathcal{F}$ of inner faces of~$(G,\varphi,c)$, the subgraph of~$G$ consisting of the
set of vertices and edges that are part of the boundary
of at least one face $F\in\mathcal{F}$ is called
an {\em $(\mathcal{S},\varphi)$-component}.
\end{definition}

By the fact that an $(\mathcal{S},\varphi)$-component consists of the
vertices on the boundary of an $(\mathcal{S},\varphi)$-connected
set of faces, we can
observe.

\begin{observation}\label{obs:biconnComp}
The $(\mathcal{S},\varphi)$-components of a biconnected graph are biconnected.
\end{observation}

For a single crest separator $X$ in
$(G,\varphi,c)$, the set $\{X\}$ splits $(G,\varphi,c)$
into exactly two $(\{X\},\varphi)$-components, which, for
an easier notation,  are also called {\em $(X,\varphi)$-components}.
 For
the $(X,\varphi)$-components
 $\DD$ and $\tilde{\DD}$,
we say that $\tilde{\DD}$ is {\em opposite} to $\DD$.
We say that $X$ {\em goes
weakly
between
two vertex sets} $U_1$ and $U_2$ if
we can number the two
$(X,\varphi)$-components with $C_1=(V_1,E_1)$ and
$C_2=(V_2,E_2)$ such that
$U_1\subseteq V_1$ and $U_2\subseteq V_2$. 
If additionally 
$U_1\cup U_2$ does not contain any vertex of $X$, 
we say that $X$
{\em goes strongly between} 
the sets.
We also say that $X$ goes strongly (or weakly)
between two subgraphs if $X$ goes strongly (or weakly)
between their corresponding vertex sets.
We want to remark that these definitions
focus on the disconnection of faces instead of vertex sets. Nevertheless,
if a crest separator $X$ weakly (strongly) goes between two non-empty
vertex sets $A$ and $B$, then the set of vertices of~$X$ weakly (strongly) disconnects $A$ and
$B$.

Recall that our goal is to find
a subset of the set of all
crest separators {large enough to} separate all crests from each other.
As a first step to restrict the set of crest separators, 
we define a special kind of
path, a so-called ridge, at the end of 
this
paragraph. 
The only crest separators that we need are those
that start at an inner vertex of a ridge. Moreover, we define the
ridge in such a way that each inner vertex of the ridge defines a crest
separator.
We start with the definition of a 
{\em height-vector} of a path $P$. This is a vector $(n_1,\ldots,n_{\ell})$ where $n_i$ ($i\in
\{1,\ldots,\ell\}$)
is the number of vertices of~$P$ whose upper height is
$i$. We say that a height vector $(n_1,...,n_{\ell})$ is smaller
than a height vector $(n'_1,...,n'_{\ell})$ if it is smaller with
respect to the lexicographical order.
{For} vertices $s$ and $t$,
a
{\em ridge} $R$ between $s$ and $t$ is a path connecting $s$ and $t$ 
with a smallest height-vector among all paths connecting $s$ and $t$.
A vertex of~$R$ with smallest upper height $h$  is called a {\em deepest vertex} of~$R$, and
$h$ is called the {\em depth} of~$R$.

\begin{lemma}\label{lem:ridgeCross}
For every inner vertex $u$ of a ridge $R$, there is a neighbor $v$ of~$u$ 
of at most the same upper height
such
that 
the 
down path 
$P_1$ of~$u$ and
the 
down path
$P_2$ of~$v$ define a crest
separator~$X=(P_1,P_2)$ or $X=(P_2,P_1)$ 
that crosses~$R$.
\end{lemma}

\begin{proof}
Consider the 
down path $P_1$ of~$u$.
Let $\bar{u}$ be the vertex of~$P_1$ that belongs to the coast. 
If $\bar{u}\neq u$,
$P_1$ has an edge $\{u',\bar{u}\}$ that is not
part of $R$. Otherwise, 
extend $P_1$
by an edge $\{\bar{u},u^*\}$ to a new virtual
vertex $u^*$ in the outer face. %
 $P_1$ can be extended by an
edge $\{u,v'\}$ such that the resulting path $P'_1$
crosses the ridge. Intuitively speaking, 
$v'$ is a neighbor of~$u$ `on the
other side of the ridge' than $P_1$ or the virtual vertex $u^*$.
Note that there must be indeed at least
one vertex $v'$ on the other side of the ridge 
since, 
otherwise,
$u$ must be a vertex of the coast, the down path of~$u$ has only $u$ as
vertex, and $u$ is incident on both sides to the
outer face, which is a contradiction to the biconnectivity of~$G$.

Let $L$ be the cyclic list of neighbors of~$u$ in clockwise order, and let
$r_1$ and $r_2$ be the two vertices of~$L$ that belong to $R$.
We split $L\setminus \{r_1,r_2\}$ into two sublists~$L_1$ and~$L_2$ where $L_i$ ($i\in \{1,2\}$)
starts with the successor of~$r_i$ and ends with the predecessor of~$r_{3-i}$.
Note that, 
for all vertices $v''$ of the list $L_j$ containing $v'$, the
concatenation of~$v''$ and $P_1$ crosses $R$.
If such a vertex $v''$ exists with $h^+_{\varphi}(v'')\le h^+_{\varphi}(u)$,
then we can take $v=v''$ and the lemma holds.
Let us assume that no such vertex $v''$ exists. Let $R'$ be the path
obtained from $R$ where $u$ is replaced by the vertices in $L_j$.
Then $R'$ has a smaller height-vector than $R$
since $R'$ has one vertex less of height $h^+_{\varphi}(u)$ whereas the
number of vertices with a smaller upper height does not change; a contradiction
to the fact that $R$ is a ridge. 
\end{proof}

\begin{definition}[mountain structure]
Let $\mathcal H$ be a set of crests of $(G,\varphi,c)$.
A {\em mountain structure} for~$(G,\varphi,c)$ and $\mathcal H$ is a tuple
$(G,\varphi,c,{\mathcal H},\mathcal{S})$
with 
$\mathcal{S} \subseteq \mathcal{S}(G,\varphi,c)$ such
that, for each pair of different crests $H_1$ and $H_2$
in
$\mathcal H$
and
for
each ridge $R$ in $(G,\varphi,c)$
with
one endpoint in $H_1$ and the other
in $H_2$, the following property holds:

\begin{itemize}
\item[\rm (a)] There is a crest separator
$X\in \mathcal{S}$
with one of its highest top vertices being a vertex of~$R$
such that $X$ strongly goes between $H_1$ and $H_2$ and such that $X$ has
smallest weighted length
 among all such crest separators in $\mathcal{S}(G,\varphi,c)$.
\end{itemize}
\end{definition}

The next lemma shows that
the simple structure of a crest separator as a tuple of two 
down paths {suffices} to
separate each pair of %
crests of the graph;\ in particular,
property (a) can be easily satisfied 
for each set $\mathcal{H}$ of crests by setting
$\mathcal{S}=\mathcal{S}(G,\varphi,c)$.

\begin{lemma}\label{lem:rid1}
Let $R$ be an ridge between two vertices $s$ and $t$
in
$(G,\varphi,c)$.
Choose $v$ as a deepest inner vertex of~$R$. 
If the depth of~$R$ is
smaller than the
upper height of both $s$ and $t$,
there is a crest separator
in $\mathcal{S}(G,\varphi,c)$
that 
\begin{itemize}
\item  goes strongly between $\{s\}$ and $\{t\}$, and
\item  contains $v$ as a highest top vertex.
\end{itemize}
\end{lemma}

\begin{proof}
By Lemma~\ref{lem:ridgeCross}, there is a crest separator $X=(P_1,P_2)$ with
top vertex~$v$ that crosses $R$. Since $v$ is a deepest vertex of~$R$ and
since $P_1$ and $P_2$ are down paths whose highest vertices have upper height
$\le h^+_{\varphi}(v)$, $R$ and $X$ can cross only once. Note that this also
implies that $X$ neither contains $s$ nor $t$. Thus, $X$ goes strongly,
between $s$ and $t$.
\end{proof}

As a consequence of the last lemma,
the  tuple $(G,\varphi,c,{\mathcal H},\mathcal{S}(G,\varphi,c))$
is a mountain structure for~$(G,\varphi,c)$ and ${\mathcal H}$.
It appears that some crest separators of a mountain structure may be useless
since they split one crest into several crests or they
cut of parts of our original graph
not containing any crests.
Hence, we define the following.

\begin{definition}[good mountain structure]
A mountain structure $(G,\varphi,c,{\mathcal H},\mathcal{S})$
is {\em  good} if, in addition to property (a), also the properties
below hold:

\begin{itemize}
\item[\rm (b)] No crest separator in~$\mathcal{S}$ contains a vertex of
a crest in~$\mathcal H$.
\item[\rm (c)] Each $(\mathcal{S}\!,\varphi )$-component
      contains vertices of a crest in~$\mathcal H$.
\end{itemize}
\end{definition}

Let $(G,\varphi,c,{\mathcal H},\mathcal{S})$ be a mountain structure. 
Note that the properties (a) and (c) imply that each
$(\mathcal{S}, \varphi )$-component contains the vertices of exactly
one crest in~$\mathcal H$, and property (b) guarantees that the
crest is completely contained in one $(\mathcal{S}, \varphi )$-component.
For each ridge $R$ of a pair of crests $H_1$ and $H_2$ in~$\mathcal H$,
the set $\mathcal{S}$ of crest separators 
contains a crest separator
that strongly goes between $H_1$ and $H_2$.

We next want to show
that a good mountain structure exists and
can be computed efficiently. For that we make use of a special graph.
\begin{definition}[mountain connection tree]
The
{\em mountain connection tree $T$} of a mountain structure
$(G,\varphi,c,{\mathcal H},\mathcal{S})$ is a graph defined
as follows:
Each node of~$T$ is
identified with an $(\mathcal{S},\varphi)$-component of~$G$,
and two nodes $w_1$ and $w_2$ of~$T$ are connected if and only if
they---or more precisely the $(\mathcal{S},\varphi)$-components with which
they are identified---have a common top edge
of a crest separator in $\mathcal{S}$.
\end{definition}

Recall that
the
border edges of
a crest separator $X=(P_1,P_2)$
consist of one
top edge $e$
of~$X$
and further down edges,
and that a down edge cannot be a top edge
of any crest separator.
Since the edges of~$X$ are the only edges that are part of both
$(X,\varphi)$-components, the top edge $e$ is the only top edge of a crest separator
that is contained
in both $(X,\varphi)$-components.
Moreover, since two down paths can not cross by definition, for each crest separator $X\in\mathcal{S}$
of a good mountain structure, we can partition the set
of all $(\mathcal{S},\varphi)$-components into
a set $\mathcal{C}_1$ of~$(\mathcal{S},\varphi)$-components
completely contained in one $(X,\varphi)$-component and the
set $\mathcal{C}_2$
of  $(\mathcal{S},\varphi)$-components contained in the other
$(X,\varphi)$-component. Then, $X$ is the only crest separator
with a top edge belonging to
$(\mathcal{S},\varphi)$-components in $\mathcal{C}_1$
as well as in $\mathcal{C}_2$. Consequently, $T$ is indeed a tree.

\begin{lemma}\label{lem:mtree} The mountain connection tree
of a mountain structure
             is a tree.
\end{lemma}

The fact that the top edge of a crest separator $X$ is the only
top edge belonging to both $(X,\varphi)$-components shows also the
correctness of the
next lemma.

\begin{lemma}\label{lem:OnlySep}Let $C_1$ and $C_2$ be two
$(\mathcal{S},\varphi$)-components that are neighbors in the
mountain connection tree of a mountain structure
$(G,\varphi,c,{\mathcal H},\mathcal{S})$. Then there is exactly one
crest separator in $\mathcal{S}$ going weakly between $C_1$ and $C_2$, which
is also the only crest separator with
a top edge belonging to both $C_1$ and $C_2$.
\end{lemma}

Since each of the $O(n)$ crest separators consists
of at most $O(\ell)$ edges, we can determine in $O(\ell n)$ time all 
$(\mathcal{S},\varphi)$-components and afterwards construct the mountain
connection tree by a simple breadth-first search on the 
dual graph of
$(G,\varphi,c)$.

\begin{lemma}\label{lem:mct1} Given 
a mountain structure $(G,\varphi,c,{\mathcal H},\mathcal{S})$ for a set $\mathcal{S}$ of crest separators, its mountain
connection tree
can be
determined in $O(\ell n)$ time.
\end{lemma}

We also can construct a good mountain structure.

\begin{lemma}\label{lem:mctct}
Given $(G,\varphi,c)$ and ${\mathcal H}$,
a good mountain structure $(G,\varphi,c,{\mathcal H},\mathcal{S})$
can be constructed
in $O(\ell n)$ time.
\end{lemma}

\begin{proof}
For a simpler notation, in this proof
we call a crest separator
of a set $\tilde{\mathcal{S}}$
of crest separators
to be a
{\em heaviest crest separator} of~$\tilde{\mathcal{S}}$
if it has a largest weighted length
among all crest separators in $\tilde{\mathcal{S}}$. 

We first construct the set $\mathcal{S}'\!=\mathcal{S}(G,\varphi,c)$ of all
crest separators in $O(\ell n)$ time (Lemma~\ref{lem:allcs}).
Lemma~\ref{lem:rid1} guarantees that, for each pair of crests in ${\mathcal H}$
and each ridge $R$ with endpoints in both crests,
there is
a crest separator $X$ in ${\mathcal S}(G,\varphi,c)$ strongly going between the
two crests 
with
a highest top vertex 
of~$X$ being a
deepest vertex of
$R$.
This in particular means that
all vertices of~$X$ have an upper height smaller than or equal to the depth
of the ridge, which is strictly smaller than the upper height of the two
crests and of all crests through which $R$ passes, i.e.,
no vertex of $X$ is part of a crest in $G$. Hence,
in $O(\ell n)$ time,
we can remove all crest separators from $\mathcal{S}'$
that contain a vertex of a crest and property (a) of a good mountain
structure is maintained.
Afterwards property (b)
holds.
Let $\mathcal{S}''$ be the resulting set of crest separators.
For guaranteeing property (c), we have to remove further
crest separators from $\mathcal{S}''$.
We start with
constructing the mountain connection tree $T$ of
$(G,\varphi,c,{\mathcal H},\mathcal{S}'')$ in $O(\ell n )$ time
(Lemma~\ref{lem:mct1}). We root $T$ at an arbitrary node of $T$.

In a sophisticated
 bottom-up traversal of~$T$
we dynamically update $\mathcal{S}''$
by removing superfluous crest separators in $O(\ell n)$ time.
For a better understanding, before we present a detailed description
of the algorithm,
we roughly sketch some ideas.
Our algorithms marks some
nodes as finished in such a way that the
following invariant (I) always holds:
If a node $C$ of~$T$ is marked as finished, the
$(\mathcal{S}'',\varphi)$-component $C$
contains exactly one crest in~${\mathcal H}$.
The idea of the algorithm is to
process a so far unfinished node
that has only
children
already marked as finished and that, among all such nodes,
has the largest depth in $T$.
When processing a node $w$, we possibly remove
a crest separator $X$ from the current set $\mathcal{S}''$ of
crest separators
with the top edge
of~$X$ belonging to the two
$(\mathcal{S}'',\varphi)$-components identified with $w$
and
a
neighbor~$w'$ of~$w$ in $T$.
If so, by the replacement of
$\mathcal{S}''$ by
$\mathcal{S}''\setminus\{X\}$, we merge the nodes~$w$ and~$w'$ in $T$ to a new node $w^*$.
Additionally, we mark $w^*$ as finished only if $w'$ is
a child of~$w$ since in this case we already know
that $w'$ is already marked as finished,
i.e., $w'$ contains a crest
in~${\mathcal H}$
that is now part of~$w^*$.

We now describe our algorithm in detail. We start
with some preprocessing steps.
In
$O(n)$ time,
we determine and store with each node $w$ of~$T$
  a value ${\mathrm{Crest}}(w)\in\{0,1\}$ that is set
to $1$ if and only if the
$(\mathcal{S}'',\varphi)$-component
identified with
$w$ contains a vertex that is part of a crest 
in~${\mathcal H}$.
In $O(\ell n)$ time, 
we additionally store for each crest separator in $\mathcal{S}''$ its
weighted length,
 mark each node as unfinished, and
store
with
each non-leaf $w$ of~$T$ in a variable
${\mathrm{MaxCrestSep}}(w)$ a heaviest crest separator
of the set of all crest separators
going weakly between
the
$(\mathcal{S}'',\varphi)$-component
identified with $w$ and an $(\mathcal{S}'',\varphi)$-component
identified with a child of
$w$. For each leaf $w$ of~$T$, we define
${\mathrm{MaxCrestSep}}(w)=\mathrm{nil}$.
As a last step of our preprocessing phase,
which also runs in $O(n)$ time,
for each node $w$
of~$T$,
we initialize a value ${\mathrm{MaxCrestSep}}^*(w)$ with $\mathrm{nil}$.
${\mathrm{MaxCrestSep}}^*(w)$ is defined analogously {to}
${\mathrm{MaxCrestSep}}(w)$ if we restrict the crest separators
to be considered only to those crest separators that go
weakly between
two
$(\mathcal{S}'',\varphi)$-components
identified with
$w$ and with a finished child of~$w$.
We will see that it suffices to know the correct values of
${\mathrm{MaxCrestSep}}(w)$ and ${\mathrm{MaxCrestSep}}^*(w)$
only for the unfinished nodes and therefore we do not update these
values for finished nodes.

We next
describe the processing of a node $w$ during the
traversal of~$T$ in detail.
Keep in mind that $\mathcal{S}''$
is always equal to the
current set of remaining crest separators,
which is updated dynamically.
First we exclude the case, where $w$ has a parent $\tilde{w}$ with
an unfinished child $\hat{w}$, and where ${\mathrm{MaxCrestSep}}(\tilde{w})$
is
equal to the crest separator going weakly
between the $(\mathcal{S}'',\varphi)$-components
identified with $w$ and $\tilde{w}$.
More precisely, in this case we delay the processing of~$w$
and continue with the processing of~$\hat{w}$. 

Second, we test whether the
$(\mathcal{S}'',\varphi)$-component $C$
identified with $w$ contains a vertex belonging to a crest 
in~${\mathcal H}$,
which is exactly the case if $\mathrm{Crest}(w)=1$.
In this case,
we mark $w$ as finished.
If there is a parent $\tilde{w}$ of~$w$ with $X$ being
the crest separator
going weakly between the two
$(\mathcal{S}'',\varphi)$-components
identified with $w$ and to $\tilde{w}$,
\mbox{we replace} ${\mathrm{MaxCrestSep}^*}(\tilde{w})$ by a
heaviest
crest separator in $\{X, {\mathrm{MaxCrestSep}^*}(\tilde{w})\}$.

Let us next consider the case where $\mathrm{Crest}(w)=0$.
We then remove
a heaviest crest separator $X$ of the set of
all crest separators going weakly between
the $(\mathcal{S}'',\varphi)$-component
identified with $w$ and an $(\mathcal{S}'',\varphi)$-component
identified with a neighbor of~$w$, i.e., either with
(Case i) a child $\hat{w}$ of~$w$ or (Case ii) the parent 
$\tilde{w}$ of
$w$.
$X$ can be taken as either ${\mathrm{MaxCrestSep}}(w)$ or
the crest separator going weakly between the two
$(\mathcal{S}'',\varphi)$-components identified with
$w$ and $\tilde{w}$.
\begin{description}
\item[Case i:]
We remove $X$ from $\mathcal{S}''$ and mark the node $w^*$ obtained from merging $w$ and $\hat{w}$ as finished
and set $\mathrm{Crest}(w^*)=1$.
Note that this is correct
since $\hat{w}$ is already marked as finished and therefore
$\mathrm{Crest}(\hat{w})=1$.
In addition, if
the parent~$\tilde{w}$ of~$w$ exists, we
replace
$\mathrm{MaxCrestSep}^*(\tilde{w})$
by the heaviest
crest separator contained in $\{X',{\mathrm{MaxCrestSep}^*}(\tilde{w})\}$
where $X'$ is
the crest separator
in $\mathcal{S}''$
going strongly
between the $(\mathcal{S}'',\varphi)$-components identified with
$w$ and $\tilde{w}$.
\item[Case ii:]
We
mark
the node $w^*$ obtained
from merging the unfinished nodes $w$ and $\tilde{w}$ as unfinished,
set $\mathrm{Crest}(w^*\hspace{-0.5pt})=\mathrm{Crest}(\tilde{w})$,
and define the value
$\mathrm{MaxCrestSep}^*(w^*\hspace{-0.5pt})$ as the heaviest
crest separator in $\{\mathrm{MaxCrestSep}^*(w),$
$\mathrm{MaxCrestSep}^*(\tilde{w})\}$ or $\mathrm{nil}$ if this set
contains no crest separator.
If $\tilde{w}$ beside $w$ has another unfinished
child,
we define $\mathrm{MaxCrestSep}(w^*)$ as the heaviest
crest separator in $\{\mathrm{MaxCrestSep}(w),
\mathrm{MaxCrestSep}(\tilde{w})\}$. 
Note
that $\mathrm{MaxCrestSep}(\tilde{w})\not=X$ since
otherwise the processing of~$w$ would have been delayed.
If $\tilde{w}$ has no other unfinished child,
we take $\mathrm{MaxCrestSep}(w^*)$ as the heaviest
crest separator that is contained in $\{
\mathrm{MaxCrestSep}(w),$
$\mathrm{MaxCrestSep}^*(\tilde{w})\}$ or $\mathrm{nil}$
if no crest separator is in this set.
Note that $\mathrm{MaxCrestSep}^*(\tilde{w})\neq X$ since
$w$ is unfinished before its processing.
\end{description}

\noindent We can conclude by induction that
  our algorithm correctly updates $\mathrm{Crest}(w)$ for all nodes $w$ as
 well as 
$\mathrm{MaxCrestSep}(w)$ and
$\mathrm{MaxCrestSep}^*(w)$ for all unfinished nodes $w$.
If the processing of
a node in $T$ is delayed, the processing of the next node considered
is not delayed. Hence the running time is dominated by the non-delayed
processing steps. If the processing of a node is not delayed either
two  $(\mathcal{S}'',\varphi)$-components are merged or a node $w$ in
$T$ is marked as finished.
Hence the algorithm stops after $O(n)$ processing steps, i.e., in $O(n)$ time,
with all nodes of~$T$ being marked as finished.

Note that, if during the processing of a so far unfinished
node $w$, we remove a crest separator $X$ weakly going between
the $(\mathcal{S}'',\varphi)$-components identified with
$w$ and a neighbor of~$w$, we know that $w$ itself contains
no crest in $\mathcal{H}$. Hence, if $X$ goes
strongly
between two crests in~${\mathcal H}$,  
then there is another crest separator
going strongly between
these two crests that also goes strongly
between 
{two} $(\mathcal{S}'',\varphi)$-components identified with $w$ and one 
{of} its neighbors. This together with the fact that we have chosen
 $X$ as heaviest crest separator guarantees that property (a) is 
maintained during our processing. Since no crest separators are added
 into $\mathcal{S}''$, property (b) also holds.
The fact that the
algorithm
marks a node as finished only if its
identified
$(\mathcal{S}'',\varphi)$-component contains the vertices of a crest
in~${\mathcal H}$
implies that
the invariant (I) holds before
and after each processing of a node $w$. 
At the end of the
algorithm
invariant (I) 
guarantees that property (c)
holds. 
To sum up,
$(G,\varphi,c,{\mathcal H},\mathcal{S}'')$ defines a good mountain structure at the end of
the algorithm.
\end{proof}

\section{Connection between Coast Separators\\ and Pseudo Shortcuts}\label{sec:interact}

As in the last section, 
${\mathcal H}$ is a set of crests in a weighted,
almost triangulated, and biconnected graph $(G,\varphi,c)$.
Let $(G,\varphi,c,{\mathcal H},\mathcal{S})$ be a
good mountain structure.
As part of our algorithm,
for  each $(\mathcal{S},\varphi)$-component, 
we want to compute a coast separator
that strongly disconnects its crest 
in~$\mathcal H$ 
from the coast and that is
of weighted size
$q \cdot {\mathrm{tw}}(G)$ for some constant~$q$.
In more detail, we are interested in %
minimal coast separators as already mentioned in Section~\ref{sec:ide}
and
in Section~\ref{sec:interactn2}
we choose %
$\mathcal{H}$ as a special set of crests that allow us to compute such 
coast separators.

For a cycle~$Q$ in
$(G,\varphi,c)$, 
we say
that $Q$ {\em encloses} a vertex $u$, a vertex set $U$, and
a subgraph $H$ of~$G$,
if the set of the vertices of the cycle weakly disconnects the coast
from
$\{u\}$%
, $U$, and the vertex set of~$H$, respectively.
A crest separator $X$ {\em encloses} a vertex $u$, a vertex set $U$,
and a subgraph $H$ of $G$ if it has a lowpoint and the cycle
induced by the %
edges of the essential boundary of
$X$ encloses $\{u\}$%
, $U$, and the vertex set of~$H$, respectively.  
The {\em inner graph} of a cycle~$Q$ or a coast separator $Q$ 
is the plane graph induced by
the vertices that are weakly disconnected by $Q$ from the coast.
Recall that a
coast separator
for a set $U$ is {\em minimal} for
$U$ if it has minimal weighted size among all coast separators for~$U$
and if among all such separators
its inner graph has a minimal number of faces.
\begin{observation}
The vertex set of a minimal coast separator
for a connected set of vertices $H$
in
a weighted, almost triangulated graph induces a cycle.
\end{observation}

Unfortunately, for a crest $H\in \mathcal H$ in an $(\mathcal{S},\varphi)$-component $\CC$,
there might be no minimal coast separator $Y$ for~$H$ of 
weighted size $O(tw(G))$ in $G$
that is
completely contained in $\CC$. 
In this case,
there is a crest separator $X\in\mathcal{S}$
such that the cycle induced by $Y$
contains a subpath $P$ 
with the following properties: $P$
starts and ends with %
vertices of~$X$ and, additionally,
$P$ is contained in the $(X,\varphi)$-component
not containing
$\CC$.
Then, since $P$ can not be replaced by a crest-separator path of~$X$ 
of at most the same length, 
$P$ must be a pseudo shortcut as defined next. 
(For property (1) of the following 
definition, see the remark immediately after the definition.)

\begin{definition}[%
pseudo shortcut, composed cycle, inner graph]
Let $X=(P_1,P_2)$ be a crest separator in
$(G,\varphi,c)$ with an $(X,\varphi)$-component $\DD$. 
Let $X^\mathrm{CP}$ be a crest-separator path of~$X$ connecting
vertices $s_1$ and $s_2$
part
of the essential
boundary of~$X$.
Then, a path $P$ from~$s_1$ to~$s_2$
in $\DD$ is called an
{\em ($s_1$-$s_2$-connecting) ($\DD$-)pseudo shortcut ({\em of} $X^\mathrm{CP}$)}
if the following three conditions hold: 
\begin{description}
\item[(1)] if $X$ encloses the $(X,\varphi)$-component opposite to $\DD$,
then
$X$ has an exterior lowpoint, i.e, a lowpoint belonging to the coast.
\item[(2)] $P$ has a strictly shorter weighted length than $X^\mathrm{CP}$.

\item[(3)] $P$ does not contain any vertex
of the coast.
\end{description}
We call the cycle
consisting of the edges of $P$ and 
of~$X^{\mathrm{CP}}$ the {\em composed cycle} of $(X^{\mathrm{CP}},P)$. 
Moreover, a path $P$ is a {\em pseudo shortcut of a crest separator} $X$ if it is
a pseudo shortcut for some crest-separator path of~$X$.
The {\em inner graph} of a
pseudo shortcut $P$ of a
crest-separator path $X^{\mathrm{CP}}$ is the inner graph of the
composed cycle of~$(X^{\mathrm{CP}},P )$. 
\end{definition}

Roughly speaking, for a crest separator $X$ with an interior lowpoint, we are only interested in 
pseudo shortcuts that are enclosed by $X$. These can be used for the
construction of a coast separator for a crest $H\in \mathcal H$
in the 
$(X,\varphi)$-component $D$ not enclosed by $X$. For the crests
in $\mathcal H$
 enclosed by $X$,
 we can use the essential boundary of~$X$ as a coast separator.
This is 
 the reason why we distinguish between exterior and interior lowpoints, 
and why
we restrict our definition of pseudo shortcuts by Condition 1.
In return, we so can avoid 
 complicated special cases in the use of pseudo
 shortcuts.

Note that,
for two vertices~$s_1$ and~$s_2$ with neither~$s_1$ nor~$s_2$ being equal to
the lowpoint of a crest separator $X$, 
there is only one
crest-separator path
connecting~$s_1$ and $s_2$, 
and we compare the weighted length of a weighted path $P$
connecting~$s_1$~and~$s_2$ with the unique weighted length of the
crest-separator path connecting~$s_1$ and~$s_2$.
Moreover, we restrict
the endpoints of a
pseudo shortcut $P$ to be part of the essential boundary
since,
otherwise, one of~$P_1$ and $P_2$ must
contain both endpoints. Then $P$ cannot
have a shorter length than the subpath of~$P_1$ or $P_2$
connecting the two endpoints as shown by part (a) of the next lemma.

\begin{lemma}\label{lem:ShortAlley}
For each crest separator $X=(P_1,P_2)$ the following holds:
\begin{itemize}
\item[\rm (a)]
For each $i\in\{1,2\}$,
no path with endpoints $v_1$ and $v_2$ in $P_i$ can have
a shorter
weighted length than the shortest $v_1$-$v_2$-connecting crest-separator path.
\item[\rm (b)] Let $R$ be a ridge connecting two vertices
of different crests $H_1$ and $H_2$ in $\mathcal H$. 
If the weighted length of~$X$ is
not larger than
the weighted length of
a crest separator of shortest weighted length separating
$H_1$ and $H_2$, $R$ and a pseudo shortcut $P$ of
a crest-separator path $X^{\mathrm CP}$ of~$X$
cannot 
cross.
\end{itemize}
\end{lemma}

\begin{proof}
To show part (a), we use the fact that %
every path from $v_1$ to $v_2$ has
weighted length
of at least
$\max(h^+_{\varphi}(v_1)-h^-_{\varphi}(v_2),
h^+_{\varphi}(v_2)-h^-_{\varphi}(v_1))+1$ whereas 
the crest-separator path with endpoints
$v_1$ and $v_2$ has exactly that length. Thus,
(a) holds.

To show part (b),
let $\tilde{X}$ be a crest separator with shortest weighted length
among all crest
separators separating $H_1$ and $H_2$.
See Fig.~\ref{fig:crossPS}.
\begin{figure}[b!]
    \begin{center}
       \scalebox{0.8}{
\scalebox{1} 
{
\begin{pspicture}(0,-2.1932693)(5.9471874,2.1812692)
\definecolor{color497}{rgb}{0.12941,0.129411,1.0}
\definecolor{color3}{rgb}{1.0,0.4,0.4}
\definecolor{color265}{rgb}{0.35294,0.352941,0.352941}
\usefont{T1}{ptm}{m}{n}
\rput(1.1373436,-2.130457){$\tilde{X}$}
\usefont{T1}{ptm}{m}{n}
\rput(3.3523438,-1.050332){\color{color3}$P$}
\psline[linewidth=0.056199998cm](1.1373436,2.3531692)(1.1373436,-1.7703317)
\usefont{T1}{ptm}{m}{n}
\rput(0.26234376,0.3096681){$R$}
\usefont{T1}{ptm}{m}{n}
\rput(1.8824686,0.457){$u$}
\rput(1.9824686,-2.170457){$X^*$}
\psdots[dotsize=0.16](4.3373437,0.30332437)
\usefont{T1}{ptm}{m}{n}
\rput(5.177656,0.171817){top vertex}
\psline[linewidth=0.056199998cm](4.3373437,2.3531692)(4.3373437,-1.7703317)
\usefont{T1}{ptm}{m}{n}
\rput(4.3374686,-2.170457){$X$}
\psdots[dotsize=0.16](1.1374687,0.30332437)
\psline[linewidth=0.056199998cm](3.0504687,0.3032693)(0.4554687,0.3032693)
\psline[linewidth=0.056199998cm,linecolor=color265](2.0573437,2.3331692)(2.0573437,-1.7903318)
\psbezier[linewidth=0.02,linecolor=color3](4.3304687,1.1612693)(3.9504688,1.1412693)(2.6504686,1.4612693)(2.0504687,0.2812693)(1.4504687,-0.8987307)(3.2704687,-1.3387307)(4.3304687,-1.3387307)
\psdots[dotsize=0.16](2.0574687,0.28332436)
\psdots[dotsize=0.12](4.3374686,-1.3484567)
\usefont{T1}{ptm}{m}{n}
\rput(4.5859373,-1.5085505){$s_2$}
\psdots[dotsize=0.12](4.3374686,1.1515433)
\usefont{T1}{ptm}{m}{n}
\rput(4.5859373,0.99398756){$s_1$}
\psbezier[linewidth=0.03,linecolor=color497,linestyle=dotted,dotsep=0.10cm](4.2904687,1.2062693)(3.9078298,1.1852694)(2.598802,1.5212693)(1.9946353,0.2822693)(1.3904687,-0.9567307)(3.2231076,-1.4187307)(4.2904687,-1.4187307)
\psline[linewidth=0.03cm,linecolor=color497,linestyle=dotted,dotsep=0.10cm](4.2704687,2.3532693)(4.2704687,1.2012693)
\psline[linewidth=0.03cm,linecolor=color497,linestyle=dotted,dotsep=0.10cm](4.2704687,-1.4187307)(4.2704687,-1.7707307)
\usefont{T1}{ptm}{m}{n}
\rput(3.8723438,1.7496681){\color{color497}$Q$}
\end{pspicture} 
}}%
    \end{center}
    \caption{The crest separators $X,\tilde{X},X^*$,  the ridge $R$, the pseudo shortcut
    $P$, and the path $Q$ as
    described in the proof of Lemma~\ref{lem:ShortAlley}(b).
             }
    \label{fig:crossPS}
  \end{figure}
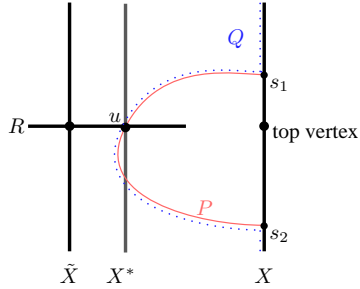
Note that $\tilde{X}$ must cross $R$.
Assume
for a contradiction that (b) does not hold, i.e., $R$ and a
pseudo shortcut $P$ of~$X^{\mathrm CP}$ cross. 
Choose $u$ among all
crossing vertices of~$R$ and $P$ with
minimal upper height.
With $s_1$ and $s_2$ being the endpoints of~$P$,
let $Q$ be the path obtained from 
$P_1\circ P_2$
by replacing 
$X^{\mathrm CP}$
by $P$.
Note that $Q$
has a smaller weighted length than
$X$ since $P$ is a pseudo shortcut of $X^{\mathrm CP}$.
Let $P^*_1$ be the down path
of~$u$, and choose $P^*_2$ as a down path starting in a neighbor of~$u$
such that
$X^*=(P^*_1,P^*_2)$ or $X^*=(P^*_2,P^*_1)$ is a crest separator 
that
crosses $R$ and that, among all possible choices, has shortest weighted length. $X^*$ exists by Lemma~\ref{lem:ridgeCross}.

The weighted length of~$P_1^*$ plus the  weighted
 length of~$P_2^*$ can not be larger than the weighted
 length of 
$Q$
 since $Q$ consists of a subpath $Q'$ from a vertex of the coast to $u$ and
 another subpath $Q''$ back to the coast.
More precisely,
assume that $P_1^*$ and~$Q'$ leave $R$ on the same side,
 whereas $P_2^*$ and $Q''$ leave $R$ on the other side. 
Let $u'$ be the first vertex of~$Q''$. 
Note that $u'$ is not the down vertex of $u$. If
$\{u,u'\}$ is   
the top edge of a crest separator
crossing $R$, then by definition of~$X^*$ as a crest separator
with smallest weighted length among all possible choices, $P^*_2$
must start in a vertex
with lower or equal upper height than that of~$u'$. 
 Otherwise, i.e., if $u$ is the down vertex of $u'$ or if $u'$
belongs to $R$ and the down
path of $u'$ leaves $R$ on the same side than the down path of $u$,
$u'$ has upper height at least
 $h^+_{\varphi}(u)$, whereas the upper height of the first vertex of $P_2^*$
 is at most $h^+_{\varphi}(u)$ by Lemma~\ref{lem:ridgeCross} and
the fact $X^*$ that was chosen as a crest separator
with smallest weighted length among all possible choices.

So far we can conclude that
$X^*$ has
weighted length smaller than or equal to the  weighted
length of~$Q$.
As long as $X^*$ and $R$ cross more than once, there is another
common vertex $\tilde{u}$ of~$R$ and $X^*$ with $h^+_{\varphi}(\tilde{u})< h^+_{\varphi}(u)$ and we 
redefine
$X^*$ as a crest separator crossing $R$ with $\tilde{u}$ being a top vertex of~$X^*$. Finally, $X^*$ and $R$ cross only once, and thus $X^*$ strongly 
disconnects $H_1$
and $H_2$. Since our iteration of choosing $X^*$ only shrinks the weighted
length of~$X^*$, $X^*$ has weighted length smaller than or equal to the  weighted
length of~$Q$, which is strictly smaller than the weighted length of~$X$.
Since the weighted length of $X$ is smaller than or equal to the
weighted
length of~$\tilde{X}$, 
the weighted length of $X^*$ is strictly
smaller than the weighted length of~$\tilde{X}$.
Contradiction.
\end{proof}

For the rest of this section, let $(G,\varphi,c,{\mathcal H},\mathcal{S})$ be a good mountain
structure of $(G,\varphi,c)$.
The endpoints of a pseudo shortcut for
a crest separator $X\in \mathcal{S}$ 
are---intuitively speaking---the vertices between which a coast separator can leave one
$(X,\varphi)$-component and later reenter the $(X,\varphi)$-component.
We describe this intuition more precisely in
Lemma~\ref{lem:StrongOptimal}.
Let $X$ be a crest separator, and let $D$ be an $(X,\varphi)$-component.
A subpath $P''$ of a path $P'$ in $G$ is called
to be a {\em \spp{\DD}{X}-subpath of} $P'$ if it is a path contained in $D$,
that
starts either with a lowpoint of~$X$ or an edge not being part of the essential
boundary of~$X$ and that also ends either with a lowpoint of~$X$ or an edge not being 
part of
the essential boundary of~$X$. 
$P''$ is called {\em maximal} if no other
\spp{\DD}{X}-subpath of~$P'$ contains $P''$ as a proper subpath.

\begin{definition}[strict pseudo shortcut] 
A $\DD$-pseudo shortcut $P$
of a crest-separator path $X^\mathrm{CP}$ 
of a crest separator $X\in\mathcal{S}$
is called {\em strict} if
it has shortest weighted length among all 
$\DD$-pseudo shortcuts of~$X^\mathrm{CP}$
and if among those 
the composed cycle $(X^\mathrm{CP},P)$ encloses a minimum number of faces.
\end{definition}

Intuitively, the next lemma shows that in some cases, the part of a coast separator
behind a crest separator is a strict pseudo shortcut and that
an analogous result holds for parts of a strict pseudo shortcut 
behind a crest separator. 
For the next lemma, see  
also the examples %
in
Fig.~\ref{fig:crestSho}.

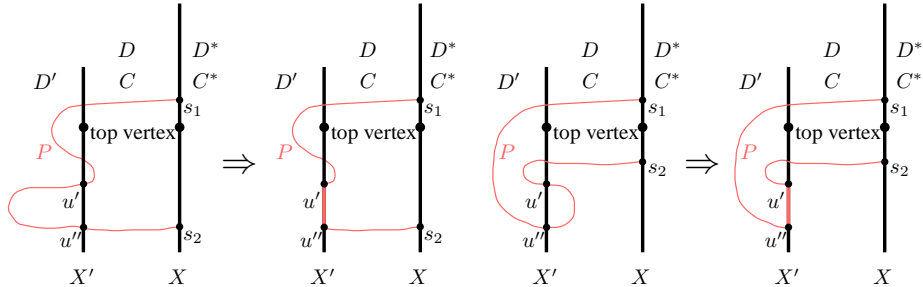
\begin{figure}[b!]
    \begin{center}
       \vspace{-4mm}
       \scalebox{0.8}{
\scalebox{1} 
{
\begin{pspicture}(0,-2.3675942)(3.1573124,2.3556943)
\definecolor{color3014}{rgb}{1.0,0.4,0.4}
\pscustom[linewidth=0.02,linecolor=color3014]
{
\newpath
\moveto(2.81,0.7359681)
\lineto(2.04,0.6959683)
\curveto(1.655,0.6759682)(1.245,0.65096825)(1.22,0.64596826)
\curveto(1.195,0.64096826)(1.13,0.6309683)(1.09,0.6259683)
\curveto(1.05,0.6209682)(1.0,0.61096823)(0.97,0.5959682)
}
\pscustom[linewidth=0.02,linecolor=color3014]
{
\newpath
\moveto(0.97,0.59596825)
\lineto(0.91,0.5559682)
\curveto(0.88,0.53596824)(0.84,0.5009683)(0.83,0.48596826)
\curveto(0.82,0.47096825)(0.795,0.43596825)(0.78,0.41596824)
\curveto(0.765,0.39596826)(0.74,0.35096824)(0.73,0.32596827)
\curveto(0.72,0.30096826)(0.71,0.25096825)(0.71,0.22596826)
\curveto(0.71,0.20096825)(0.715,0.15596825)(0.72,0.13596825)
\curveto(0.725,0.11596825)(0.745,0.07096825)(0.76,0.045968253)
\curveto(0.775,0.02096825)(0.81,-0.019031748)(0.83,-0.03403175)
\curveto(0.85,-0.04903175)(0.89,-0.07903175)(0.91,-0.09403175)
\curveto(0.93,-0.10903175)(0.97,-0.13903175)(0.99,-0.15403175)
\curveto(1.01,-0.16903175)(1.055,-0.19403175)(1.08,-0.20403175)
\curveto(1.105,-0.21403176)(1.15,-0.22903174)(1.17,-0.23403175)
\curveto(1.19,-0.23903175)(1.23,-0.25903174)(1.25,-0.27403176)
\curveto(1.27,-0.28903174)(1.31,-0.31903175)(1.33,-0.33403176)
\curveto(1.35,-0.34903175)(1.38,-0.38403174)(1.39,-0.40403175)
\curveto(1.4,-0.42403173)(1.41,-0.48403174)(1.41,-0.52403176)
\curveto(1.41,-0.5640317)(1.39,-0.61403173)(1.37,-0.6240317)
\curveto(1.35,-0.6340318)(1.305,-0.64903176)(1.28,-0.65403175)
\curveto(1.255,-0.65903175)(1.2,-0.66903174)(1.17,-0.67403173)
\curveto(1.14,-0.6790317)(1.075,-0.6940318)(1.04,-0.70403177)
\curveto(1.005,-0.71403176)(0.92,-0.72403175)(0.87,-0.72403175)
\curveto(0.82,-0.72403175)(0.73,-0.72403175)(0.69,-0.72403175)
\curveto(0.65,-0.72403175)(0.57,-0.72403175)(0.53,-0.72403175)
\curveto(0.49,-0.72403175)(0.415,-0.72403175)(0.38,-0.72403175)
\curveto(0.345,-0.72403175)(0.285,-0.73403174)(0.26,-0.7440317)
\curveto(0.235,-0.7540318)(0.185,-0.77403176)(0.16,-0.78403175)
\curveto(0.135,-0.79403174)(0.095,-0.8240318)(0.08,-0.84403175)
\curveto(0.065,-0.86403173)(0.04,-0.90903175)(0.03,-0.9340317)
\curveto(0.02,-0.95903176)(0.0050,-1.0090318)(0.0,-1.0340317)
\curveto(-0.0050,-1.0590317)(0.0,-1.1040318)(0.01,-1.1240318)
\curveto(0.02,-1.1440318)(0.045,-1.1840317)(0.06,-1.2040317)
\curveto(0.075,-1.2240318)(0.105,-1.2640318)(0.12,-1.2840317)
\curveto(0.135,-1.3040317)(0.175,-1.3290317)(0.2,-1.3340317)
\curveto(0.225,-1.3390317)(0.275,-1.3490318)(0.3,-1.3540318)
\curveto(0.325,-1.3590318)(0.375,-1.3690318)(0.4,-1.3740318)
\curveto(0.425,-1.3790318)(0.51,-1.3890318)(0.57,-1.3940318)
\curveto(0.63,-1.3990318)(0.745,-1.3990318)(0.8,-1.3940318)
\curveto(0.855,-1.3890318)(0.935,-1.3790318)(0.96,-1.3740318)
\curveto(0.985,-1.3690318)(1.04,-1.3640318)(1.07,-1.3640318)
\curveto(1.1,-1.3640318)(1.16,-1.3740318)(1.19,-1.3840318)
\curveto(1.22,-1.3940318)(1.29,-1.4090317)(1.33,-1.4140317)
\curveto(1.37,-1.4190317)(1.455,-1.4240317)(1.5,-1.4240317)
\curveto(1.545,-1.4240317)(1.63,-1.4290317)(1.67,-1.4340317)
\curveto(1.71,-1.4390317)(1.805,-1.4390317)(1.86,-1.4340317)
\curveto(1.915,-1.4290317)(2.02,-1.4240317)(2.07,-1.4240317)
\curveto(2.12,-1.4240317)(2.215,-1.4240317)(2.26,-1.4240317)
\curveto(2.305,-1.4240317)(2.395,-1.4240317)(2.44,-1.4240317)
\curveto(2.485,-1.4240317)(2.565,-1.4140317)(2.6,-1.4040318)
\curveto(2.635,-1.3940318)(2.695,-1.3840318)(2.72,-1.3840318)
\curveto(2.745,-1.3840318)(2.785,-1.3790318)(2.8,-1.3740318)
}
\usefont{T1}{ptm}{m}{n}
\rput(2.7855937,-2.1960318){$X$}
\usefont{T1}{ptm}{m}{n}
\rput(1.9454688,1.5640931){$D$}
\rput(3.2454688,1.5640931){$D^*$}
\rput(3.2454688,1.0440931){$C^*$}
\psline[linewidth=0.056199998cm](2.811875,2.3275943)(2.811875,-1.7959068)
\usefont{T1}{ptm}{m}{n}
\rput(0.57546875,-0.13575801){\color{color3014}$P$}
\psdots[dotsize=0.12](2.812,0.7259682)
\psdots[dotsize=0.12](2.812,-1.3740318)
\usefont{T1}{ptm}{m}{n}
\rput(1.9554688,1.0440931){$C$}
\psline[linewidth=0.056199998cm](1.23,1.2759682)(1.23,-1.7959068)
\usefont{T1}{ptm}{m}{n}
\rput(1.2055937,-2.1960318){$X'$}
\psdots[dotsize=0.16](1.232,0.27774936)
\usefont{T1}{ptm}{m}{n}
\rput(1.0290625,-0.92412555){$u'$}
\usefont{T1}{ptm}{m}{n}
\rput(1.0290625,-1.5841255){$u''$}
\rput(3.0290625,-1.5341255){$s_2$}
\rput(3.0290625,0.56841255){$s_1$}
\psdots[dotsize=0.12](1.232,-0.6622506)
\psdots[dotsize=0.12](1.232,-1.3822507)
\psdots[dotsize=0.16](2.812,0.27774936)
\usefont{T1}{ptm}{m}{n}
\rput(2.0460937,0.146242){top vertex}
\usefont{T1}{ptm}{m}{n}
\rput(0.60546875,1.0440931){$D'$}
\end{pspicture} 
}\hspace{0.2cm}%
       \rput(0.2,2){\scalebox{1.7}{$\Rightarrow$}}\hspace{0.2cm}%
\scalebox{1} 
{
\begin{pspicture}(0,-2.3675942)(3.2073126,2.3556943)
\definecolor{color116}{rgb}{1.0,0.4,0.4}
\pscustom[linewidth=0.02,linecolor=color116]
{
\newpath
\moveto(1.02,0.59596825)
\lineto(0.96,0.5559682)
\curveto(0.93,0.53596824)(0.89,0.5009683)(0.88,0.48596826)
\curveto(0.87,0.47096825)(0.845,0.43596825)(0.83,0.41596824)
\curveto(0.815,0.39596826)(0.79,0.35096824)(0.78,0.32596827)
\curveto(0.77,0.30096826)(0.76,0.25096825)(0.76,0.22596826)
\curveto(0.76,0.20096825)(0.765,0.15596825)(0.77,0.13596825)
\curveto(0.775,0.11596825)(0.795,0.07096825)(0.81,0.045968253)
\curveto(0.825,0.02096825)(0.86,-0.019031748)(0.88,-0.03403175)
\curveto(0.9,-0.04903175)(0.94,-0.07903175)(0.96,-0.09403175)
\curveto(0.98,-0.10903175)(1.02,-0.13903175)(1.04,-0.15403175)
\curveto(1.06,-0.16903175)(1.105,-0.19403175)(1.13,-0.20403175)
\curveto(1.155,-0.21403176)(1.2,-0.22903174)(1.22,-0.23403175)
\curveto(1.24,-0.23903175)(1.28,-0.25903174)(1.3,-0.27403176)
\curveto(1.32,-0.28903174)(1.36,-0.31903175)(1.38,-0.33403176)
\curveto(1.4,-0.34903175)(1.43,-0.38403174)(1.44,-0.40403175)
\curveto(1.45,-0.42403173)(1.46,-0.48403174)(1.46,-0.52403176)
\curveto(1.46,-0.5640317)(1.44,-0.61403173)(1.42,-0.6240317)
\curveto(1.4,-0.6340318)(1.355,-0.64903176)(1.33,-0.65403175)
\curveto(1.305,-0.65903175)(1.25,-0.66903174)(1.22,-0.67403173)
\curveto(1.19,-0.6790317)(1.125,-0.6940318)(1.09,-0.70403177)
\curveto(1.055,-0.71403176)(0.97,-0.72403175)(0.92,-0.72403175)
\curveto(0.87,-0.72403175)(0.78,-0.72403175)(0.74,-0.72403175)
\curveto(0.7,-0.72403175)(0.62,-0.72403175)(0.58,-0.72403175)
\curveto(0.54,-0.72403175)(0.465,-0.72403175)(0.43,-0.72403175)
\curveto(0.395,-0.72403175)(0.335,-0.73403174)(0.31,-0.7440317)
\curveto(0.285,-0.7540318)(0.235,-0.77403176)(0.21,-0.78403175)
\curveto(0.185,-0.79403174)(0.145,-0.8240318)(0.13,-0.84403175)
\curveto(0.115,-0.86403173)(0.09,-0.90903175)(0.08,-0.9340317)
\curveto(0.07,-0.95903176)(0.055,-1.0090318)(0.05,-1.0340317)
\curveto(0.045,-1.0590317)(0.05,-1.1040318)(0.06,-1.1240318)
\curveto(0.07,-1.1440318)(0.095,-1.1840317)(0.11,-1.2040317)
\curveto(0.125,-1.2240318)(0.155,-1.2640318)(0.17,-1.2840317)
\curveto(0.185,-1.3040317)(0.225,-1.3290317)(0.25,-1.3340317)
\curveto(0.275,-1.3390317)(0.325,-1.3490318)(0.35,-1.3540318)
\curveto(0.375,-1.3590318)(0.425,-1.3690318)(0.45,-1.3740318)
\curveto(0.475,-1.3790318)(0.56,-1.3890318)(0.62,-1.3940318)
\curveto(0.68,-1.3990318)(0.795,-1.3990318)(0.85,-1.3940318)
\curveto(0.905,-1.3890318)(0.985,-1.3790318)(1.01,-1.3740318)
\curveto(1.035,-1.3690318)(1.09,-1.3640318)(1.12,-1.3640318)
\curveto(1.15,-1.3640318)(1.21,-1.3740318)(1.24,-1.3840318)
\curveto(1.27,-1.3940318)(1.34,-1.4090317)(1.38,-1.4140317)
\curveto(1.42,-1.4190317)(1.505,-1.4240317)(1.55,-1.4240317)
\curveto(1.595,-1.4240317)(1.68,-1.4290317)(1.72,-1.4340317)
\curveto(1.76,-1.4390317)(1.855,-1.4390317)(1.91,-1.4340317)
\curveto(1.965,-1.4290317)(2.07,-1.4240317)(2.12,-1.4240317)
\curveto(2.17,-1.4240317)(2.265,-1.4240317)(2.31,-1.4240317)
\curveto(2.355,-1.4240317)(2.445,-1.4240317)(2.49,-1.4240317)
\curveto(2.535,-1.4240317)(2.615,-1.4140317)(2.65,-1.4040318)
\curveto(2.685,-1.3940318)(2.745,-1.3840318)(2.77,-1.3840318)
\curveto(2.795,-1.3840318)(2.835,-1.3790318)(2.85,-1.3740318)
}
\psframe[linewidth=0.04,linecolor=white,dimen=outer,fillstyle=solid](1.28,-0.603758)(0.0,-1.483758)
\psline[linewidth=0.056199998cm](1.28,1.2759682)(1.28,-1.7959068)
\psline[linewidth=0.056199998cm,linecolor=color116](1.28,-0.643758)(1.28,-1.403758)
\pscustom[linewidth=0.02,linecolor=color116]
{
\newpath
\moveto(2.86,0.7359681)
\lineto(2.09,0.6959683)
\curveto(1.705,0.6759682)(1.295,0.65096825)(1.27,0.64596826)
\curveto(1.245,0.64096826)(1.18,0.6309683)(1.14,0.6259683)
\curveto(1.1,0.6209682)(1.05,0.61096823)(1.02,0.5959682)
}
\usefont{T1}{ptm}{m}{n}
\rput(2.8355937,-2.1960318){$X$}
\usefont{T1}{ptm}{m}{n}
\rput(1.9954687,1.5640931){$D$}
\rput(3.2454688,1.5640931){$D^*$}
\rput(3.2454688,1.0440931){$C^*$}
\psline[linewidth=0.056199998cm](2.861875,2.3275943)(2.861875,-1.7959068)
\psdots[dotsize=0.12](2.862,0.7259682)
\psdots[dotsize=0.12](2.862,-1.3740318)
\usefont{T1}{ptm}{m}{n}
\rput(2.0054688,1.0440931){$C$}
\usefont{T1}{ptm}{m}{n}
\rput(1.2555938,-2.1960318){$X'$}
\psdots[dotsize=0.16](1.282,0.27774936)
\usefont{T1}{ptm}{m}{n}
\rput(1.0790626,-0.92412555){$u'$}
\psdots[dotsize=0.16](2.862,0.27774936)
\usefont{T1}{ptm}{m}{n}
\rput(2.0960937,0.146242){top vertex}
\usefont{T1}{ptm}{m}{n}
\rput(0.65546876,1.0440931){$D'$}
\psdots[dotsize=0.12](1.282,-0.6622506)
\psdots[dotsize=0.12](1.282,-1.3822507)
\usefont{T1}{ptm}{m}{n}
\rput(1.0790625,-1.5841255){$u''$}
\rput(3.0790625,-1.5341255){$s_2$}
\rput(3.0790625,0.56841255){$s_1$}
\usefont{T1}{ptm}{m}{n}
\rput(0.62546877,-0.13575801){\color{color116}$P$}
\end{pspicture} 
}}%
       \hspace{0.3cm}%
       \scalebox{0.8}{
\scalebox{1} 
{
\begin{pspicture}(0,-2.3675942)(2.9273126,2.3556943)
\definecolor{color202}{rgb}{1.0,0.4,0.4}
\pscustom[linewidth=0.02,linecolor=color202]
{
\newpath
\moveto(0.74,0.596242)
\lineto(0.68,0.57624197)
\curveto(0.65,0.566242)(0.6,0.541242)(0.58,0.526242)
\curveto(0.56,0.511242)(0.52,0.481242)(0.5,0.466242)
\curveto(0.48,0.451242)(0.44,0.416242)(0.42,0.396242)
\curveto(0.4,0.37624198)(0.36,0.336242)(0.34,0.31624198)
\curveto(0.32,0.296242)(0.29,0.261242)(0.28,0.24624199)
\curveto(0.27,0.231242)(0.245,0.191242)(0.23,0.16624199)
\curveto(0.215,0.141242)(0.195,0.086242)(0.19,0.056241993)
\curveto(0.185,0.026241994)(0.17,-0.033758007)(0.16,-0.06375801)
\curveto(0.15,-0.09375801)(0.14,-0.153758)(0.14,-0.183758)
\curveto(0.14,-0.213758)(0.135,-0.283758)(0.13,-0.323758)
\curveto(0.125,-0.363758)(0.12,-0.428758)(0.12,-0.453758)
\curveto(0.12,-0.478758)(0.12,-0.528758)(0.12,-0.553758)
\curveto(0.12,-0.578758)(0.125,-0.633758)(0.13,-0.663758)
\curveto(0.135,-0.693758)(0.145,-0.748758)(0.15,-0.773758)
\curveto(0.155,-0.79875803)(0.165,-0.848758)(0.17,-0.873758)
\curveto(0.175,-0.898758)(0.195,-0.948758)(0.21,-0.973758)
\curveto(0.225,-0.998758)(0.265,-1.048758)(0.29,-1.073758)
\curveto(0.315,-1.098758)(0.36,-1.1387581)(0.38,-1.153758)
\curveto(0.4,-1.168758)(0.44,-1.198758)(0.46,-1.213758)
\curveto(0.48,-1.228758)(0.53,-1.253758)(0.56,-1.2637581)
\curveto(0.59,-1.273758)(0.645,-1.293758)(0.67,-1.303758)
\curveto(0.695,-1.313758)(0.735,-1.333758)(0.75,-1.343758)
\curveto(0.765,-1.353758)(0.81,-1.368758)(0.84,-1.373758)
\curveto(0.87,-1.378758)(0.925,-1.3887581)(0.95,-1.393758)
\curveto(0.975,-1.398758)(1.005,-1.398758)(1.02,-1.383758)
}
\pscustom[linewidth=0.02,linecolor=color202]
{
\newpath
\moveto(1.0,-0.283758)
\lineto(0.95,-0.273758)
\curveto(0.925,-0.268758)(0.875,-0.268758)(0.85,-0.273758)
\curveto(0.825,-0.27875802)(0.775,-0.293758)(0.75,-0.303758)
\curveto(0.725,-0.31375802)(0.68,-0.338758)(0.66,-0.353758)
\curveto(0.64,-0.368758)(0.62,-0.408758)(0.62,-0.43375802)
\curveto(0.62,-0.458758)(0.625,-0.503758)(0.63,-0.523758)
\curveto(0.635,-0.54375803)(0.655,-0.573758)(0.67,-0.583758)
\curveto(0.685,-0.593758)(0.715,-0.613758)(0.73,-0.623758)
\curveto(0.745,-0.633758)(0.78,-0.648758)(0.8,-0.653758)
\curveto(0.82,-0.658758)(0.865,-0.663758)(0.89,-0.663758)
\curveto(0.915,-0.663758)(0.975,-0.66875803)(1.01,-0.67375803)
\curveto(1.045,-0.678758)(1.105,-0.683758)(1.13,-0.683758)
\curveto(1.155,-0.683758)(1.205,-0.688758)(1.23,-0.693758)
\curveto(1.255,-0.698758)(1.3,-0.718758)(1.32,-0.73375803)
\curveto(1.34,-0.748758)(1.375,-0.783758)(1.39,-0.803758)
\curveto(1.405,-0.823758)(1.43,-0.863758)(1.44,-0.883758)
\curveto(1.45,-0.903758)(1.46,-0.948758)(1.46,-0.973758)
\curveto(1.46,-0.998758)(1.46,-1.053758)(1.46,-1.083758)
\curveto(1.46,-1.113758)(1.455,-1.163758)(1.45,-1.183758)
\curveto(1.445,-1.203758)(1.42,-1.243758)(1.4,-1.2637581)
\curveto(1.38,-1.283758)(1.345,-1.313758)(1.33,-1.323758)
\curveto(1.315,-1.333758)(1.275,-1.353758)(1.25,-1.363758)
\curveto(1.225,-1.373758)(1.175,-1.383758)(1.15,-1.383758)
\curveto(1.125,-1.383758)(1.075,-1.383758)(1.05,-1.383758)
\curveto(1.025,-1.383758)(0.995,-1.3887581)(0.98,-1.403758)
}
\usefont{T1}{ptm}{m}{n}
\rput(2.5555937,-2.1960318){$X$}
\usefont{T1}{ptm}{m}{n}
\rput(1.7154688,1.5640931){$D$}
\rput(2.9954688,1.5640931){$D^*$}
\rput(2.9954688,1.0440931){$C^*$}
\psline[linewidth=0.056199998cm](2.581875,2.3275943)(2.581875,-1.7959068)
\usefont{T1}{ptm}{m}{n}
\rput(0.34546875,-0.13575801){\color{color202}$P$}
\usefont{T1}{ptm}{m}{n}
\rput(1.7254688,1.0440931){$C$}
\psline[linewidth=0.056199998cm](1.0,1.2759682)(1.00,-1.7959068)
\usefont{T1}{ptm}{m}{n}
\rput(0.97559375,-2.1960318){$X'$}
\psdots[dotsize=0.16](1.002,0.27774936)
\usefont{T1}{ptm}{m}{n}
\rput(0.79906255,-0.92412555){$u'$}
\usefont{T1}{ptm}{m}{n}
\rput(0.79906243,-1.5841255){$u''$}
\rput(2.8290625,-0.4341255){$s_2$}
\rput(2.8290625,0.56841255){$s_1$}
\psdots[dotsize=0.12](1.002,-0.6622506)
\psdots[dotsize=0.12](1.002,-1.3822507)
\psdots[dotsize=0.16](2.582,0.27774936)
\usefont{T1}{ptm}{m}{n}
\rput(1.8160938,0.146242){top vertex}
\usefont{T1}{ptm}{m}{n}
\rput(0.37546876,1.0440931){$D'$}
\pscustom[linewidth=0.02,linecolor=color202]
{
\newpath
\moveto(2.58,-0.283758)
\lineto(2.5273685,-0.293758)
\curveto(2.5010529,-0.298758)(2.4273684,-0.30875802)(2.38,-0.31375802)
\curveto(2.3326316,-0.318758)(2.227368,-0.328758)(2.1694736,-0.333758)
\curveto(2.111579,-0.338758)(2.0063157,-0.34375802)(1.9589472,-0.34375802)
\curveto(1.911579,-0.34375802)(1.8063159,-0.338758)(1.7484211,-0.333758)
\curveto(1.6905261,-0.328758)(1.5905261,-0.323758)(1.548421,-0.323758)
\curveto(1.506316,-0.323758)(1.4221051,-0.323758)(1.38,-0.323758)
\curveto(1.3378949,-0.323758)(1.253684,-0.323758)(1.211579,-0.323758)
\curveto(1.1694739,-0.323758)(1.1010528,-0.318758)(1.074737,-0.31375802)
\curveto(1.048421,-0.30875802)(1.011579,-0.298758)(0.98,-0.283758)
}
\pscustom[linewidth=0.02,linecolor=color202]
{
\newpath
\moveto(2.58,0.7359681)
\lineto(1.81,0.6959683)
\curveto(1.425,0.6759682)(1.015,0.65096825)(0.99,0.64596826)
\curveto(0.965,0.64096826)(0.9,0.6309683)(0.86,0.6259683)
\curveto(0.82,0.6209682)(0.77,0.61096823)(0.74,0.5959682)
}
\psdots[dotsize=0.12](2.582,0.7259682)
\psdots[dotsize=0.12](2.582,-0.2940318)
\end{pspicture} 
}\hspace{0.3cm}%
       \rput(0.1,2){\scalebox{1.7}{$\Rightarrow$}}\hspace{0.4cm}%
\scalebox{1} 
{
\begin{pspicture}(0,-2.3675942)(2.9273126,2.3556943)
\definecolor{color202}{rgb}{1.0,0.4,0.4}
\pscustom[linewidth=0.02,linecolor=color202]
{
\newpath
\moveto(1.0,-0.283758)
\lineto(0.95,-0.273758)
\curveto(0.925,-0.268758)(0.875,-0.268758)(0.85,-0.273758)
\curveto(0.825,-0.27875802)(0.775,-0.293758)(0.75,-0.303758)
\curveto(0.725,-0.31375802)(0.68,-0.338758)(0.66,-0.353758)
\curveto(0.64,-0.368758)(0.62,-0.408758)(0.62,-0.43375802)
\curveto(0.62,-0.458758)(0.625,-0.503758)(0.63,-0.523758)
\curveto(0.635,-0.54375803)(0.655,-0.573758)(0.67,-0.583758)
\curveto(0.685,-0.593758)(0.715,-0.613758)(0.73,-0.623758)
\curveto(0.745,-0.633758)(0.78,-0.648758)(0.8,-0.653758)
\curveto(0.82,-0.658758)(0.865,-0.663758)(0.89,-0.663758)
\curveto(0.915,-0.663758)(0.975,-0.66875803)(1.01,-0.67375803)
\curveto(1.045,-0.678758)(1.105,-0.683758)(1.13,-0.683758)
\curveto(1.155,-0.683758)(1.205,-0.688758)(1.23,-0.693758)
\curveto(1.255,-0.698758)(1.3,-0.718758)(1.32,-0.73375803)
\curveto(1.34,-0.748758)(1.375,-0.783758)(1.39,-0.803758)
\curveto(1.405,-0.823758)(1.43,-0.863758)(1.44,-0.883758)
\curveto(1.45,-0.903758)(1.46,-0.948758)(1.46,-0.973758)
\curveto(1.46,-0.998758)(1.46,-1.053758)(1.46,-1.083758)
\curveto(1.46,-1.113758)(1.455,-1.163758)(1.45,-1.183758)
\curveto(1.445,-1.203758)(1.42,-1.243758)(1.4,-1.2637581)
\curveto(1.38,-1.283758)(1.345,-1.313758)(1.33,-1.323758)
\curveto(1.315,-1.333758)(1.275,-1.353758)(1.25,-1.363758)
\curveto(1.225,-1.373758)(1.175,-1.383758)(1.15,-1.383758)
\curveto(1.125,-1.383758)(1.075,-1.383758)(1.05,-1.383758)
\curveto(1.025,-1.383758)(0.995,-1.3887581)(0.98,-1.403758)
}
\psframe[linewidth=0.04,linecolor=white,dimen=outer,fillstyle=solid](1.82,-0.563758)(1.0,-1.523758)
\psline[linewidth=0.056199998cm](1.0,1.2759682)(1.00,-1.7959068)
\psline[linewidth=0.056199998cm,linecolor=color202](1.0,-0.643758)(1.0,-1.403758)
\pscustom[linewidth=0.02,linecolor=color202]
{
\newpath
\moveto(0.74,0.596242)
\lineto(0.68,0.57624197)
\curveto(0.65,0.566242)(0.6,0.541242)(0.58,0.526242)
\curveto(0.56,0.511242)(0.52,0.481242)(0.5,0.466242)
\curveto(0.48,0.451242)(0.44,0.416242)(0.42,0.396242)
\curveto(0.4,0.37624198)(0.36,0.336242)(0.34,0.31624198)
\curveto(0.32,0.296242)(0.29,0.261242)(0.28,0.24624199)
\curveto(0.27,0.231242)(0.245,0.191242)(0.23,0.16624199)
\curveto(0.215,0.141242)(0.195,0.086242)(0.19,0.056241993)
\curveto(0.185,0.026241994)(0.17,-0.033758007)(0.16,-0.06375801)
\curveto(0.15,-0.09375801)(0.14,-0.153758)(0.14,-0.183758)
\curveto(0.14,-0.213758)(0.135,-0.283758)(0.13,-0.323758)
\curveto(0.125,-0.363758)(0.12,-0.428758)(0.12,-0.453758)
\curveto(0.12,-0.478758)(0.12,-0.528758)(0.12,-0.553758)
\curveto(0.12,-0.578758)(0.125,-0.633758)(0.13,-0.663758)
\curveto(0.135,-0.693758)(0.145,-0.748758)(0.15,-0.773758)
\curveto(0.155,-0.79875803)(0.165,-0.848758)(0.17,-0.873758)
\curveto(0.175,-0.898758)(0.195,-0.948758)(0.21,-0.973758)
\curveto(0.225,-0.998758)(0.265,-1.048758)(0.29,-1.073758)
\curveto(0.315,-1.098758)(0.36,-1.1387581)(0.38,-1.153758)
\curveto(0.4,-1.168758)(0.44,-1.198758)(0.46,-1.213758)
\curveto(0.48,-1.228758)(0.53,-1.253758)(0.56,-1.2637581)
\curveto(0.59,-1.273758)(0.645,-1.293758)(0.67,-1.303758)
\curveto(0.695,-1.313758)(0.735,-1.333758)(0.75,-1.343758)
\curveto(0.765,-1.353758)(0.81,-1.368758)(0.84,-1.373758)
\curveto(0.87,-1.378758)(0.925,-1.3887581)(0.95,-1.393758)
\curveto(0.975,-1.398758)(1.005,-1.398758)(1.02,-1.383758)
}
\usefont{T1}{ptm}{m}{n}
\rput(2.5555937,-2.1960318){$X$}
\usefont{T1}{ptm}{m}{n}
\rput(1.7154688,1.5640931){$D$}
\rput(2.9954688,1.5640931){$D^*$}
\rput(2.9954688,1.0440931){$C^*$}
\psline[linewidth=0.056199998cm](2.581875,2.3275943)(2.581875,-1.7959068)
\usefont{T1}{ptm}{m}{n}
\rput(0.34546875,-0.13575801){\color{color202}$P$}
\usefont{T1}{ptm}{m}{n}
\rput(1.7254688,1.0440931){$C$}
\usefont{T1}{ptm}{m}{n}
\rput(0.97559375,-2.1960318){$X'$}
\psdots[dotsize=0.16](1.002,0.27774936)
\usefont{T1}{ptm}{m}{n}
\rput(0.79906255,-0.92412555){$u'$}
\usefont{T1}{ptm}{m}{n}
\rput(0.79906243,-1.5841255){$u''$}
\rput(2.8290625,-0.4341255){$s_2$}
\rput(2.8290625,0.56841255){$s_1$}
\psdots[dotsize=0.12](1.002,-0.6622506)
\psdots[dotsize=0.12](1.002,-1.3822507)
\psdots[dotsize=0.16](2.582,0.27774936)
\usefont{T1}{ptm}{m}{n}
\rput(1.8160938,0.146242){top vertex}
\usefont{T1}{ptm}{m}{n}
\rput(0.37546876,1.0440931){$D'$}
\pscustom[linewidth=0.02,linecolor=color202]
{
\newpath
\moveto(2.58,-0.283758)
\lineto(2.5273685,-0.293758)
\curveto(2.5010529,-0.298758)(2.4273684,-0.30875802)(2.38,-0.31375802)
\curveto(2.3326316,-0.318758)(2.227368,-0.328758)(2.1694736,-0.333758)
\curveto(2.111579,-0.338758)(2.0063157,-0.34375802)(1.9589472,-0.34375802)
\curveto(1.911579,-0.34375802)(1.8063159,-0.338758)(1.7484211,-0.333758)
\curveto(1.6905261,-0.328758)(1.5905261,-0.323758)(1.548421,-0.323758)
\curveto(1.506316,-0.323758)(1.4221051,-0.323758)(1.38,-0.323758)
\curveto(1.3378949,-0.323758)(1.253684,-0.323758)(1.211579,-0.323758)
\curveto(1.1694739,-0.323758)(1.1010528,-0.318758)(1.074737,-0.31375802)
\curveto(1.048421,-0.30875802)(1.011579,-0.298758)(0.98,-0.283758)
}
\pscustom[linewidth=0.02,linecolor=color202]
{
\newpath
\moveto(2.58,0.7359681)
\lineto(1.81,0.6959683)
\curveto(1.425,0.6759682)(1.015,0.65096825)(0.99,0.64596826)
\curveto(0.965,0.64096826)(0.9,0.6309683)(0.86,0.6259683)
\curveto(0.82,0.6209682)(0.77,0.61096823)(0.74,0.5959682)
}
\psdots[dotsize=0.12](2.582,0.7259682)
\psdots[dotsize=0.12](2.582,-0.2940318)
\end{pspicture} 
}}%
    \end{center}
    \caption{The replacement of a pseudo shortcut $P$ crossing
            $X'$ more than once.
             }
    \label{fig:crestSho}
  \end{figure}

\begin{lemma}\label{lem:StrongOptimal}
Let
$D$ and $D^*$ be the $(X,\varphi)$-components for a
crest separator $X\in\mathcal{S}$ such that $D^*$ is not enclosed by $X$
or $X$ has an exterior lowpoint.
Moreover, let $P$ either 
be 
\begin{itemize}
\item[(a)] 
the part contained in $D$ of a
minimal coast separator $Y$ for
a crest $H\in \mathcal H$
 in $D^*$ or
\item[(b)] a strict
$\DD$-pseudo shortcut of a crest-separator path $X^{\mathrm{CP}}$ of
 $X$.
\end{itemize}
Then, 
for all crest separators $X'=(P'_1,P'_2) \in\mathcal{S}$ with
an $(X',\varphi)$-component $\DD'\subseteq \DD$, possibly $X'=X$,
$P$ has at most one 
maximal \spp{\DD'}{X'}-subpath. If
$P$
has such a path $P^*$, $P^*$ 
is 
a 
strict $\DD'$-pseudo shortcut of 
the crest-separator path $(X')^{\mathrm{CP}}$ of~$X'$ 
such that $(X')^{\mathrm{CP}}$ has the same endpoints as $P^*$ and is contained
\mbox{in the inner graph ${I}$ of}
\begin{itemize}
\item 
the cycle induced by the vertex set of~$Y$ (Case (a)) or
\item the composed cycle 
 $Q$ of~$(X^{\mathrm{CP}},P)$ (Case (b)).
\end{itemize}
\end{lemma}

\begin{proof}
Roughly speaking, we first want to show that, for all $j\in\{1,2\}$, 
there is at most one crossing between $P$ and $P'_j$.
In fact, we want to show something more with respect to two concerns.
\begin{enumerate}
\item If $P$ starts in a common vertex of~$X$ and $X'$, 
we also want to consider this 
entering of~$D'$ 
as a crossing vertex. Therefore, we let $\tilde{P}$ be a path obtained
from $P$ by adding two edges of~$D^*$ not being border edges of~$X$ at
the beginning and the ending of~$P$.
\item 
We
consider one of
possible several crossings of 
$\tilde{P}$ and $P'_j$ and we mark all crossing vertices that 
belong to exactly this crossing. 
More precisely, we choose 
the crossing whose crossing vertices are the last crossing
vertices on $P'_j$. 
Note, that the marked vertices induce a subpath of~$P'_j$. 
In the next paragraph, we show that no other vertex
of~$\tilde{P}$ can appear before the marked vertices on $P'_j$. This implicitly
shows that, for each $j'\in\{1,2\}$, 
there can be at most one crossing between $\tilde{P}$ and $P'_j$
and, more important, that there can be only at most one 
maximal \spp{\DD'}{X'}-subpath.
\end{enumerate}

Let $v$ be the last marked vertex on 
$P'_j$.
Since the coast is not part of~$I$, all vertices after $v$ on $P'_j$
are not part of~$I$. Hence, if $P'_j$ contains vertices of~$\tilde{P}$
before the marked vertices, then the subpath of~$P'_j$ 
between the last such vertex $u'$ and the first marked vertex
$u''$
is contained in $I$ and can be replaced by the 
crest-separator path
of~$X'$ between $u'$ and $u''$ (see Fig.~\ref{fig:crestSho} for two 
possible examples).
This replacement 
does not increase the weighted length of~$\tilde{P}$/of~$Y$
(Lemma~\ref{lem:ShortAlley}(a)), but reduces the number of inner faces of
${I}$; which is a contradiction to the definition of~$\tilde{P}$ being
a
strict
pseudo shortcut or being part of a minimal coast separator. 
Thus, our assumption that there are vertices of~$\tilde{P}$ before the 
marked vertices on $P'_j$ is wrong.

We next conclude that, if there is a 
maximal \spp{\DD'}{X'}-subpath $P'$, since it is the only one and since
$P$ is a strict pseudo shortcut or part of a minimal coast separator,
$P'$ must have a shorter weighted length than the 
crest-separator path
of~$X'$
that is part of~$I$ and between the 
endpoints of~$P'$. Note also that, if $X'$ has an interior
lowpoint, it cannot enclose the
$(X',\varphi)$-component opposite to $\DD'$ since, otherwise, $X$ must
also have an interior lowpoint and must enclose $\DD^*$.
Hence $P'$ is a pseudo-shortcut of 
$(X')^{\mathrm{CP}}$. Since 
$\tilde{P}$ is a strict pseudo shortcut or part of a minimal coast separator,
$P'$ as the only maximal \spp{\DD'}{X'}-subpath must be 
strict.
\markatright~\end{proof}

In the following we want to compute pseudo shortcuts for
the different crest separators by a bottom-up traversal
in the mountain connection tree $T$ of 
$(G,\varphi,c,{\mathcal H},\mathcal{S})$.
Therefore, let us assume that $T$ is rooted at some arbitrary node.
Let $X$ be
a crest separator
going
weakly between an $(\mathcal{S},\varphi)$-component $\CC$
and its parent in $T$, and
let $\DD$ be
the ($X,\varphi)$-component $\DD$ containing $\CC$.
The idea to construct
a
$\DD$-pseudo shortcut of~$X$ can be
described
as follows:
A $\DD$-pseudo shortcut first follows a path in $\CC$ and possibly, after
reaching a vertex $v$ of~$\CC$ that is part of a crest separator
$X'\in\mathcal{S}$ weakly going between $\CC$ and a child
$\CC'$ of~$\CC$ in $T$, it follows a (precomputed) $\DD'$-pseudo shortcut
of~$X'$
for the $(X',\varphi)$-component $D'$ containing $C'$, then it follows
again a path in $\CC$
and, after possibly containing further pseudo shortcuts, it returns
to $\CC$ and never leaves $\CC$ anymore.
Indeed, if a $D$-pseudo shortcut of~$X$, immediately after reaching a
vertex $v$ of~$\CC$
that is part of a crest separator $X'$
with the properties described above,
contains
an edge outside the $(X',\varphi)$-component containing $C$,
then we show that $P$ contains a
$\DD'$-pseudo shortcut $P'$
for the
 $(X',\varphi)$-component $D'$
containing $C'$. However,
since 
$X$ may not completely contained in $C$, a
$D$-pseudo shortcut of~$X$ may not start in a vertex of~$\CC$.
This
is the reason why, for subsets
$\mathcal{S}',\mathcal{S}''\subseteq\mathcal{S}(G,\varphi,c)$
with $\mathcal{S}'\subseteq\mathcal{S}''$
(an example for~$S'\neq S''$ can be found in Lemma~\ref{lem:ExtCrest})
and for an $(\mathcal{S'},\varphi)$-component
$C$, we define
the {\em extended
component}
$\mathrm{ext}(C,\mathcal{S''})$
as the plane graph obtained from
$C$ by adding the border edges of all
crest separators $X\in\mathcal{S''}$
with a top edge in $C$.
This should mean of course
that the endpoints of the border edges are also added
as vertices to~$C$.
As embedding of~$\mathrm{ext}(C,\mathcal{S''})$, we always take
$\varphi|_{\mathrm{ext}(C,\mathcal{S''})}$.

The next three lemmas prove
some properties of extended components that 
allows us to guarantee the existence of pseudo shortcuts with
nice properties 
in Lemma~\ref{lem:PSCComputing} from which we show
in Lemma~\ref{lem:ShortCutSet} that they can be constructed efficiently.

\begin{lemma}\label{lem:FirstEdge}
Let
$C$ be an $(\mathcal{S},\varphi)$-component
and let
$e$ be an edge with exactly one endpoint
$v$
in
$\mathrm{ext}(C,\mathcal{S})$.
Then, $v$ is part of
a crest separator in $\mathcal{S}$ with a top~edge~in~$C$.
\end{lemma}

\begin{proof}
By the definition of~$\mathrm{ext}(C,\mathcal{S})$ the
lemma holds if $v$ is not contained in $C$. It
remains to consider the case that $v$ is in $C$.
The definition of an
$(\mathcal{S},\varphi)$-component
implies that
$v$ is a vertex of a crest separator in $\mathcal{S}$.
We define the {\em boundary} of~$C$
to be the graph that consists of 
the vertices and edges of~$C$ that are incident to a
face~$f$
of~$\varphi|_C$ such that $f$ is not a face of
$\varphi$. (Roughly speaking, $f$ is the union of several faces of~$\varphi$.)

Let $u$
be a vertex of largest
upper height
such that there is a
down path
from $u$ to $v$ that
is contained in the boundary of $C$.
Since $G$ is biconnected, Obs.~\ref{obs:biconnComp}
implies that $C$ is biconnected. Thus,
$u$ is incident
to two edges $\{u,v_1\}$ and $\{u,v_2\}$ on the boundary of~$C$.
Assume for a moment that both edges are down edges.  Since each vertex is connected by
down edges to
at most one vertex of smaller upper height and since down edges only connect
vertices of different upper heights, 
either $v_1$ or~$v_2$ must have larger upper height than $u$. This is a contradiction to our
choice of~$u$.
Consequently, one of~$\{u,v_1\}$ and $\{u,v_2\}$ 
is a top edge in $C$ that belongs to  
a crest separator $X\in\mathcal{S}$ with a top edge in $C$. Since the down path from $u$ contains
$v$, it follows that 
$X$
contains $v$.
\end{proof}

For the next two lemmas, let $\DD$ and $\tilde{\DD}$ be
the two $(X,\varphi)$-components
of a crest separator $X$
in
$\mathcal{S}(G,\varphi,c)$,
and let $\mathcal{S}$ be a set of crest separators
with $\{X\} \subseteq \mathcal{S}\subseteq \mathcal{S}(G,\varphi,c)$.

\begin{lemma}\label{lem:ExtCrest}
The extended component $\mathrm{ext}(\DD,\mathcal{S})$
consists
only of
the vertices and edges of~$\DD$,
the border edges of~$X$, and their endpoints.
\end{lemma}

\begin{proof}
Each vertex $v\in\mathrm{ext}(\DD,\mathcal{S})$ either belongs to $\DD$ or
is an endpoint of a border edge of a crest separator with a top
edge $\{v',v''\}$ in $\DD$, i.e., $v$ is
reachable by a down path from vertex $v'$ or $v''$.
A down path
$P$ starting in a vertex of~$D$ can
leave $\DD$ only after reaching a vertex $x \in X$. But
after reaching $x$, $P$ must follow the down path of~$x$ and therefore
all edges after $x$ on $P$ must be border edges of~$X$.
\end{proof}

\begin{lemma}\label{lem:CContained}
If $e$ is an edge
with exactly one endpoint $v$ in $\mathrm{ext}(\DD,\mathcal{S})$,
then $v$ is part of~$X$.
\end{lemma}

\begin{proof}
Note that there are no
direct edges from a vertex $v\in\DD$ to a vertex $\tilde{v}\in\tilde{\DD}$
with neither $v$ nor $\tilde{v}$ being part of~$X$. Hence, $v$ must be part of
$X$ by Lemma~\ref{lem:ExtCrest}.
\end{proof}

As mentioned above, we want to precompute pseudo shortcuts for the 
crest separators in $\mathcal{S}$ since we later want to use
them to construct coast separators. This only works if the pseudo
shortcuts have some nice properties 
(%
in some kind similar
to the properties of strict pseudo
shortcuts, but the strict pseudo shortcuts have the problem 
that they cannot be computed efficiently). Therefore, let us consider
a crest separator $X_0$ and an $(X_0,\varphi)$-component $\DD$ of
$X_0$. Moreover, let us root the mountain connection
tree $T$ of~$(G,\varphi,c,{\mathcal H},\mathcal{S})$ such that
for the $(\mathcal{S},\varphi)$-components $\CC$ and $\CC_0$ containing 
the top edge of~$X_0$, the $(\mathcal{S},\varphi)$-component $\CC_0$ 
not being contained in $\DD$ 
is the parent of the other $(\mathcal{S},\varphi)$-component $\CC$. 
Moreover, let   
$\CC_1,\ldots,\CC_j$ be the children of~$C$ in $T$ and,
for
$i\in\{1,\ldots,j\}$, $X_i$ be the crest separator
with a top edge in~$\CC$ and $\CC_i$. 
Finally,
let $\DD_i$ ($i\in\{1,\ldots,j\}$) be the
$(X_i,\varphi)$-component that contains $\CC_i$.
Using these definitions we define a new kind of pseudo shortcuts.

\begin{definition}[nice pseudo shortcut]
\label{def:nice}
An $s_1$-$s_2$-connecting $\DD$-pseudo shortcut $P$ 
of~$X_0$ is called {\em
nice} if it has shortest weighted length among all 
$s_1$-$s_2$-connecting $\DD$-pseudo shortcuts of~$X_0$ and if
it
consists exclusively of subpaths in
$\mathrm{ext}(C,\mathcal{S})$ and, 
for each
$i\in\{1,\ldots,j\}$, 
of at most one 
nice $\DD_i$-pseudo shortcut $P_i$ for some 
crest-separator path
$X_i^{CP}$ of~$X_i$ that is part of the inner graph of the composed cycle of
$(X_0^{CP},P)$. 
\end{definition}

The following lemma guarantees the existence
of nice pseudo shortcuts.

\begin{lemma}\label{lem:PSCComputing}
If there exists
 an
$s_1$-$s_2$-connecting
 $\DD^*$-pseudo shortcut for the
$(X^*,\varphi)$-component $\DD^*$ of a crest separator 
$X^*\in\mathcal{S}$,
then there is also such a pseudo shortcut that is nice.
\end{lemma}

\begin{proof}
First, observe that with an $s_1$-$s_2$-connecting $\DD^*$-pseudo shortcut of~$X^*$,
there is also a strict $s_1$-$s_2$-connecting $\DD^*$-pseudo shortcut
of~$X^*$.
We next show that, for each crest separator $X_0$
and each $(X_0,\varphi)$-component $\DD$ of~$X_0$, each 
strict $\DD$-pseudo shortcut $P$ is nice. This is shown
by
induction 
over the number of crest separators $X\in\mathcal{S}$ for which there is a crossing of~$X$ and $P$.
Let us define $\CC,\CC_1,\ldots, \CC_j$, $X_0,X_1,\ldots,X_j$, and
$\DD,\DD_1,\ldots,\DD_j$ as strictly before Definition~\ref{def:nice}.
If there are no 
crossings between $P$ and crest separators $X\neq X_0$ in $\mathcal{S}$,
$P$ is just a
pseudo shortcut of shortest weighted length connecting the endpoints of~$P$
and is contained in $\mathrm{ext}(C,\mathcal{S})$. Then
$P$ is clearly nice.
Otherwise,
let us consider the first edge $e$
of~$P$ not contained in $\mathrm{ext}(\CC,\mathcal{S})$.
Hence, one
endpoint $v'$
of~$e$ must be part of
a crest separator $X_i$ with $i\in\{1,\ldots,j\}$%
---here we use 
Lemma~\ref{lem:FirstEdge} and the fact that
$X_0,X_1,\ldots,X_j$ are the only
crest separators with a
top edge in $C$; $i\neq 0$ since $e$ is in $D$---%
and
the other endpoint
is contained in the $(X_i,\varphi)$-component opposite to $\DD_i$.
Thus, there must be a vertex $v''$ after $v'$ on $P$ that is contained in
$\mathrm{ext}(C,\mathcal{S})$.
W.l.o.g., let $v''$ be the first such vertex. 
Because of 
Lemma~\ref{lem:CContained},
$v''$
must be also part
of~$X_i$. 
Since
$P$ is a
strict
$\DD$-pseudo shortcut of~$X_0$,
by Lemma~\ref{lem:StrongOptimal}, it has
at most one maximal $\DD_i$-subpath $P_i$
from
$v'$ to $v''$, and it
must be
a strict pseudo shortcut of the 
crest-separator path
$X_i^{\mathrm{CP}}$ of~$X_i$ with
endpoints~$v'$ and $v''$ such that $X_i^{\mathrm{CP}}$
is
contained in the composed cycle $(X_0^{\mathrm{CP}},P)$. Since the number of crest separators for
which there is a crossing of the crest separator and 
$P_i$
is smaller than 
the 
corresponding number for the whole path
$P$, we can conclude
that the subpath from $v'$ to $v''$ is nice.
If there are further parts of~$P$ not contained in
$\mathrm{ext}(\CC,\mathcal{S})$,
they can similarly shown to be
the only
 and nice
pseudo shortcut
for one of the other crest separators in
$\{X_1,\ldots,X_j\}\setminus\{X_i\}$. Together with the fact
that $P$ as a strict pseudo shortcut is a $\DD$-pseudo shortcut of 
shortest weighted length between its endpoints, we can conclude that $P$ is nice.
\end{proof}

For an $(X,\varphi)$-component
$\DD$ of a crest separator $X\in\mathcal{S}$,
let us define a {\em $d$-bounded $\DD$-pseudo shortcut set for} $X$
 to be a set
consisting of an 
$s_1$-$s_2$-connecting nice $\DD$-pseudo shortcut
 $P$ of weighted length at most $d$
for each pair $s_1$ and $s_2$
of vertices of~$X$
for which such an
$s_1$-$s_2$-connecting 
$\DD$-pseudo shortcut exists. The main aim of 
Lemma~\ref{lem:ShortCutSet} is to efficiently
construct
 a $d$-bounded $\mathcal{M}$-shortcut~set.

\begin{definition}[$d$-bounded $\mathcal{M}$-shortcut set]
For a mountain structure 
$\mathcal{M}=(G,\varphi,c,\mathcal{H},\mathcal{S})$, 
a {\em $d$-bounded $\mathcal{M}$-shortcut set} is a set consisting
of the union of a $d$-bounded $D$-pseudo shortcut set for each
crest separator $X$ and each $(X,\varphi)$-component $D$.
\end{definition}

 The pseudo shortcuts
in such a %
$d$-bounded $\mathcal{M}$-shortcut set
are later used for the construction of coast separators. However,
in order to avoid an ``overlapping'' of coast separators, we will remove some
of the constructed pseudo shortcuts. For making such a removal more efficient,
we have to store some additional informations and to introduce some
further definitions. 

An $(\mathcal{S},\varphi)$-component $C$ is called the {\em root component} of
a $\DD$-pseudo shortcut $P$ of a crest separator $X\in\mathcal{S}$ if
$C$ is the one of the two
$(\mathcal{S},\varphi)$-components containing a top edge of $X$ that is
contained in $\DD$. 
A set $\mathcal{P}$ of pseudo shortcuts  of crest separators in $\mathcal{S}$
is 
{\em consistent w.r.t.\ $\mathcal{S}$} if, for
each %
pseudo shortcut $P$ %
contained in $\mathcal{P}$ the
following holds: 
If $P$ has a subpath $P'$ such that (1) $P'$ is a pseudo shortcut
of a crest separator $X'\in\mathcal{S}$ and (2) the vertex sets
of $P'$ and $X$ are weakly disconnected by the vertex set of $X'$, 
then $P'$ is also contained in $\mathcal{P}$.
Intuitively speaking, condition (2) guarantees that $P'$
is not a pseudo shortcut on the wrong side of $X'$.
If $\mathcal{S}$ is clear from the context, we call
$\mathcal{P}$ simply consistent.
The {\em consistence graph} of a consistent set $\mathcal{P}$ 
of pseudo shortcuts is a directed graph whose node set 
consists of the pseudo 
shortcuts $P\in\mathcal{P}$ %
and whose edge set has
an edge %
from a pseudo shortcut $P\in \mathcal{P}$ 
to a pseudo shortcut 
$P'\in\mathcal{P}$ of a crest separator $X'\neq X$ %
if and only if (1) $P'$ is
a subpath of $P$ and (2) the root component of $P'$ is a neighbor of the root
component of $P$ in the mountain connection tree.

For the next lemma 
keep in mind that $G$ is %
weighted $\ell$-outerplanar;
in particular,
all
crest separators
have weighted length at most $2\ell$.

\begin{lemma}\label{lem:ShortCutSet}
Let $d\le\ell$ and $\mathcal{M}=(G,\varphi,c,\mathcal{H},\mathcal{S})$. Then 
a consistent $d$-bounded $\mathcal{M}$-shortcut 
set $\mathcal{P}$ and the consistence graph of $\mathcal{P}$ can
be constructed in $O(|{\mathcal{H}}|\ell^3 + |V|\ell )$-time.
Within the same time  one can determine the root component of each pseudo
shortcut in $\mathcal{P}$.
\end{lemma}

\begin{proof}
We already know that we can compute the mountain connection tree $T$
of~$(G,\varphi,c,{\mathcal H},\mathcal{S})$ in $O(|V| {\ell})$ time (Lemma~\ref{lem:mct1}). 
We then root $T$ and
compute the pseudo shortcut sets in
a bottom-up traversal of~$T$ followed by a top-down
traversal. Roughly speaking, the bottom-up traversal computes the
pseudo shortcut sets for the $(X,\varphi)$-components below
crest separators $X$ whereas the top-down traversal computes
the pseudo shortcut sets for the $(X,\varphi)$-components above crest separators $X$.
More precisely, let us assume that in the bottom-up traversal of 
the mountain connection tree, we want to compute a $\DD$-pseudo shortcut set
for a crest-separator $X_0$ such that $\DD$ consists of the union of 
$(\mathcal{S},\varphi)$-components of a complete subtree of the
mountain connection tree. Let us define 
$\CC,\CC_0,\CC_1,\ldots,\CC_j$,
$X_0,\ldots,X_j$,
and
 $\DD, \DD_1,\ldots,\DD_j$ in the same way as 
as strictly before Definition~\ref{def:nice}.
Then
we
can assume that we have already computed a $\DD_i$-pseudo shortcut set named
$\mathcal{L}({\DD_i})$ for~$X_i$ for all $i\in\{1,\ldots,j\}$.

If $X_0$ has an interior lowpoint and does not enclose $C$, by definition of the
pseudo shortcuts, there are no $D$-pseudo shortcuts of $X_0$, i.e., the $\DD$-pseudo
shortcut set $\mathcal{L}(\DD)$ for~$X_0$ is empty.
Otherwise, we
determine the set by
using
a single-source shortest path algorithm for each vertex $s$ of~$X_0$ as
source vertex
on the vertex-and-edge-weighted graph
$\CC'$ described in the next paragraph.
In contrast to the rest of the paper, in this proof we consider 
also graphs with both, vertex and edge weights.
We define
the {\em weighted length} of a path 
in a graph with vertex and edge weights
as the total sum of the weights of the vertices and edges of the path.
In addition, the {\em distance} of two vertices in a graph (with or
without edge weights) is the minimal weighted length of a path
connecting the vertices.
Note that the
classical Dijkstra-algorithm \cite{EWD:NumerMath59} can be easily
modified to compute the distances from a source vertex to all other
vertices in a graph with vertex and edges weights
in time linear to the size of the given graph plus the maximal
number of different
considered distances during the computation.
In the same time, the algorithm can additionally compute a so-called {\em shortest-path tree}
that, given a vertex $t$,
allows us to find a path 
of shortest weighted length
connecting the source vertex and $t$
in a time linear to the
number of vertices of the path.

The graph $\CC'$ is obtained from
$\mathrm{ext}(\CC,\mathcal{S})$
by deleting all vertices 
of the coast
and, for each $i\in \{1,\ldots,j\}$ and each nice pseudo
shortcut $P$ in $\mathcal{L}(\DD_i)$,
inserting  an edge~$e_P$
connecting the endpoints of~$P$.
We say that $e_P$
{\em represents} $P$.
The weight of a vertex is equal to its {\cost} in $G$.
Whereas the weight of each edge of~$\mathrm{ext}(\CC,\mathcal{S})$ is zero, the
weighted length for an edge representing
a pseudo shortcut is equal to the weighted length of the pseudo shortcut
plus $\epsilon$
minus
the {\cost}s of its endpoints, where $\epsilon=1/(2\ell)$ is a small penalty
that makes a shortest weighted path with many pseudo shortcuts a little bit
more expensive than a shortest weighted path with fewer pseudo shortcuts.
We subtract the weight of its endpoints since they are already 
taken into account as weights of the endpoints.
Note that, for each $D$-pseudo shortcut
$P$ of~$X_0$,
the integer parts of the distances of the two endpoints of~$P$
in $C'$ and in $D$ are the same (Lemma~\ref{lem:PSCComputing})
since $P$ is of length $\le (2\ell-1)$ so that 
the sum of the fractional parts is 
of size $\le (2\ell-1) \epsilon<1$.

Hence, after running 
the modified Dijkstra-algorithm $O(\ell)$ times, we know the distance
in $D$ of each pair of
vertices $s$ and $t$ of~$X_0$. 
Thus, we
can test 
if
their distance 
is shorter
than %
the weighted length of a
crest-separator path of~$X_0$ from~$s$ to~$t$ and
if the integer part of their distance is at most $d$.
If both is true, there is a pseudo shortcut connecting $s$ and $t$ of length at most
$d$. 
Then, using the shortest path tree we compute such a pseudo shortcut
$P^*$
from a path of shortest
weighted length in $C'$ by
replacing each edge $e_P=\{x,y\}$ representing a
pseudo shortcut $P$ by $P$ itself. Thereby we also add an edge $(P^*,P)$ into
the consistence graph of $\mathcal{P}$.
Finally, we add
$P^*$ to
the pseudo shortcut set $\mathcal{L}(\DD)$
and store $C$ as the root component of $P^*$.

Before analyzing the running time,
let us define $m_{C'}$ to be
the number of edges in $\CC'$ and 
$n_C$ to be the number of vertices in $\CC$.
Since all pseudo shortcuts have a weighted length less than 
$2{\ell}$---they are shorter than the weighted length of 
their crest separators---we can terminate every single-source shortest path computation
after the computation of
all
vertices for which the distance to the source vertex is smaller
than
$2{\ell}$. Since
there are
$(2\ell)^2$
possible distances values to consider, namely
$0,\epsilon,\ldots,(2\ell-1)\epsilon,1,1+\epsilon,\ldots,2\ell-1+(2\ell-1)\epsilon$,
and since there are
at most ${\ell}^2$ edges introduced for the pseudo shortcuts
of each crest separator in $\mathcal{S}$, each of the $O({\ell})$
single source-shortest paths problems on $C'$ can be solved
in 
$O(m_{C'}+\ell^2)=O(n_C+(j+1){\ell}^2)$ 
where $j$ is the degree $\deg_T(w)$ for
the node $w$ in $T$ identified with $C$. Note that
the number of nodes in $T$ is $O(|{\mathcal{H}}|)$.
Therefore,
the whole bottom-up traversal
runs in
$O(|{\mathcal{H}}|\ell^3 + |V|\ell)$
time.

Afterwards in a top-down traversal, we consider 
each node
$C$ of~$T$ with
its crest separators $X_0,\ldots,X_j$ defined as before.
Let $\DD'$ be the $(X_0,\varphi)$-component not containing $C$, which is
opposite to $\DD$, and
let $\DD'_i$ ($i\in \{1,\ldots,j\}$) be the $(X_i,\varphi)$-component containing
$C$, which is opposite to $\DD_i$. 
For each such crest separator $X_i \in \{X_1,\ldots,X_j\}$,
 we can compute a
$\DD'_i$-pseudo shortcut set
as follows. 
We first determine $q\in\{1,\ldots,j\}$ such that $X_q$ is a crest separator 
with the largest weighted length among all crest
separators in $\{ X \in \{X_1,\ldots,X_j\} \mid (X $
has no lowpoint) or $(X$ has an exterior lowpoint) or $(X$ encloses $C)\}$.
Intuitively, $X_q$ is the crest separator of largest weighted
length among the crest separators in $\{X_1,\ldots,X_j\}$ for which we 
need to compute
pseudo shortcuts.

Then, we compute the pseudo shortcut set $\mathcal{L}(\DD'_q)$
  for~$X_q$ 
analogously as described in the bottom-up traversal.
A similar
computation for the remaining crest separators 
$X_i\in \{X_1,\ldots,X_j\}/\{X_q\}$ would be correct, however,  
we have to proceed
differently to guarantee the running time of the lemma. 
Let $C'$ be the graph obtained from
$\mathrm{ext}(\CC,\mathcal{S})$ by 
deleting all vertices of the coast and
adding edges $e_P$
for all pseudo shortcuts $P$ in 
$\mathcal{L}(D')\,\cup\, \bigcup_{i\in\{1,\ldots,j\} \setminus \{q\}}
\mathcal{L}(D_i)$ such that $e_P$ connects the endpoints of~$P$, where $\mathcal{L}(D')$ is the empty set if $C$ is the
root of the mountain connection tree.
Note that a non-empty pseudo shortcut set $\mathcal{L}(D')$ was computed in
the previous step of the top-down traversal whereas all remaining 
pseudo shortcut sets 
were already computed
during the bottom-up traversal. 
Assign a weight to each such edge $e_P$ that is equal to the
weighted length 
of~$P$ plus $\epsilon$
minus
the {\cost}s of its endpoints, where $\epsilon=1/(2\ell)$.
All other edges of~$C'$ have weight $0$ and the vertices $v$ of~$C'$ have weight
$c(v)$; thus, equal to their weights in $G$.
In $C'$, determine 
the distance $d(x,y)$ of~$x$ and $y$ 
for
all vertices 
$x\in X_q$~and~%
$y$~of~$C'$.

After the computation of these distances we can determine
the $\DD'_i$-pseudo shortcut set in a reduced subgraph $G_i$ for
each $i\in
\{1,\ldots,j\}\setminus \{q\}$.
For its definition, let us consider the $(\{X_i,X_q\},\varphi)$-component $C^*_i$
that contains the top edges of~$X_i$ and $X_q$. 
If we remove all vertices of the coast from
the plane graph 
$\mathrm{ext}(C^*_i,\mathcal{S})$, the resulting graph is
divided 
into two sides by
two 
ridges 
connecting the crest
$H\in {\mathcal H}$
contained in $C$ with the crests 
$H_i\in {\mathcal H}$
and 
$H_q\in {\mathcal H}$
of
$C_i$ and $C_q$,
respectively. 
Vertices of the ridge belong to both sides. More formally, a {\em side}
is the graph induced by a maximal set of vertices in
$\mathrm{ext}(C^*_i,\mathcal{S})$ that do not belong to the coast and that
are not weakly
separated by the vertices part of the two
ridges.
For an illustration, see Fig.~\ref{fig:pshortcuts}.
$G_i$ contains the vertices and the edges of~$X_i$ and $X_q$ except the vertices belonging
to the coast.
In addition, $G_i$ has
edges between each vertex $x$ on $X_q$ 
and each vertex $y$ of
$X_i$
with $x$ and $y$ being part of the same side of~$\mathrm{ext}(C^*_i,\mathcal{S})$.
We assign to each such edge $\{x,y\}$ a weight 
$d_{G_i}(x,y)=d(x,y)-c(x)-c(y)$.
Recall that $D_q$ is the $(X_q,\varphi)$-component opposite to $D'_q$.
Finally for each
pair of vertices $x$ and $y$ of~$X_q$, 
$G_i$ has an edge $\{x,y\}$ with weighted
length $d_{G_i}(x,y)=d_{X_q}(x,y)-c(x)-c(y)$,
where
 $d_{X_q}(x,y)$ is either
$\epsilon$ plus 
the weighted length of a 
$D_q$-pseudo shortcut of~$X_q$,
if it exists, or otherwise, the weighted length of the
crest-separator path
of~$X_q$ from~$x$~to~$y$.

  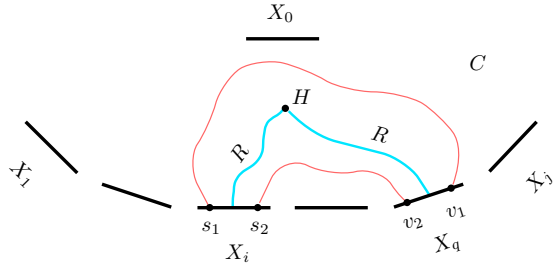
\begin{figure}[b!]
    \begin{center}
      \scalebox{0.8}{
\scalebox{1} 
{
\begin{pspicture}(0,-1.8092187)(9.607312,1.9092187)
\definecolor{color740}{rgb}{0.0,0.9019607843137255,1.0}
\definecolor{color3014}{rgb}{1.0,0.4,0.4}
\pscustom[linewidth=0.02,linecolor=color3014]
{
\newpath
\moveto(3.5282407,-1.2226175)
\lineto(3.3282406,-0.8386173)
\curveto(3.2282407,-0.64661777)(3.2282407,-0.31061774)(3.3282406,-0.1666178)
\curveto(3.4282405,-0.02261755)(3.5782406,0.3133822)(3.6282406,0.5053823)
\curveto(3.6782405,0.69738245)(3.8782406,0.9853826)(4.0282407,1.0813828)
\curveto(4.178241,1.1773825)(4.5282407,1.2253823)(4.7282405,1.1773825)
\curveto(4.928241,1.1293826)(5.2782407,0.9853823)(5.428241,0.8893826)
\curveto(5.5782404,0.7933825)(5.9782405,0.69738245)(6.2282405,0.69738245)
\curveto(6.4782405,0.69738245)(6.8782406,0.6013824)(7.0282407,0.5053823)
\curveto(7.178241,0.40938258)(7.428241,0.16938257)(7.5282407,0.025382327)
\curveto(7.6282406,-0.11861761)(7.678241,-0.4066175)(7.5282407,-0.83861756)
}
\pscustom[linewidth=0.02,linecolor=color3014]
{
\newpath
\moveto(4.348241,-1.1426176)
\lineto(4.5450406,-0.7292844)
\curveto(4.6434407,-0.5226176)(4.9386406,-0.41928443)(5.1354403,-0.5226176)
\curveto(5.332241,-0.6259513)(5.7258406,-0.67761755)(5.922641,-0.6259513)
\curveto(6.1194406,-0.57428443)(6.414641,-0.67761755)(6.5130405,-0.8326175)
\curveto(6.61144,-0.98761755)(6.7590404,-1.1426176)(6.8082404,-1.1426176)
\curveto(6.857441,-1.1426176)(6.857441,-1.1426176)(6.709841,-1.1426176)
}
\pscustom[linewidth=0.04,linecolor=color740]
{
\newpath
\moveto(3.9082406,-1.2426176)
\lineto(3.9082406,-1.1491996)
\curveto(3.9082406,-1.1024905)(3.9132407,-1.0298331)(3.9182405,-1.0038834)
\curveto(3.9232407,-0.9779337)(3.9282405,-0.9260355)(3.9282405,-0.9000858)
\curveto(3.9282405,-0.8741361)(3.9382405,-0.8222379)(3.9482405,-0.7962882)
\curveto(3.9582407,-0.7703385)(3.9882407,-0.7184403)(4.0082407,-0.6924906)
\curveto(4.0282407,-0.6665415)(4.0732408,-0.62502295)(4.098241,-0.6094535)
\curveto(4.1232405,-0.5938834)(4.1732407,-0.5627445)(4.1982408,-0.54717445)
\curveto(4.223241,-0.531605)(4.2682405,-0.49527565)(4.2882404,-0.47451636)
\curveto(4.3082404,-0.45375678)(4.3432407,-0.40704778)(4.3582406,-0.3810984)
\curveto(4.3732405,-0.3551493)(4.3932405,-0.3032505)(4.3982406,-0.27730107)
\curveto(4.4032407,-0.25135168)(4.408241,-0.19426306)(4.408241,-0.16312383)
\curveto(4.408241,-0.13198462)(4.4132404,-0.06970617)(4.4182405,-0.03856695)
\curveto(4.4232407,-0.007427731)(4.4332404,0.054850712)(4.4382405,0.08598993)
\curveto(4.4432406,0.117129155)(4.4682407,0.17421778)(4.4882407,0.20016718)
\curveto(4.5082407,0.22611658)(4.5682406,0.27282557)(4.6082406,0.29358485)
\curveto(4.6482406,0.31434444)(4.6982408,0.35067347)(4.7282405,0.39738244)
}
\pscustom[linewidth=0.04,linecolor=color740]
{
\newpath
\moveto(7.1482406,-1.0226176)
\lineto(6.954907,-0.72261757)
\curveto(6.8582406,-0.57261753)(6.5682406,-0.37261754)(6.374907,-0.32261756)
\curveto(6.181574,-0.27261755)(5.7949076,-0.17261755)(5.6015744,-0.12261755)
\curveto(5.408241,-0.07261755)(5.118241,0.07738245)(4.8282404,0.37738246)
}
\usefont{T1}{ptm}{m}{n}
\rput(3.5527718,-1.5542188){$s_1$}
\psline[linewidth=0.056199998cm](4.1282406,1.5773824)(5.3282404,1.5773824)
\usefont{T1}{ptm}{m}{n}
\rput(4.692772,2.0057812){$X_0$}
\usefont{T1}{ptm}{m}{n}
\rput(7.942772,1.1857812){$C$}
\psdots[dotsize=0.12](4.7746468,0.4257815)
\usefont{T1}{ptm}{m}{n}
\rput(5.0387717,0.6201538244){$H$}
\psline[linewidth=0.056199998cm](3.3282406,-1.2226175)(4.5282407,-1.2226175)
\psline[linewidth=0.056199998cm](4.928241,-1.2226175)(6.1282406,-1.2226175)
\psline[linewidth=0.056199998cm](6.5590305,-1.2123542)(7.6974506,-0.8328809)
\psline[linewidth=0.056199998cm](1.7590307,-0.8328809)(2.8974507,-1.2123542)
\usefont{T1}{ptm}{m}{n}
\rput{20.5}(7.427717,-1.8242187){$X_q$}
\psline[linewidth=0.056199998cm](1.3525047,-0.64688164)(0.5039766,0.20164652)
\psline[linewidth=0.056199998cm](8.128241,-0.62261754)(8.903976,0.20164652)
\usefont{T1}{ptm}{m}{n}
\rput(3.98827717,-1.9742187){$X_i$}
\usefont{T1}{ptm}{m}{n}
\rput{-14.5}(0.29825303,1.6627085){\rput(6.252772,-0.09261755){$R$}}
\usefont{T1}{ptm}{m}{n}
\rput{55.84552}(1.491236,-3.3418628){\rput(3.912772,-0.39261755){$R$}}
\psdots[dotsize=0.12](4.317194,-1.2226175)
\psdots[dotsize=0.12](3.5282407,-1.2226175)
\psdots[dotsize=0.12](6.777194,-1.1366178)
\psdots[dotsize=0.12](7.497194,-0.89661777)
\usefont{T1}{ptm}{m}{n}
\rput(4.3527718,-1.5542188){$s_2$}
\usefont{T1}{ptm}{m}{n}
\rput{17.5}(6.882772,-1.485542187){$v_2$}
\usefont{T1}{ptm}{m}{n}
\rput{17.5}(7.592772,-1.302542188){$v_1$}
\usefont{T1}{ptm}{m}{n}
\rput{-45}(0.63010305,0.116012305){\rput(0.42277172,-0.6742188){$X_1$}}
\usefont{T1}{ptm}{m}{n}
\rput{50}(2.635408,-7.300693){\rput(9.052772,-0.65421873){$X_j$}}
\end{pspicture} 
}}%
    \end{center}
    \caption{Computing a pseudo shortcut for~$X_i$.}
    \label{fig:pshortcuts}
  \end{figure}

For each pair of
vertices $s_1$ and $s_2$ on the
essential boundary of~$X_i$ belonging to different sides of $C'$,
we then determine
a path $P$ of shortest weighted length
from $s_1$ to $s_2$ in $G_i$ and 
test whether its weighted length is shorter than 
the weighted length
of the crest-separator 
path of~$X_i$ from~$s_1$ to~$s_2$ and if
the integer part of their distance is at most $d$. If both is true, 
we replace each edge $\{x,y\}$ of~$P$ 
of weighted length $d_{G_i}(x,y)$
by a path in $D'_i$ with
length $d_{G_i}(x,y)$ and endpoints $x$ and
$y$. 
Finally, we add the so modified path $P^*$ 
to the pseudo shortcut set $\mathcal{L}(\DD'_i)$ for~$X_i$ 
and store $C$ as root component of $P^*$.
Moreover, analogously to the bottom-up traversal, 
if some subpaths of $P^*$ result from replacing
an edge $e_{P'}$ representing a pseudo shortcut $P'$
by the pseudo shortcut $P'$ itself, 
we add an edge $(P^*,P')$ into the consistence graph.

We next show the correctness of those steps of the top-down traversal
that are different from those of the bottom-up traversal.
As mentioned above
the concatenation $R$ of the ridges between $H$ and $H_i$ and 
between $H$ and $H_q$ divides
$C^*_i$ into two sides. 
If there is a $D'_i$-pseudo shortcut $P$ connecting $s_1$ with
$s_2$, this path must lead from the side containing $s_1$ to the
other side containing $s_2$. 
It cannot cross $R$
by Lemma~\ref{lem:ShortAlley}(b).
Hence the only possibility for a pseudo shortcut to change sides
is 
to use a pseudo shortcut of~$X_q$,
but not 
several times because
of Lemma~\ref{lem:ShortAlley}(a). 
To sum up, $P$
consists of the concatenation of
\begin{itemize} 
\item a path from $s_1$ to a vertex $v_1$ of~$X_q$ such that the path is completely contained in the same side 
as~$s_1$, 
\item a pseudo shortcut from $v_1$ to a vertex $v_2$ of~$X_q$ that is  on the
other side 
\item a path from $v_2$ to $s_2$ being completely contained in the same
side
as~$s_2$. 
\end{itemize}
Therefore, the construction of~$G_i$ implies that,
for each  $\DD'_i$-pseudo shortcut $P$ for~$X_i$
connecting a vertex $s_1$ and a vertex $s_2$, there is a path from $s_1$ to
$s_2$ of the same length as $P$ in $G_i$. Note also that the
distance of two vertices in $G_i$ is never smaller than the distance of that
vertices in $\DD'_i$ since each path in $G_i$ can be replaced by a path in
$\DD'_i$ with the same endpoints and the same length. Thus, 
it is correct to use the graph $G_i$ for
our  
computation of a pseudo shortcut set 
$\mathcal{L}(\DD'_i)$ for~$X_i$.

Concerning the running time, note that
the distances for the edges $\{x,y\}$ of~$C'$ or of~$G_i$ with
$x$ and $y$ being endpoints of a pseudo shortcut
are already computed during the bottom-up traversal
or in a previous step of the top-down traversal (as already
mentioned). 
Recall that $j=\deg_T(w)$.
All other distances 
needed for the construction of all graphs $G_1,\ldots,j$
can be computed by $O(\ell)$ single-source shortest
paths computations in $C'$, one for each vertex 
of~$X_q$ as source vertex.
If once again, $m_{C'}$ and $n_C$ are the number of 
edges of~$C'$ and vertices
of~$C$, respectively, each of the $O(\ell)$ single-source shortest path computations
runs in
$O(m_{C'}+\ell^2)=O(n_C+(j+1){\ell}^2)$
time.
Each subgraph $G_i$ consists of 
$O(\ell^2)$ edges, and we have to consider only 
$4\ell^2$
distances values. Thus, 
the $O(\ell)$ single-source shortest
path computations on all graphs 
$G_1,\ldots,G_{j}$ run in total time 
$O(j \cdot \ell (\ell^2+4\ell^2) )=O(j \cdot \ell^3 )$. 
Consequently, the whole top-down traversal can be done
in  $O(|{\mathcal H}| \ell^3 + |V|\ell)$
time.
\end{proof}

\section{Computing Coast Separators}\label{sec:interactn2}

In this section we want to show how one can construct
a set {of 
pairwise non-crossing       
coast} separators
such that
the inner graphs of the 
coast separators contain
all crests
of `large'  height.
For this purpose, we use coast separators of three different types and for
each such type we present at least
one  central technical lemma that
helps us to guarantee the disjointness of the constructed
coast separators.
Recall that $c_{\mathrm{max}}$ is the maximum weight over all vertices.

Take $k\in \Nat$.
Let %
$(G^+,\varphi^+,c^+)$
be a %
weighted, almost triangulated,
and biconnected plane graph
with treewidth $k'\le k$,
where $\varphi^+$ is a
$(2k+2c_{\mathrm{max}})$-weighted-outerplanar embedding. Moreover,
let $(G,\varphi,c)$ be the %
weighted, almost triangulated,
and biconnected plane 
graph obtained from
$(G^+,\varphi^+,c^+)$ as follows: 
each maximal connected set 
of vertices 
whose height interval contains a value of size  
at least
$k+2c_{\mathrm{max}}$ 
is merged to one vertex %
with a weight equal to $k+2c_{\mathrm{max}}+1$ minus the smallest lower height
among the merged vertices. This means that each merged vertex has a
lower height equal to the smallest lower height among the vertices
in $G^+$ merged to this
vertex and upper height $k+2c_{\mathrm{max}}$. 
Let $\mathcal H$ be the set
of those merged vertices that 
are 
obtained from merging 
a set of vertices that
contains a vertex of of $G^+$ with upper height $2k+2c_{\mathrm{max}}$.
Note that $\mathcal H$ is a set of different crests in $(G,\varphi,c)$.

In the following, assume that we are given the integer $k$
and, beside the graph $(G,\varphi,c)$ and the crest set $\mathcal{H}$,
a set $\mathcal{S}$ of crest separators such that 
$(G,\varphi,c,{\mathcal H},\mathcal{S})$ is a good mountain structure 
as well as the mountain connection tree $T$ of this good mountain structure.

{\bf Enclosing crest separators.}
The
first type of coast separators 
that we use is a cycle induced by the essential boundary 
of a crest separator with an interior lowpoint. The following lemma describes an important
property for these kind of coast separators:

\begin{lemma}\label{lemma:enclose}
An $(\mathcal{S},\varphi)$-component $C$ can not be enclosed by more than one
crest separator with a top edge in $C$.
\end{lemma}

\begin{proof}
Assume that there are two crest separators $X_1$ and $X_2$ 
enclosing $C$
that have a top edge in $C$.
Let $I_i$ ($i\in\{1,2\}$) be the inner graph of the essential boundary of~$X_i$.
If $X_1$ encloses $C$, then the top
edge $e_2$ of~$X_2$ is part of~$I_1$, but not part of the essential boundary
of~$X_1$.
By definition of the 
down paths
the essential boundary
of~$X_2$ cannot cross the essential boundary of~$X_1$.
Thus, 
$I_2$ 
is a subgraph of~$I_1$. Moreover, since the top edge $e_1$ of~$X_1$ can not be an edge of
$X_2$, edge $e_1$ can not be part of~$I_2$.
Then $X_2$ can not enclose $C$ since C contains~$e_1$.
\end{proof}

{\bf Composed cycles.} 
The next two lemmas describe 
important properties of pseudo shortcuts and composed cycles.
A composed cycle is the second type of coast separators that we
 use.
\begin{lemma}\label{lem:enclosev}
Let $P$ be a $D$-pseudo shortcut for a 
crest-separator path $X^{\mathrm{CP}}$ of a crest separator
$X\in\mathcal{S}$ with an $(X,\varphi)$-component $D$. 
For the $(\mathcal{S},\varphi)$-component $C$ containing
the top edge of~$X$ and being contained in $D$, the 
composed cycle of $(X^{\mathrm{CP}},P)$ encloses
the crest in~${\mathcal H}$ contained in $C$.
\end{lemma}

\begin{proof}
We define $C'$ to be the
$(\mathcal{S},\varphi)$-component different from $C$ 
containing the top edge of~$X$.
$H$ and $H'$ should denote the
crests 
in~${\mathcal H}$
contained in $C$ and $C'$, respectively.
Let 
$R$ be a ridge between $H$ and $H'$.
Since $P$ and $R$ cannot
cross (Lemma~\ref{lem:ShortAlley}(b)),
the composed cycle of~$(X^{\mathrm{CP}},P)$ encloses $H$.
\end{proof}

\begin{lemma}\label{lem:SinglePS} 
Let $C_0, C_1,\ldots,C_r$ consecutive nodes of a path
in the mountain decomposition tree of 
$(G,\varphi,c,\mathcal{H},\mathcal{S})$, and, for $i\in\{1,\ldots,r\}$,
let $X_i$ be the crest separator with a top edge in $C_{i-1}$ and $C_i$ (compare
Fig.~\ref{fig:pathMCT2}). 
\begin{figure}[b!]
    \begin{center}
         \scalebox{1}{
\scalebox{1} 
{
\begin{pspicture}(0,-1.9439106)(9.61,1.9320107)
\definecolor{color303}{rgb}{0.0,0.8,0.6}
\psline[linewidth=0.04cm,linestyle=dotted,dotsep=0.10583334cm](2.0,0.27890182)(4.2,0.27890182)
\psline[linewidth=0.04cm,linestyle=dotted,dotsep=0.10583334cm](4.2,0.27890182)(5.4,0.27890182)
\psline[linewidth=0.02cm](0.6,0.27890182)(1.8,0.27890182)
\psline[linewidth=0.04cm,linestyle=dotted,dotsep=0.10583334cm](5.6,0.27890182)(7.8,0.27890182)
\psline[linewidth=0.02cm](7.8,0.27890182)(9.0,0.27890182)
\psline[linewidth=0.02cm](0.6,0.27890182)(0.0,0.67890185)
\psline[linewidth=0.02cm](0.6,0.27890182)(0.0,-0.12109817)
\psline[linewidth=0.02cm](0.6,0.27890182)(0.4,-0.32109815)
\psline[linewidth=0.02cm](0.6,0.27890182)(0.0,0.27890182)
\psline[linewidth=0.02cm](1.8,0.27890182)(1.6,-0.32109815)
\psline[linewidth=0.02cm](1.8,0.27890182)(2.0,-0.32109815)
\psline[linewidth=0.02cm](4.2,0.27890182)(4.0,-0.32109815)
\psline[linewidth=0.02cm](4.2,0.27890182)(4.4,-0.32109815)
\psline[linewidth=0.02cm](9.0,0.27890182)(9.6,0.27890182)
\psline[linewidth=0.02cm](9.0,0.27890182)(9.6,0.67890185)
\psline[linewidth=0.02cm](9.0,0.27890182)(9.6,-0.12109817)
\psline[linewidth=0.02cm](9.0,0.27890182)(9.2,-0.32109815)
\psline[linewidth=0.02cm](7.8,0.27890182)(8.0,-0.32109815)
\psline[linewidth=0.02cm](7.8,0.27890182)(7.6,-0.32109815)
\psbezier[linewidth=0.02,linecolor=red](1.2,-1.3210982)(3.4,-1.3210982)(4.62625,-0.669497)(4.62625,0.29050297)(4.62625,1.250503)(3.4,1.878902)(1.2,1.878902)
\psbezier[linewidth=0.02,linecolor=color303](8.4,1.4789017)(8.0,1.4789017)(4.952,1.4589018)(4.94625,0.25050297)(4.9405,-0.9578958)(8.0,-0.92109823)(8.4,-0.92109823)
\pscircle[linewidth=0.019199999,dimen=outer,fillstyle=solid](0.6,0.27890182){0.2975}
\usefont{T1}{ptm}{m}{n}
\rput(0.58140624,0.27890182){$C_0$}
\pscircle[linewidth=0.02,dimen=outer,fillstyle=solid](1.8,0.27890182){0.2915}
\pscircle[linewidth=0.02,dimen=outer,fillstyle=solid](4.2,0.27890182){0.2915}
\pscircle[linewidth=0.02,dimen=outer,fillstyle=solid](7.8,0.27890182){0.2915}
\pscircle[linewidth=0.02,dimen=outer,fillstyle=solid](9.0,0.27890182){0.2915}
\usefont{T1}{ptm}{m}{n}
\rput(8.981406,0.27890182){$C_r$}
\usefont{T1}{ptm}{m}{n}
\rput(1.1195314,-1.9210982){$X_1$}
\usefont{T1}{ptm}{m}{n}
\rput(8.3695305,-1.9210982){$X_r$}
\psline[linewidth=0.056199998cm](1.2,2.1039107)(1.2,-1.5210983)
\psline[linewidth=0.056199998cm](2.4,2.1039107)(2.4,-1.5210983)
\psline[linewidth=0.056199998cm](3.6,2.1039107)(3.6,-1.5210983)
\psline[linewidth=0.056199998cm](8.4,2.1039107)(8.4,-1.5210983)
\psline[linewidth=0.056199998cm](7.2,2.1039107)(7.2,-1.5210983)
\usefont{T1}{ptm}{m}{n}
\rput(1.8214062,0.27890182){$C_1$}
\psline[linewidth=0.056199998cm](6.0,2.1039107)(6.0,-1.5210983)
\psline[linewidth=0.056199998cm](4.8,2.1039107)(4.8,-1.5210983)
\psline[linewidth=0.02cm](5.4,0.27890182)(5.2,-0.32109815)
\psline[linewidth=0.02cm](5.4,0.27890182)(5.6,-0.32109815)
\pscircle[linewidth=0.02,dimen=outer,fillstyle=solid](5.4,0.27890182){0.2915}
\usefont{T1}{ptm}{m}{n}
\rput(5.3814063,0.27890182){$C_i$}
\usefont{T1}{ptm}{m}{n}
\rput(4.7695313,-1.9210982){$X_i$}
\usefont{T1}{ptm}{m}{n}
\rput(1.8195313,-1.05210982){\color{red}$P_1$}
\usefont{T1}{ptm}{m}{n}
\rput(7.8195314,-1.1210982){\color{color303}$P_2$}
\end{pspicture} 
}
         }%
    \end{center}
    \vspace{-2mm}
    \caption{%
        The $(S,\varphi)$-components and crest separators described in Lemma~\ref{lem:SinglePS}.
    }
    \label{fig:pathMCT2}
  \end{figure}
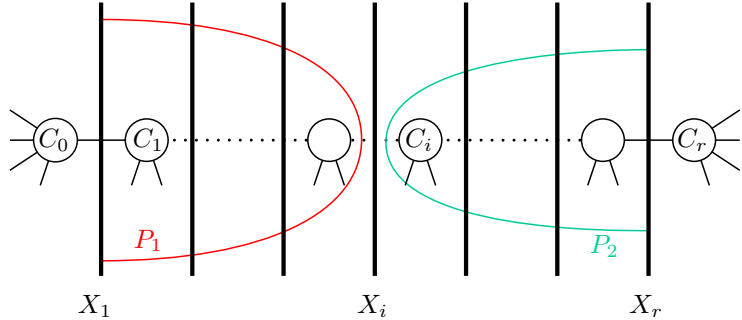%
Moreover, let $\DD_1$ be the $(X_1,\varphi)$-component containing
$C_1$, and let $\DD'_r$ be the $(X_r,\varphi)$-component containing
$C_{r-1}$. Then, for a nice $\DD_1$-pseudo shortcut $P_1$ of $X_1$ that
is completely contained in $\DD'_r$ and for a nice $\DD'_r$-pseudo shortcut
$P_r$ of $X_r$ that is completely contained in $\DD_1$, there exists
an $i\in\{2,\ldots,r-1\}$ such that $P_1$ is completely contained
in the $(X_i,\varphi)$-component containing $C_1$ and $P_r$ is
completely contained in the $(X_i,\varphi)$-component containing
$C_{r-1}$.
\end{lemma}

\begin{proof}
For $i\in\{2,\ldots,r-1\}$,
let $\DD_i$ be the $(X_i,\varphi)$-component containing $C_i$, and let
$\DD'_{i}$ be the $(X_{i},\varphi)$-component containing $C_{i-1}$.
Assume that the lemma does not hold, i.e., for at least one $i\in\{1,\ldots,r-1\}$,
$P_1$ must have a maximal $(\DD_i,X_i)$-subpath $P_1^*$ and $P_r$ a maximal
$(\DD'_{i+1},X_{i+1})$-subpath $P_r^*$. Let us assume w.l.o.g.\ that
we have chosen $i$ as large as possible.  
Since $P_1$ and 
$P_r$ are 
nice, $P_1^*$ and 
$P_r^*$ also
 define nice pseudo shortcuts 
and together with the crest separator 
paths between their endpoints enclose 
$H$ (Lemma~\ref{lem:enclosev}).
Hence, $P_1$ together with the crest separator path between its 
endpoints also encloses $H$ and, 
since we have chosen $i$ as large as 
possible, $P_1$ is completely contained in  
$D'_{i+1}$ whereas $P_r^*$ together with the crest separator path between its 
endpoints encloses
$H$ (Lemma~\ref{lem:enclosev}) and is completely contained in $D_1$ (see
also Fig.~\ref{fig:two_intersect}).
Hence $P_1$ and $P_r^*$ intersect.
As 
illustrated in Fig.~\ref{fig:two_intersect}.
one can interchange
subpaths of~$P_1$ with subpaths of~$P_r^*$ of the same length such
that afterwards the  inner graphs of the new composed cycles
do not intersect anymore except in some common vertices
of the old paths $P_1$ and $P_r^*$ and such that
the new composed cycles do not
enclose the crest $H$ anymore. In the case of the modified version
of $P_r^*$, this
 is a contradiction
to Lemma~\ref{lem:enclosev}.
\end{proof}

\begin{figure}[h]
\begin{center}
\scalebox{1}{
\scalebox{1} 
{
\begin{pspicture}(0.1,-0.8939063)(10.8079062,1.00939063)
\psset{xunit=1.3cm,yunit=0.8cm,runit=1cm}
\definecolor{color1079}{rgb}{0.0,0.8,0.6}
\definecolor{color1082}{rgb}{0.0,0.0,0.8}
\psbezier[linewidth=0.02,linecolor=color1079,linestyle=dashed,dash=0.16cm 0.16cm](7.900156,1.1736325)(7.5001564,1.1736325)(6.512156,0.99029917)(6.521875,0.17363244)(6.531594,-0.6430343)(7.5001564,-0.82636756)(7.900156,-0.82636756)
\psbezier[linewidth=0.035,linecolor=color1082,linestyle=dotted,dotsep=0.15cm](5.361875,-0.82636756)(5.781875,-0.82636756)(6.7956333,-0.612084)(6.759087,0.2038225)(6.7225413,1.019729)(5.748478,1.17114)(5.348694,1.1579975)
\psframe[linewidth=0.0020,linecolor=white,dimen=outer,fillstyle=solid](6.801875,0.6336324)(6.461875,-0.26636755)
\psbezier[linewidth=0.02,linecolor=color1079,linestyle=dashed,dash=0.16cm 0.16cm](3.1001563,1.1736325)(2.7001562,1.1736325)(1.7121563,0.99029917)(1.721875,0.17363244)(1.7315937,-0.6430343)(2.7001562,-0.82636756)(3.1001563,-0.82636756)
\usefont{T1}{ptm}{m}{n}
\rput(1.2373438,1.2820313){$P_1$}
\usefont{T1}{ptm}{m}{n}
\rput(2.4073438,1.2820313){$P_r^*$}
\usefont{T1}{ptm}{m}{n}
\rput(0.45546874,-1.3179687){$X_1$}
\usefont{T1}{ptm}{m}{n}
\rput(3.1073437,-1.3179687){$X_{i+1}$}
\psbezier[linewidth=0.035,linecolor=color1082,linestyle=dotted,dotsep=0.15cm](0.561875,-0.82636756)(0.981875,-0.82636756)(1.9956332,-0.612084)(1.9590873,0.2038225)(1.9225411,1.019729)(0.948478,1.17114)(0.54869395,1.1579975)
\usefont{T1}{ptm}{m}{n}
\rput(6.077344,1.2820313){$P'_1$}
\usefont{T1}{ptm}{m}{n}
\rput(7.1454687,1.2820313){$P'_r$}
\usefont{T1}{ptm}{m}{n}
\rput(5.255469,-1.3179687){$X_1$}
\usefont{T1}{ptm}{m}{n}
\rput(7.9073434,-1.3179687){$X_{i+1}$}
\usefont{T1}{ptm}{m}{n}
\rput(4.235469,0.18363245){\scalebox{1.8}{$\Rightarrow$}}
\psarc[linewidth=0.02,linecolor=color1079,linestyle=dashed,dash=0.16cm
0.16cm](6.02199875,0.17363244){0.97}{-27.5}{27.5}
\rput{180}(14.75,0.33857426){\psarc[linewidth=0.035,linecolor=color1082,linestyle=dotted,dotsep=0.15cm](7.499875,0.17363244){0.97}{-27.5}{27.5}}
\psline[linewidth=0.056199998cm](0.521875,1.3736324)(0.521875,-1.0263675)
\psline[linewidth=0.056199998cm](3.121875,1.3736324)(3.121875,-1.0263675)
\psline[linewidth=0.056199998cm](5.321875,1.3736324)(5.321875,-1.0263675)
\psline[linewidth=0.056199998cm](7.921875,1.3736324)(7.921875,-1.0263675)
\psdots[dotsize=0.12](1.831875,0.17363244)
\psdots[dotsize=0.12](6.643875,0.17363244)
\usefont{T1}{ptm}{m}{n}
\rput(2.1354687,0.16363245){$H$}
\usefont{T1}{ptm}{m}{n}
\rput(6.9354687,0.16363245){$H$}
\end{pspicture}
}}
\end{center}
\caption{The 
interchange 
of a subpath of~$P_1$ with a subpath of
$P_r^*$.}
\label{fig:two_intersect}
\end{figure}
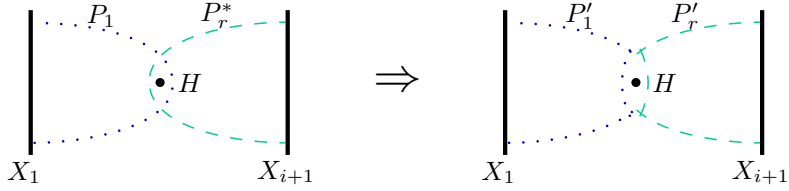%

Let $\mathcal{M}=(G,\varphi,c,\mathcal{H},\mathcal{S})$.
For a $d$-bounded $\mathcal{M}$-shortcut set 
$\mathcal{P}$, let
us define a {\em valid $d$-bounded $\mathcal{M}$-shortcut set}
to be a subset $\mathcal{P}'$ of $\mathcal{P}$ such that, for 
each $(\mathcal{S},\varphi)$-component $C$,
for which there is 
at least one crest separator $X\in\mathcal{S}$ 
with a top edge in 
$C$ that has a $\DD$-pseudo shortcut in $\mathcal{P}$ 
for the $(X,\varphi)$-component $\DD$ 
containing $C$, $\mathcal{P'}$ contains
at least one $\DD'$-pseudo shortcut 
for a  crest separator $X'\in\mathcal{S}$ 
with a top edge in 
$C$ %
where $\DD'$ is
the $(X',\varphi)$-component
containing $C$.
Roughly speaking, each $(\mathcal{S},\varphi)$-component $C$ with a
$d$-bounded pseudo
shortcut has a $d$-bounded pseudo shortcut in $\mathcal{P'}$.
If ``at least one'' in the above definition
is replaced by ``exactly one'',
the valid $d$-bounded $\mathcal{M}$-shortcut set is also 
called {\em non-overlapping}. 
We next show that such a shortcut set can be computed efficiently.

\begin{lemma}\label{lem:non-overlapping} Given a consistent $d$-bounded $\mathcal{M}$-shortcut
set, the consistence graph of $\mathcal{P}$ as well as, for each pseudo shortcut
$P\in\mathcal{P}$, its root component, a consistent non-overlapping
$d$-bounded $\mathcal{M}$-shortcut set 
can be constructed in $O(|\mathcal{H}|k^2)$ time.
\end{lemma}
\begin{proof} Let us call a vertex of a directed graph to be a {\em source}
if it has no incoming edge. If $F$ is the consistence graph of 
$\mathcal{P}$, we repeatedly apply the following step.
\begin{enumerate}%
\item[] Take a not-yet-considered source $P$ of $F$, and let $C$ be the root component of $P$.
    If there is another pseudo shortcut with root component $C$,
    then remove $P$ from $\mathcal{P}$ and $F$. \smallskip
\end{enumerate}%
\noindent The existence of $P'$ guarantees that the new set $\mathcal{P}$ is still
valid. %
Since we only remove
sources of $F$,
the new set is also consistent. Thus, apart from the running time, we only
have to show that, after having applied the above removal step as long as 
possible, we obtain a non-overlapping set of pseudo shortcuts. Therefore,
let us assume that after the removal steps there is still one $(S,\varphi)$-component $C$ for which there are two different crest separators $X'_1$ and $X'_2$ 
with top edges in $C$ such that, for the $(X'_1,\varphi)$-component $\DD'_1$ 
and the $(X'_2,\varphi)$-component $\DD'_2$ containing $C$, there exist
a $\DD'_1$-pseudo shortcut $P'_1$ of $X'_1$ and a $\DD'_2$-pseudo shortcut $P'_2$ 
of $X'_2$ in $\mathcal{P}$. Then neither $P'_1$ nor $P'_2$ can be a source of $F$
since, otherwise,
the removal step could be applied once more. Hence, let $P_1$ and $P_2$
be sources of $F$ for which there are paths from $P_1$ to $P'_1$ and from
$P_2$ to $P'_2$ in $F$. In particular, this means that $P'_1$ is a subpath
of $P_1$ and $P'_2$ a subpath of $P_2$ in $G$. Let $X_1$ and $X_2$ be the crest
separators such that $P_1$ and $P_2$ are pseudo shortcuts of $X_1$ and $X_2$,
respectively. Since the removal step
was  applied as long as possible, for the root component of $P_1$ and 
equivalently for that of $P_2$, there is exactly one crest separator $X$
with a pseudo shortcut in the $(X,\varphi)$-component containing the root
component of $P_1$, namely $X=X_1$ and equivalently---concerning the root component
of $P_2$---we have $X=X_2$. Hence $P_i$ ($i=1,2$) is completely contained in 
the $(X_{3-i},\varphi)$-component that contains $C$. %
Since $P_i$ %
with its subpath $P'_i$ %
 encloses 
the crest in $C$ (Lemma~\ref{lem:enclosev}),
there is no crest separator disconnecting $P_1$ and $P_2$; 
a contradiction to Lemma~\ref{lem:SinglePS}. 

Note that no source has to be considered a second time; if a source can
not be removed because of a missing second pseudo-shortcut in an
$(\mathcal{S},\varphi)$-component, then this is the case until the end of
the algorithm. Consequently,
the running time is linear in the number of pseudo shortcuts 
if, 
for each source of $F$, we can decide whether
it has to be removed in constant time. 
For that we store initially with each $(\mathcal{S},\varphi)$-component $C$
the number $n_C$ of pseudo shortcuts
with root component $C$.
If such a pseudo shortcut is a source in $F$ and
removed by the algorithm,
we decrease %
$n_C$
by one. 
The update caused by a removal of a %
pseudo shortcut
then can be done in
constant time. The decision whether a pseudo shortcut $P$ has to be removed
reduces to the question whether $n_C$ is greater than one for the root component
$C$ of $P$. This is correct because of the following: If there is another pseudo shortcut $P'$ that allows us to
remove $P$, then a subpath of $P'$ is a pseudo shortcut with root component
$C$, i.e., $n_C$ is indeed greater than one.
The initialization of all numbers stored with the components
can be done in $O(|\mathcal{P}|)$ time. So the whole running time is bounded
by $O(|\mathcal{P}|)=O(|\mathcal{H}|k^2)$. 
\end{proof}

\medskip
{\bf Minimal coast separators}.
The third type of coast separators that we use are
minimal
coast separators.
Such coast separators are used in
a special kind of~$(S,\varphi)$-components
defined now. 
An
$(\mathcal{S},\varphi)$-com\-po\-nent
$\CC$ is called
{\em pseudo shortcut free} if (1) 
$C$ is not enclosed by a crest separator in $\mathcal{S}$ with a top edge in $C$ 
and an interior lowpoint and %
(2)
for all crest
separators $X\in\mathcal{S}$
with a top edge
in $\CC$ and the
$(X,\varphi)$-component $\DD$ containing $\CC$,
there is no $\DD$-pseudo shortcut
of weighted length $\le k-1$ for~$X$.

Moreover, let us define a crest separator to be {\em pseudo shortcut free}
if it has neither a pseudo shortcut of weighted length $\le k-1$
nor an interior lowpoint.
Intuitively speaking, a pseudo shortcut is 
of interest if %
it is possibly
part of a coast
separator of weighted size at most~$k$.
This is the reason why we only consider pseudo shortcuts
of weighted length $\le k-1$.
As the next lemma shows,
pseudo shortcut free components are separated by
pseudo shortcut free crest separators.

\begin{lemma}\label{lcor:interactPath2}
Let $H',H''\in {\mathcal H}$
be crests of different
pseudo shortcut free $({\mathcal S},\varphi)$-components.
Let ${\mathcal S'} \subseteq {\mathcal S}$ be the set of pseudo shortcut
free crest separators.
Then there is a crest separator in ${\mathcal S'}$
strongly going between the crests.
In particular, a minimal coast separator for~$H'$ 
of weighted size $\le k$
can only enclose
those crests in $\mathcal H$ that are  part of the $({\mathcal S'},\varphi)$-component that contains $H'$.
\end{lemma}

\begin{figure}[b!]
    \vspace{3mm}
    \begin{center}
         \scalebox{1}{
\scalebox{1} 
{
\begin{pspicture}(0,-1.9697193)(8.947187,2.0578193)
\definecolor{color1600}{rgb}{0.0,0.8,0.6}
\psline[linewidth=0.02cm](0.901875,0.0062181647)(2.101875,0.0062181647)
\psline[linewidth=0.02cm](2.101875,0.0062181647)(3.301875,0.0062181647)
\psline[linewidth=0.04cm,linestyle=dotted,dotsep=0.10583334cm](3.301875,0.0062181647)(4.501875,0.0062181647)
\psline[linewidth=0.04cm,linestyle=dotted,dotsep=0.10583334cm](4.501875,0.0062181647)(6.701875,0.0062181647)
\psline[linewidth=0.02cm](6.701875,0.0062181647)(7.901875,0.0062181647)
\psline[linewidth=0.02cm](0.901875,0.0062181647)(0.301875,0.40621817)
\psline[linewidth=0.02cm](0.901875,0.0062181647)(0.301875,-0.39378184)
\psline[linewidth=0.02cm](0.901875,0.0062181647)(0.701875,-0.5937818)
\psline[linewidth=0.02cm](0.901875,0.0062181647)(0.301875,0.0062181647)
\psline[linewidth=0.02cm](2.101875,0.0062181647)(1.901875,-0.5937818)
\psline[linewidth=0.02cm](2.101875,0.0062181647)(2.301875,-0.5937818)
\psline[linewidth=0.02cm](3.301875,0.0062181647)(3.101875,-0.5937818)
\psline[linewidth=0.02cm](3.301875,0.0062181647)(3.501875,-0.5937818)
\psline[linewidth=0.02cm](4.501875,0.0062181647)(4.301875,-0.5937818)
\psline[linewidth=0.02cm](4.501875,0.0062181647)(4.701875,-0.5937818)
\psline[linewidth=0.02cm](7.901875,0.0062181647)(8.501875,0.0062181647)
\psline[linewidth=0.02cm](7.901875,0.0062181647)(8.501875,0.40621817)
\psline[linewidth=0.02cm](7.901875,0.0062181647)(8.501875,-0.39378184)
\psline[linewidth=0.02cm](7.901875,0.0062181647)(8.101875,-0.5937818)
\psline[linewidth=0.02cm](6.701875,0.0062181647)(6.901875,-0.5937818)
\psline[linewidth=0.02cm](6.701875,0.0062181647)(6.501875,-0.5937818)
\psbezier[linewidth=0.02,linecolor=red](1.501875,-1.5937818)(3.701875,-1.5937818)(5.24,-0.9419068)(5.24,-0.0019068355)(5.24,0.9380932)(3.701875,1.6062182)(1.501875,1.6062182)
\psbezier[linewidth=0.02,linecolor=red](2.701875,1.2062181)(3.601875,1.0062182)(4.973875,0.60621816)(4.973875,0.0062181647)(4.973875,-0.5937818)(3.601875,-0.99378186)(2.701875,-1.1937819)
\psbezier[linewidth=0.02,linecolor=red](3.901875,1.6062182)(4.801875,1.4062182)(5.489875,0.80621815)(5.489875,0.0062181647)(5.489875,-0.7937818)(4.801875,-1.3937818)(3.901875,-1.5937818)
\psbezier[linewidth=0.02,linecolor=color1600](7.301875,1.2062181)(6.901875,1.2062181)(6.173875,0.60621816)(6.173875,0.0062181647)(6.173875,-0.5937818)(6.901875,-1.1937819)(7.301875,-1.1937819)
\pscircle[linewidth=0.019199999,dimen=outer,fillstyle=solid](0.901875,0.0062181647){0.2975}
\usefont{T1}{ptm}{m}{n}
\rput(0.86546874,0.0062181647){$C_0$}
\pscircle[linewidth=0.02,dimen=outer,fillstyle=solid](2.101875,0.0062181647){0.2915}
\usefont{T1}{ptm}{m}{n}
\rput(2.0654688,0.0062181647){$C_1$}
\pscircle[linewidth=0.02,dimen=outer,fillstyle=solid](3.301875,0.0062181647){0.2915}
\usefont{T1}{ptm}{m}{n}
\rput(3.2654688,0.0062181647){$C_2$}
\pscircle[linewidth=0.02,dimen=outer,fillstyle=solid](4.501875,0.0062181647){0.2915}
\pscircle[linewidth=0.02,dimen=outer,fillstyle=solid](6.701875,0.0062181647){0.2915}
\pscircle[linewidth=0.02,dimen=outer,fillstyle=solid](7.901875,0.0062181647){0.2915}
\usefont{T1}{ptm}{m}{n}
\rput(4.504688,0.0062181647){\scalebox{.8}{$C_{i\mbox{-}1}$}}
\usefont{T1}{ptm}{m}{n}
\rput(7.865469,0.0062181647){$C_r$}
\usefont{T1}{ptm}{m}{n}
\rput(1.4654688,-2.1937819){$X_1$}
\usefont{T1}{ptm}{m}{n}
\rput(2.6654687,-2.1937819){$X_2$}
\usefont{T1}{ptm}{m}{n}
\rput(3.8654687,-2.1937819){$X_{i-1}$}
\usefont{T1}{ptm}{m}{n}
\rput(5.7654687,-2.1937819){$X_i$}
\usefont{T1}{ptm}{m}{n}
\rput(7.2654686,-2.1937819){$X_r$}
\psline[linewidth=0.056199998cm](1.501875,2.3297193)(1.501875,-1.7937819)
\psline[linewidth=0.056199998cm](2.701875,2.3297193)(2.701875,-1.7937819)
\psline[linewidth=0.056199998cm](3.901875,2.3297193)(3.901875,-1.7937819)
\psline[linewidth=0.056199998cm](5.801875,2.3297193)(5.801875,-1.7937819)
\psline[linewidth=0.056199998cm](7.301875,2.3297193)(7.301875,-1.7937819)
\end{pspicture} 
}
         }%
    \end{center}
    \vspace{2mm}
    \caption{%
     A path $\tilde{P}$ in the mountain connection tree $T$         
      of the mountain structure 
    $(G,\varphi,c,{\mathcal H},{\mathcal S})$        
    with the
    crest separators $X_1,\ldots,X_r$ between the
    $(\mathcal{S},\varphi)$-components     
    corresponding to the
    nodes of~$\tilde{P}$.                                            
    If in the figure a curve representing a pseudo 
    shortcut $P$ together with its crest-separator path $X^{\mathrm CP}$
    encloses a node
    representing an $(\mathcal{S},\varphi)$-component $C$, this should
    mean that at least some vertices of~$C$ are enclosed by the
    composed cycle of~$(X^{\mathrm CP},P)$. 
    }
    \label{fig:pathMCT}
  \end{figure}
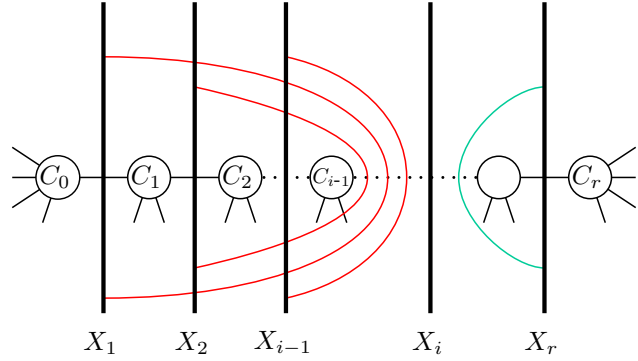

\begin{proof}
As illustrated in Fig.~\ref{fig:pathMCT}, let 
 $C_0,\ldots,C_r$ be the $({\mathcal S},\varphi)$-components
on the path in the mountain connection tree $T$ of
$(G,\varphi,c,{\mathcal H},\mathcal{S})$
from the 
$(\mathcal{S},\varphi)$-component~$C_0$
containing $H'$ to the
$(\mathcal{S},\varphi)$-component $C_r$
containing $H''$. For $i\in\{1,\ldots,r\}$, let $X_i \in \mathcal{S}$ be
 the crest separator with a
top edge part of both $C_{i-1}$ and $C_i$. 
First note that no crest separator $X\in\{X_1,\ldots,X_r\}$
can have an interior lowpoint.
Otherwise, such a crest separator
 would enclose either $C_0$ or $C_r$.
This would imply that either $X_1$ encloses
$C_0$ and has an interior lowpoint or
$X_r$ encloses $C_r$ and  has an interior lowpoint.
This contradicts
our choice of~$C_0$ and $C_r$ as
pseudo 
shortcut free $({\mathcal S},\varphi)$-components.

For $i\in \{1,\ldots,r\}$, let $D_i$ be the $(X_i,\varphi)$-component that contains
$C_r$. 
Take $X_i$ as the first 
crest separator
in $X_1,\ldots,X_r$ 
with no $D_i$-pseudo shortcut
of weighted length $\le k-1$.
$X_i$ exists
since $X_r$ has no $D_r$-pseudo shortcut
of weighted length $\le k-1$.
Then $X_i$ also has no pseudo shortcut
for the $(X_i,\varphi)$-component opposite to $D_i$. This follows
either from the fact that $C_0$ is
pseudo shortcut free if $i=1$ and
from 
Lemma~\ref{lem:SinglePS} otherwise. Hence $X_i$ is
pseudo shortcut free.
Finally note that a minimal coast separator of weighted size 
$\le k$
for
$H'$  
can only enclose
those crests in $\mathcal H$ that are
part of the $({\mathcal S'},\varphi)$-component that contains $H'$,
since, otherwise, it must cross a crest~separator~$X$ in ${\mathcal S'}$ and 
$X$ then has a pseudo shortcut of weighted length $\le k-1$
by Lemma~\ref{lem:StrongOptimal}.
\end{proof}

We next show that minimal coast separators can be computed
efficiently.

\begin{lemma}\label{lem:generalMaxFlowMinCut}
There is an algorithm that
computes a minimal coast separator $Y$ for a given crest $H$ in
a weighted, almost triangulated graph $(G,\varphi,c)$
 in $O(n w(Y))$ time where $n$ is the
 number of vertices of~$G$ and $w(Y)$ is the total weight of the vertices in
 $Y$.
\end{lemma}

\begin{proof}
In an unweighted graph $G'=(V',E')$, by network flow theory, one can 
construct a
separator $Y'$ of minimal size for
two 
connected sets $S$ and $T$ 
as follows:
Construct a graph $G^+=(V^+,E^+)$ obtained from the following
steps and sketched in Fig.\ref{fig:transormation}.
First, merge the vertices of~$S$ and $T$ to two single vertices 
$s$ and $t$,
respectively.
Second, replace each undirected edge $\{u,v\}$ by two directed edges,
one leading from $u$ to $v$ and the other from $v$ to $u$. Third, split
each {\em original} vertex $v$ of $G'$ into an {\em invertex} $v^{\mathrm{i}}$
and an {\em outvertex} $v^{\mathrm{o}}$ such that each directed edge originally ending in $v$ 
afterwards ends in $v^{\mathrm{i}}$ and each directed edge originally starting in $v$ 
afterwards starts in $v^{\mathrm{o}}$, and add an edge from $v^{\mathrm{i}}$
to~$v^{\mathrm{o}}$.

Afterwards, compute a maximal set ${\mathcal P}$ of edge-disjoint
paths from $s^{\mathrm{o}}$ to $t^{\mathrm{i}}$ in $G^+$. %
\begin{figure}[b!]
   \centering
     \scalebox{1.25}{\input{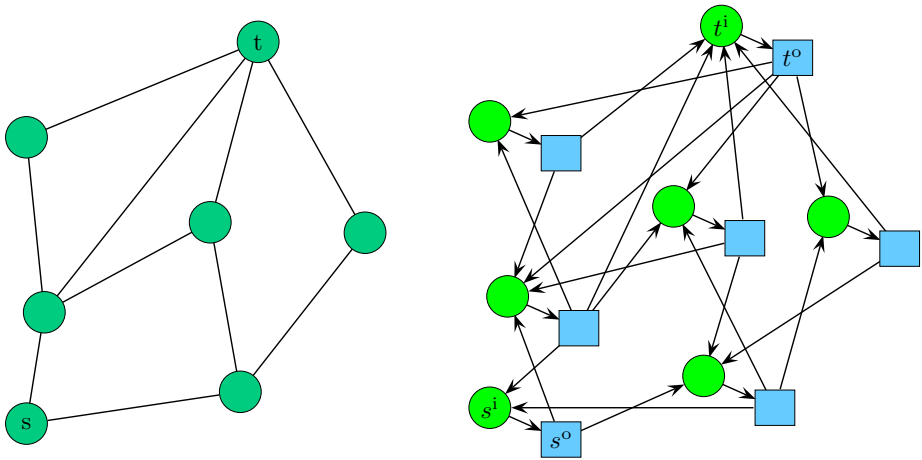}}%
      \caption{Vertex-disjoint path from $s$ to $t$
               in the left graph can be found
       by finding edge-disjoint path from $s^{\mathrm{o}}$ to $t^{\mathrm{i}}$ in the right
 graph. The squared vertices represent
           outvertices.}
      \label{fig:transormation}
\end{figure}
Let $E({\mathcal P})$ %
be the set of the
edges 
that are used by the paths in ${\mathcal P}$.
Then construct the so-called {\em residual graph} $G^R=(V^R,E^R)$ of
$G^+$ defined as the graph obtained from $G^+$ by replacing 
each edge $e\in E({\mathcal P})$ 
by an edge in reverse direction.
Define $U^+$ to be the set that consists, for each path ${P}\in {\mathcal P}$,
of the last vertex of ${P}$ that is reachable from
$s^{\mathrm{o}}$ in $G^R$. Note that
$U^+$ must be a set of invertices
because of the following Fact 1---consider
also Fig.~\ref{fig:minSep}: the only incoming edge of 
an outvertex $v^{\mathrm{o}}$ in $G^R$
belonging to a path ${P}\in {\mathcal P}$ is a reverse edge of ${P}$ 
that starts in an invertex appearing after $v^{\mathrm{o}}$ on ${P}$. By network flow theory it
is well known that the set $U$ obtained from $U^+$ by replacing
each %
vertex
in $U^+$ by its
original vertex is a separator of minimal size in $G'$. Indeed,
every
separator must contain at least one original vertex for 
every path in ${\mathcal P}$ 
and it is easy to see 
that after the removal of the vertices in $U^+$ from $G^+$, 
there
is no path from $s^{\mathrm{o}}$ to $t^{\mathrm{i}}$ in $G^+$ anymore, and that 
this implies 
that
$G'$ has no path from %
a vertex in $S$ to a vertex in $T$
after the removal of the 
vertices in $U$.

For each separator $W$ of minimal size for the two 
connected sets $S$ and $T$,
let us define $\mathrm{CC}(W)$ as the set of vertices 
that are in the same connected component as 
the vertices of $S$ in $G'[V\setminus W]$.
We now show that $U$ is an %
$S$-{\em minimal separator}, i.e.,
$\mathrm{CC}(U) \subseteq \mathrm{CC}(W)$ for all separators $W$ of
minimal size 
for %
$S$ and $T$. 
Hence, assume for a contradiction that this does not hold for such a
separator $W$.
Let $W^+$ be the vertex set obtained from $W$
after replacing each vertex $w\in W$ by the invertex $w^{\mathrm{i}}$.
Note that $W^+$ must contain %
exactly one
vertex on each path of~${\mathcal P}$ since $|{\mathcal P}|=|W|$.
By our choice of $W$, there must be at least one path
${P}'\in {\mathcal P}$ for which the invertex $u^{\mathrm{i}}$ of $U^+$ on ${P}'$ appears
after the invertex $w^{\mathrm{i}}$ in $W^+$ on ${P}'$. 
Let $\tilde{{P}}$ be a path in $G^R$ %
from $s^{\mathrm{o}}$ to 
$u^{\mathrm{i}}$,
and
let ${P}$ be the subpath of
$\tilde{{P}}$ 
ending in the first vertex $v_1$ of $\tilde{{P}}$ that
appears on a path in ${\mathcal P}$ after a vertex of $W^+$. (Possibly,
$v_1=u^{\mathrm{i}}$.)
See also Fig.~\ref{fig:minSep}.
\begin{figure}[b!]
   \centering
     \vspace{4mm}
     \input{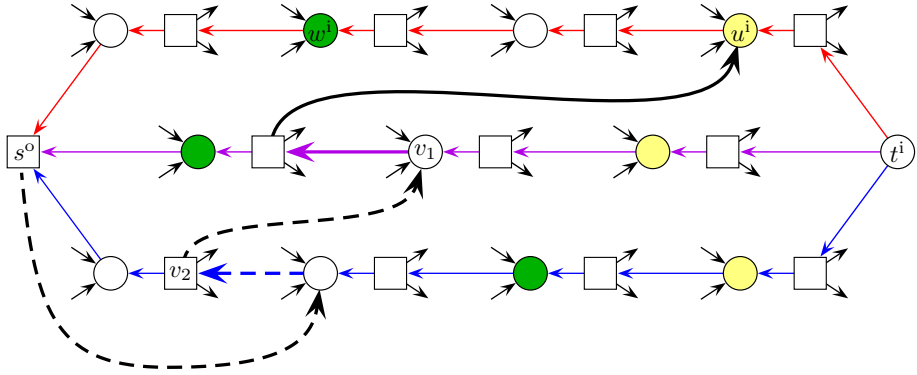}%
     \vspace{3mm}
      \caption{%
A part of the residual graph $G^R$. Each straight line
           between two vertices of the figure represents a reverse edge
           of a path in $\mathcal{P}$. 
The squared vertices represent outvertices.
           The light or dark colored vertices are the
           vertices of $U^+$ and $W^+$, respectively. The bold curves
           together with the bold reverse edges define the
           path $\tilde{P}$, where the dashed curves and edges represent the
           subpath
           $P$ of $\tilde{P}$.}
      \label{fig:minSep}
\end{figure}
Note that $v_1$ must be an invertex (Fact 1).

Let us now consider the
last vertex $v_2$ on ${P}$ appearing before $v_1$ on one of the
paths in ${\mathcal P}$. (Possibly, $v_2=s^{\mathrm{o}}$.) Vertex $v_2$ must be an outvertex 
because of a reason similar to Fact 1.
Let ${P}_1$
and ${P}_2$ be the paths of ${\mathcal P}$ containing $v_1$ and $v_2$, 
respectively.
Since
the only incoming edge of $v_2$ in $G^R$ is a reverse edge of ${P}_2$,
$v_2$
must appear
strictly before the only vertex of $W^+$ on ${P}_2$.
Consequently, the subpath of ${P}_2$ from $s^{\mathrm{o}}$ to $v_2$ and
the subpath of ${P}_1$ from $v_1$ to $t^{\mathrm{i}}$ are paths 
in $G^+$ that %
do not contain any vertices of $W^+$.
The same is true for the subpath of ${P}$ from $v_2$
to $v_1$ since the inner vertices of this subpath are not
part of any path in ${\mathcal P}$. Hence, there is a path
from $s^{\mathrm{o}}$ to $t^{\mathrm{i}}$ in $G^+$ not containing any vertex of
$W^+$.
By replacing the in- and outvertices of these paths by 
the original vertices, we obtain a path from a vertex in $S$ to a vertex in $T$ in
$G'$ not
containing any vertex of $W$. This is a contradiction to our
choice of $W$ as a separator and means that
$U$ is indeed $S$-minimal.

If we are given a weighted, almost triangulated graph $(G,\varphi,c)$
with $G=(V,E)$,
we can generalize the approach above to find a minimal coast separator for a
set $H$: First, replace each vertex $v$ of~$G$
by $c(v)$ copies. For each edge $\{u,v\}$ of~$G$, add an edge between each copy of
$u$ with each copy of~$v$. Let $G'=(V',E')$ be the unweighted graph obtained.
Define $S$ to be the vertex set consisting of the copies of the vertices of 
$H$, and let $T$ be the vertex set consisting of the copies of the vertices of
the coast.
Then construct
a separator $Y'$ strongly disconnecting $S$ and $T$ as above.
Let $Y$ be
the set of vertices of~$G$ whose copies are all in $Y'$. 
Then $Y$ is an $H$-minimal separator.
By
the fact that $G$ is almost triangulated,
$Y$ is a minimal coast separator for~$H$ in~$G$.

For an efficient implementation, there is no need to replace 
$(G,\varphi,c)$ by an unweighted graph. Instead, 
we use classical network flow techniques to construct
a maximum
number of paths
connecting a vertex of~$H$ and a vertex of~$H'$ such that each vertex $v \notin
H\cup H'$
is part
of at most
$c(v)$ paths 
and then read of a separator from these paths in a similar
way as described above. Since we can construct each path in $O(n)$ time, 
and since we have to construct $O(w(Y))$ 
paths,
we can
compute a minimal coast separator in $O(n w(Y))$ time.
\end{proof}

\def\ellk{%
k%
}

\medskip\medskip
{\bf Algorithm.} 
 Recall that 
$(G^+,\varphi^+,c^+)$
is $(2k+2c_{\mathrm{max}})$-weighted-outerpla\-nar.
We next present an algorithm to
construct a set $\mathcal{Y}$ of
coast separators of size $O(k)$ for~$\mathcal H$
in $(G,\varphi,c)$. As we will see, these sets 
 also define coast separators
for the crests 
of
height $2k+2c_{\mathrm{max}}$ in $(G^+,\varphi^+,c^+)$.
By Observation~\ref{obs:enclosed} and 
Lemma~\ref{lem:mdisjoint}, we will conclude that each crest 
in $\mathcal H$
is enclosed by
exactly one coast separator in $\mathcal{Y}$.
In Lemma~\ref{lem:mconnected}, we show that the 
$(S,\varphi)$-components of the crests enclosed by a coast
separator in $\mathcal{Y}$ induce a connected 
subtree of
$T$.
\medskip\medskip

\noindent {\bf Step~1:} Run the algorithm from Lemma~\ref{lem:ShortCutSet}
to compute a consistent $(k-1)$-bounded $\mathcal{M}$-shortcut set $\mathcal{P}$
for $\mathcal{M}=(G,\varphi,c,\mathcal{H},\mathcal{S})$, the consistence
graph of $\mathcal{P}$ and for each pseudo shortcut in $\mathcal{P}$
its root component. This subsequently allows us
to
determine the set $\mathcal{S'}$ of all 
pseudo shortcut free crest separators 
in $\mathcal{S}$
as well as all $(\mathcal{S}',\varphi)$-components. 

For each crest
$H \in {\mathcal H}$
contained in a
pseudo shortcut free 
$(\mathcal{S},\varphi)$-component contained in an
$(\mathcal{S}',\varphi)$-component $C$,  
compute a minimal coast
separator~$Y_C$ as follows:
Compute first 
the plane graph $(C',\psi)$ obtained from the
plane graph
$(C,\varphi|_C)$ as follows. For each subsequence of
vertices $v_0,v_1,\ldots,v_q$ appearing clockwise on the coast of 
$\varphi|_C$ where only the vertices $v_0$ and $v_q$ belong
to the coast of $\varphi$, add
an extra vertex $v^*$ into the outer face
of~$(C,\varphi|_C)$ and insert edges
from $v^*$ to each of the vertices $v_0,\ldots,v_{q}$
so that the outer face of~$\psi$ consists of the newly inserted
vertices together with the vertices of~$C$
that belong to the coast
of $\varphi$.
More precisely, in the special case where the coast of
$\varphi|_C$ contains no vertex of the coast of $\varphi$, add
a cycle consisting of 
three new vertices $v_1^*, v_2^*, v_3^*$ of weight 1 into the outer face of 
$\varphi|_C$ and afterwards connect these new vertices to
vertices of the coast of $\varphi|_C$ until the
graph obtained is almost triangulated.
Then, use the algorithm of Lemma~\ref{lem:generalMaxFlowMinCut}
to construct a 
minimal coast separator~$Y_C$ for~$H$ in $(C',\psi)$. The existence of $Y_C$
is shown in
Lemma~\ref{lem:coastSepExist}.
After having found $Y_C$, determine the inner graph of~$Y$ and all crests 
in ${\mathcal H}$ that are
part
of the inner graph.
Add $Y_C$
to the initial empty set $\mathcal{Y}$.
\medskip\medskip

\noindent {\bf Step~2:} Use Lemma~\ref{lem:non-overlapping} to construct a 
non-overlapping
consistent $(k-1)$-bounded $\mathcal{M}$-shortcut set $\mathcal{Z}\subseteq\mathcal{P}$.
Initialize $E'$ as an empty set of directed edges.
For each crest
$H \in {\mathcal H}$
that is 
not
enclosed by a coast separator constructed 
(possibly for some crest $H'\neq H$)
in Step~1, take $C$ as the $(\mathcal{S},\varphi)$-component
containing $H$ and test whether 
there is a crest separator $X\in\mathcal{S}$ with a top edge
in $C$ that has an interior lowpoint and encloses $C$.%
\begin{itemize}
\item[A:] If so, let $X_C$ be this crest separator and $Y_C$ be
the cycle induced by edges of the essential boundary of~$X_C$.

\item[B:]  Otherwise, define $X_C$ to be the crest separator with a top
edge in $C$ that has a nice
pseudo shortcut
$P$ in $\mathcal Z$.
$X_C$ exists since $C$ is not pseudo shortcut 
free and is unique because %
$\mathcal Z$ is non-overlapping.
Take $X^{\mathrm CP}_C$ as the crest-separator path of~$X_C$ such that $P$ is a pseudo
shortcut for~$X^{\mathrm CP}_C$ and define 
$Y_C$ as the composed cycle of~$(X^{\mathrm CP}_C,P)$.%
\end{itemize}
Let $C^*$ be 
the $(\mathcal{S},\varphi)$-component that contains the top edge
of~$X_C$, but is different to $C$. 
Finally add a directed  edge $(C,C^*)$
to $E'$ if the following conditions are all satisfied.%
\begin{description}
\item[Condition 1:] $C^*$ has a 
crest
$H \in {\mathcal H}$
that is 
not
enclosed by a coast separator constructed in Step~1.
\item[Condition 2:] Either $C^*$ is not enclosed by $X_C$ %
or $X_C$ has not
an interior lowpoint.
\item[Condition 3:] $\mathcal{Z}$ does not contain a $D^*$-pseudo 
shortcut for~$X_C$ %
where $D^*$ is the
$(X_C,\varphi)$-compo\-nent containing $C^*$.%
\end{description}
\noindent Intuitively, an edge $(C,C^*)$ indicates the possibility that the
coast separator constructed for the crest in $C^*$ also encloses the crest
in $C$ (compare also Lemma~\ref{lem:directedEdges}).
\medskip\medskip

\noindent {\bf Step~3:} Let $\tilde{F}$ be 
a graph consisting of that vertices of
the mountain connection tree
that represent
$(\mathcal{S},\varphi)$-components
that have a crest part of ${\mathcal H}$ such that the crest
is not enclosed by the coast separators constructed
in Step~1. The edges of $\tilde{F}$ are the edge set ${E'}$ constructed in Step~2. 
This is a directed forest of intrees since we assigned at most
one parent to each $(\mathcal{S},\varphi)$-component since 
the unweighted version of the edges are part of a tree, namely
the mountain connection tree, and since because of the Conditions 2 and 3,
for each pair $C_1$ and $C_2$ of $(\mathcal{S},\varphi)$-components,
there exists at most one of the edges $(C_1,C_2)$ and $(C_2,C_1)$.
See also Fig.~\ref{fig:alg_init2} and~\ref{fig:algPH}.
Then, for each tree $\tilde{T}$ of~$\tilde{F}$ run the following substeps:
Let
$H \in {\mathcal H}$
be the crest contained in the 
$(\mathcal{S},\varphi)$-component $C$
 being the root of 
$\tilde{T}$.
Put $Y_C$ into ${\mathcal{Y}}$,
determine the inner graph of~$Y_C$ and all crests in $\mathcal H$
that are part
of the inner graph, and
remove all nodes from $\tilde{T}$ that are
$(\mathcal{S},\varphi)$-components containing
those crests.
Then recursively proceed in the same way for each  
remaining intree being a subgraph of~$\tilde{T}$.%
\medskip\medskip

\noindent {\bf Step~4:} Return the set ${\mathcal{Y}}$.

  \begin{figure}[]
    \centering{
       {
\scalebox{1.1} 
{
\begin{pspicture}(0,-2.07175)(10.1511,2.57185)
\definecolor{greenblue}{rgb}{0.0,0.8,0.6}
\definecolor{color306b}{rgb}{1.0,0.8,0.6}
\definecolor{color296b}{rgb}{0.4,1.0,1.0}
\psline[linewidth=0.02cm](3.8757,1.45705)(2.6757,0.05704997)
\psline[linewidth=0.02cm](5.8757,1.45705)(3.8757,1.45705)
\psline[linewidth=0.02cm](2.6757,0.05704997)(1.9287,-1.54295)
\psline[linewidth=0.02cm](2.6757,0.05704997)(3.4167,-1.54295)
\psline[linewidth=0.02cm](3.8757,1.45705)(4.8757,0.05704997)
\psline[linewidth=0.02cm](7.8757,1.45705)(6.6757,0.05704997)
\psline[linewidth=0.02cm](7.8757,1.45705)(9.0757,0.05704997)
\psline[linewidth=0.02cm](6.6757,0.05704997)(5.9287,-1.54295)
\psline[linewidth=0.02cm](9.0757,0.05704997)(9.0757,-1.54295)
\psline[linewidth=0.16cm,linecolor=red](3.0457,1.00305)(3.4927,0.46104997)
\psline[linewidth=0.16cm,linecolor=red](7.0117,1.07905)(7.5287,0.44904998)
\psline[linewidth=0.06cm](2.5947,-0.95295006)(2.0067,-0.55195004)
\psline[linewidth=0.06cm](2.7637,-0.96395004)(3.3287,-0.55695003)
\psline[linewidth=0.06cm](5.9557,-0.52295005)(6.5887,-0.94095004)
\psline[linewidth=0.06cm](4.8757,1.05705)(4.8757,1.85705)
\psline[linewidth=0.16cm,linecolor=red](4.0857,0.47104996)(4.6397,1.0070499)
\psline[linewidth=0.06cm](2.8757,1.05705)(2.8757,1.85705)
\psline[linewidth=0.06cm](4.4007,-0.65495)(5.1447,-0.78295004)
\psline[linewidth=0.06cm](8.7347,1.0530499)(8.2487,0.44804996)
\psline[linewidth=0.06cm](8.7177,-0.74895)(9.4297,-0.74395)
\psline[linewidth=0.0207cm](3.8757,1.45705)(1.8757,1.45705)
\psline[linewidth=0.0207cm](6.6757,0.05704997)(4.8757,0.05704997)
\psline[linewidth=0.0207cm](6.6757,0.05704997)(7.4757,-1.54295)
\psline[linewidth=0.0207cm](4.8757,0.05704997)(4.6757,-1.54295)
\psline[linewidth=0.06cm](6.7597,-0.93695)(7.3827,-0.47095004)
\psline[linewidth=0.06cm](5.7697,-0.35995004)(5.7697,0.50504994)
\psbezier[linewidth=0.0207,linecolor=greenblue](2.8727,1.75105)(3.7594,1.82095)(4.3394,2.2209501)(4.3394,1.46095)(4.3394,0.70095)(3.7594,1.10095)(2.8794,1.18095)
\psbezier[linewidth=0.0216,linecolor=greenblue](9.1247,-2.00795)(8.5867,-2.26095)(8.2527,-1.29595)(8.8897,-1.0539501)
\psbezier[linewidth=0.0218,linecolor=greenblue](9.1247,-2.00795)(9.8247,-1.8959501)(9.5647,-0.78195006)(8.8897,-1.0539501)
\psbezier[linewidth=0.0218,linecolor=greenblue](2.8794,1.18095)(2.2394,1.24095)(1.4794,0.58095)(1.4794,1.42095)(1.4794,2.26095)(2.2594,1.70095)(2.9194,1.7609501)
\psbezier[linewidth=0.0218,linecolor=greenblue](4.5207,-0.66495)(4.6317,-0.48495004)(4.2677,0.08304997)(4.4597,0.32904997)(4.6517,0.57505)(5.1794,0.44095004)(5.7594,0.40095004)
\psbezier[linewidth=0.0218,linecolor=greenblue](5.7907,0.40004995)(6.2794,0.42095003)(6.7747,0.90405)(7.0594,0.28095004)(7.3441,-0.3421499)(6.4194,-0.47904998)(6.4594,-0.85905)
\psbezier[linewidth=0.0218,linecolor=greenblue](4.5207,-0.66495)(4.2137,-1.54995)(4.2387,-1.79995)(4.5817,-1.90895)(4.9247,-2.01795)(5.1394,-1.69905)(5.0394,-0.77905)
\psbezier[linewidth=0.0218,linecolor=greenblue](6.4697,-0.84995)(6.3677,-0.98895)(6.8637,-1.8449501)(6.1347,-1.9619501)(5.4057,-2.07895)(5.3994,-1.15905)(5.9994,-0.59905)
\psbezier[linewidth=0.0218,linecolor=greenblue](5.0467,-0.76395005)(5.0437,-0.50995004)(5.4987,-0.19895004)(5.7367,-0.16395003)(5.9747,-0.12895003)(6.1517,-0.39195004)(5.9997,-0.62095004)
\psline[linewidth=0.02cm](2.6757,0.05704997)(1.0757,0.05204997)
\psline[linewidth=0.02cm](1.0757,0.05204997)(0.2757,1.45205)
\psline[linewidth=0.02cm](1.0757,0.05204997)(0.2757,-1.54795)
\psline[linewidth=0.06cm](1.8757,0.45204997)(1.8757,-0.34795004)
\psline[linewidth=0.02cm](7.8757,1.45705)(9.8757,1.45205)
\psline[linewidth=0.06cm](8.8757,1.85205)(8.8757,1.05205)
\psline[linewidth=0.06cm](0.3594,-0.55284995)(0.9394,-1.0128499)
\psline[linewidth=0.06cm](0.9394,1.00715)(0.3594,0.52715003)
\pscircle[linewidth=0.019199999,dimen=outer,fillstyle=solid,fillcolor=color296b](5.8757,1.45705){0.2782}
\psframe[linewidth=0.02,dimen=outer,fillstyle=solid,fillcolor=green](4.1417,1.72305)(3.6097,1.1910499)
\pscircle[linewidth=0.02,dimen=outer,fillstyle=solid,fillcolor=blue](2.6757,0.05704997){0.2782}
\pscircle[linewidth=0.0225,dimen=outer,fillstyle=solid,fillcolor=blue](1.9287,-1.54295){0.2782}
\pscircle[linewidth=0.0225,dimen=outer,fillstyle=solid,fillcolor=blue](3.4167,-1.54295){0.2782}
\psframe[linewidth=0.0225,dimen=outer,fillstyle=solid,fillcolor=green](5.1417,0.32304996)(4.6097,-0.20895003)
\pscircle[linewidth=0.0225,dimen=outer,fillstyle=solid](6.6757,0.05704997){0.2782}
\pscircle[linewidth=0.0225,dimen=outer,fillstyle=solid](5.9287,-1.54295){0.2782}
\psframe[linewidth=0.0225,dimen=outer,fillstyle=solid,fillcolor=green](9.3417,-1.27695)(8.8097,-1.8089501)
\pscircle[linewidth=0.0185,dimen=outer,fillstyle=solid](1.8757,1.45705){0.2782}
\pscircle[linewidth=0.0185,dimen=outer,fillstyle=solid,fillcolor=color306b](7.4757,-1.54295){0.2782}
\pscircle[linewidth=0.0185,dimen=outer,fillstyle=solid](4.6757,-1.54295){0.2782}
\pscircle[linewidth=0.020299999,dimen=outer,fillstyle=solid,fillcolor=blue](1.0757,0.05204997){0.2757}
\pscircle[linewidth=0.020299999,dimen=outer,fillstyle=solid,fillcolor=blue](0.2757,1.45205){0.2757}
\pscircle[linewidth=0.020299999,dimen=outer,fillstyle=solid,fillcolor=blue](0.2757,-1.54795){0.2757}
\pscircle[linewidth=0.0225,dimen=outer,fillstyle=solid,fillcolor=yellow](7.8754,1.44895){0.2782}
\pscircle[linewidth=0.0225,dimen=outer,fillstyle=solid,fillcolor=yellow](9.0754,0.04895003){0.2782}
\pscircle[linewidth=0.020299999,dimen=outer,fillstyle=solid,fillcolor=yellow](9.8754,1.44395){0.2757}
\end{pspicture} 
}}%
    }
    \caption{
         A mountain connection tree
         $T$ of the mountain structure $(G,\varphi,c,{\mathcal H},\mathcal{S})$
         and the
         crest separators in $\mathcal{S}$; the latter
         denoted by straight lines
         crossing the edges of the tree. 
         Thick straight lines denote crest separators that are pseudo shortcut free. 
         Squared
         vertices represent
         pseudo shortcut free components.
         The curves denote coast separators
         constructed for the crests in $\mathcal H$ that are in the pseudo shortcut free
         $(S,\varphi)$-components.
         If such a curve in the figure encloses some nodes representing
         an $(S,\varphi)$-component $C$, this should mean that the
         corresponding
         coast separator encloses the
         crest in $\mathcal H$ contained in $C$. 
         The non-white round vertices are the nodes of the forest $\tilde{F}$.         
         }
    \label{fig:alg_init2}
\vspace{8mm}
     \centering{
\scalebox{1.2} 
{
\begin{pspicture}(0,-2.615025)(8.52995,2.6149749)
\definecolor{brightorange}{rgb}{1.0,0.7,0.0}
\definecolor{brightblue}{rgb}{0.7,0.7,1.0}
\definecolor{color3142}{rgb}{0.0,0.0,0.6}
\psline[linewidth=0.02cm]{<-}(2.4263,0.8463)(1.5353,-0.1932)
%
\psline[linewidth=0.02cm]{<-}(4.3495,1.5642)(2.8043,1.1006)
\psline[linewidth=0.02cm]{<-}(4.8283,1.5642)(6.3735,1.1006)
\psline[linewidth=0.02cm]{<-}(1.2832,-0.5904)(0.7370,-1.7602)
\psline[linewidth=0.02cm]{<-}(1.4939,-0.5908)(2.0354,-1.7599)
\psline[linewidth=0.02cm]{<-}(2.7342,0.8326)(3.4582,-0.1810)
\psline[linewidth=0.02cm]{<-}(4.5889,1.3861)(4.5889,0.5079)
\psline[linewidth=0.02cm]{<-}(6.4263,0.8463)(5.5353,-0.1932)
\psline[linewidth=0.02cm]{<-}(6.5889,0.7860249)(6.5889,-0.6411)
\psline[linewidth=0.02cm]{<-}(6.7515,0.8463)(7.6425,-0.1932)
\psline[linewidth=0.02cm]{<-}(5.2832,-0.5904)(4.7370,-1.7602)
\psline[linewidth=0.02cm]{<-}(7.7889,-0.6139)(7.7889,-1.7391)
%
%
\psbezier[linewidth=0.02,linecolor=brightblue](0.815875,-1.028975)(0.282875,-1.460975)(0.0,-2.165025)(0.417875,-2.4119751)(0.83575,-2.6589253)(1.213875,-2.064975)(1.211875,-1.304975)
\psbezier[linewidth=0.02,linecolor=brightblue](1.815875,0.50202495)(0.54,-0.14502494)(0.695875,-0.41097507)(1.26,-0.94502497)(1.824125,-1.4790748)(1.6,-2.245025)(1.838875,-2.468975)
\psbezier[linewidth=0.02,linecolor=color3142](1.629875,-1.271975)(1.78,-1.7450249)(1.72,-2.605025)(2.26,-2.3450248)(2.8,-2.0850248)(2.4,-1.685025)(1.894875,-1.0789751)
\psbezier[linewidth=0.02,linecolor=brightblue](2.143875,0.11802494)(1.386875,-0.88697505)(2.307875,-1.1919751)(2.651875,-1.886975)
\psbezier[linewidth=0.02,linecolor=brightblue](1.838875,-2.468975)(2.14,-2.705025)(2.82,-2.3250248)(2.651875,-1.886975)
\psbezier[linewidth=0.02,linecolor=color3142](3.510875,1.6380249)(2.7,1.6349751)(2.129875,1.693025)(2.08,0.91497505)(2.030125,0.13692518)(1.36,0.41497505)(1.013875,-0.25997508)
\psbezier[linewidth=0.02,linecolor=color3142](1.013875,-0.25997508)(0.945875,-0.64997506)(1.403875,-0.93297505)(1.640875,-0.67297506)
\psbezier[linewidth=0.02,linecolor=color3142](1.640875,-0.67297506)(2.199875,0.23702493)(1.409875,0.03302494)(3.640875,1.016025)
\psbezier[linewidth=0.02,linecolor=brightblue](2.872875,0.10702494)(3.250875,-0.7089751)(3.56,-1.085025)(3.88,-0.70502496)(4.2,-0.32502493)(3.925875,0.008024939)(3.273875,0.49702495)
\psbezier[linewidth=0.02,linecolor=color3142](4.329875,0.94302493)(4.030875,0.07002494)(4.04,-0.18502493)(4.54,-0.18502493)(5.04,-0.18502493)(5.16,0.03497506)(4.854875,0.93702495)
\psbezier[linewidth=0.0199,linecolor=brightblue](5.549875,1.124025)(6.371875,1.124025)(5.68,0.6949751)(5.18,0.13497506)(4.68,-0.42502493)(4.993875,-0.6899751)(4.611875,-1.169975)
\psbezier[linewidth=0.0199,linecolor=brightblue](5.651875,1.575025)(6.96,1.434975)(7.1,1.554975)(6.92,0.53497505)(6.74,-0.48502493)(7.26,-0.90502495)(6.9,-1.305025)
\psbezier[linewidth=0.0199,linecolor=color3142](6.419875,-0.25497505)(6.08,-0.9850249)(6.26,-1.305025)(6.6,-1.305025)(6.94,-1.305025)(7.06,-1.0050249)(6.803875,-0.25997508)
\psbezier[linewidth=0.0199,linecolor=brightblue](6.9,-1.305025)(6.4,-1.7650249)(6.0,-1.1050249)(6.2,-0.60502493)(6.4,-0.10502494)(6.58,0.83497506)(5.984875,-0.23197506)
\psbezier[linewidth=0.0199,linecolor=color3142](5.78,0.55497503)(5.2,0.09497506)(4.7633305,-0.44310445)(5.16,-0.78502494)(5.5566697,-1.1269454)(5.9,-0.12502494)(6.1,0.21497506)
\psbezier[linewidth=0.0199,linecolor=brightblue](5.984875,-0.23197506)(5.781875,-0.66397506)(5.372875,-1.0759751)(5.228875,-1.376975)(5.084875,-1.677975)(5.26,-2.105025)(4.849875,-2.474975)
\psbezier[linewidth=0.0199,linecolor=brightblue](4.849875,-2.474975)(4.66,-2.625025)(4.439875,-2.5529752)(4.24,-2.3050249)(4.040125,-2.0570748)(4.28,-1.565025)(4.611875,-1.169975)
\psbezier[linewidth=0.0199,linecolor=brightorange](2.583875,1.4340249)(4.18,1.194975)(3.996875,2.705025)(5.662875,1.6710249)
\psbezier[linewidth=0.0199,linecolor=brightorange](2.583875,1.4340249)(2.44,1.454975)(2.226875,1.254025)(2.219875,1.0480249)(2.212875,0.8420249)(2.36,0.5749751)(2.555875,0.60902494)
\psbezier[linewidth=0.0199,linecolor=brightorange](2.555875,0.60902494)(2.9,0.6949751)(2.969875,0.98002493)(3.444875,1.103025)(3.919875,1.226025)(4.22,1.3549751)(4.453875,1.101025)
\psbezier[linewidth=0.0199,linecolor=brightorange](4.453875,1.101025)(4.54,0.9349751)(4.36,0.7549751)(4.304875,0.5920249)(4.24975,0.42907482)(4.156875,0.18402494)(4.3,0.014975061)(4.443125,-0.15407482)(5.272875,-0.21497506)(4.803875,0.6940249)
\psbezier[linewidth=0.0199,linecolor=brightorange](5.521875,1.0330249)(4.996875,1.220025)(4.637875,1.050025)(4.803875,0.6940249)
\psbezier[linewidth=0.0199,linecolor=color3142](4.832875,-1.0459751)(4.6,-1.4850249)(4.5,-1.425025)(4.332875,-1.753975)(4.16575,-2.082925)(4.38,-2.3450248)(4.48,-2.365025)
\psbezier[linewidth=0.0199,linecolor=color3142](4.48,-2.365025)(4.68,-2.425025)(5.04,-2.3050249)(5.02,-1.965025)(5.0,-1.6250249)(5.072875,-1.482975)(5.142875,-1.2659751)
\psbezier[linewidth=0.0199,linecolor=brightblue](7.035875,0.12402494)(7.38,-0.64502496)(7.72,-1.165025)(8.12,-0.6850249)(8.52,-0.20502494)(7.94,0.17497507)(7.351875,0.50202495)
\psbezier[linewidth=0.0199,linecolor=brightblue](7.526875,-1.1639751)(7.314875,-1.9689751)(7.3,-2.405025)(7.8,-2.405025)(8.3,-2.405025)(8.289875,-1.8959751)(8.040875,-1.1639751)
\pscircle[linewidth=0.019199999,dimen=outer,fillstyle=solid,fillcolor=red](4.588875,1.636025){0.2499}
\usefont{T1}{ptm}{m}{n}
\rput(4.588875,1.636025){\scalebox{0.8333}{$R$}}
\pscircle[linewidth=0.02,dimen=outer,fillstyle=solid](2.588875,1.0360249){0.2499}
\pscircle[linewidth=0.02,dimen=outer,fillstyle=solid](1.388875,-0.36397505){0.2499}
\usefont{T1}{ptm}{m}{n}
\rput(1.388875,-0.36397505){\scalebox{0.8333}{$C_2$}}
\pscircle[linewidth=0.0225,dimen=outer,fillstyle=solid](6.588875,1.0360249){0.2499}
\pscircle[linewidth=0.0225,dimen=outer,fillstyle=solid](0.641875,-1.9639751){0.2499}
\usefont{T1}{ptm}{m}{n}
\rput(0.641875,-1.9639751){\scalebox{0.8333}{$C_3$}}
\pscircle[linewidth=0.0225,dimen=outer,fillstyle=solid](2.129875,-1.9639751){0.2499}
\pscircle[linewidth=0.0225,dimen=outer,fillstyle=solid](3.588875,-0.36397505){0.2499}
\usefont{T1}{ptm}{m}{n}
\rput(3.588875,-0.36397505){\scalebox{0.8333}{$C_5$}}
\pscircle[linewidth=0.0225,dimen=outer,fillstyle=solid](4.588875,0.28302494){0.2499}
\usefont{T1}{ptm}{m}{n}
\rput(4.588875,0.28302494){\scalebox{0.8333}{$C_6$}}
\pscircle[linewidth=0.0225,dimen=outer,fillstyle=solid](5.388875,-0.36397505){0.2499}
\usefont{T1}{ptm}{m}{n}
\rput(5.388875,-0.36397505){\scalebox{0.8333}{$C_8$}}
\pscircle[linewidth=0.0225,dimen=outer,fillstyle=solid](6.588875,-0.8659751){0.2499}
\pscircle[linewidth=0.0225,dimen=outer,fillstyle=solid](7.788875,-0.36397505){0.2499}
\pscircle[linewidth=0.0225,dimen=outer,fillstyle=solid](4.641875,-1.9639751){0.2499}
\pscircle[linewidth=0.0225,dimen=outer,fillstyle=solid](7.788875,-1.9639751){0.2499}
\usefont{T1}{ptm}{m}{n}
\rput(2.129875,-1.9639751){\scalebox{0.8333}{$C_4$}}
\usefont{T1}{ptm}{m}{n}
\rput(4.641875,-1.9639751){\scalebox{0.8333}{$C_9$}}
\usefont{T1}{ptm}{m}{n}
\rput(6.588875,-0.8659751){\scalebox{0.8333}{$C_{0}$}}
\usefont{T1}{ptm}{m}{n}
\rput(7.788875,-1.983975){\scalebox{0.8333}{$C_{B}$}}
\usefont{T1}{ptm}{m}{n}
\rput(7.788875,-0.36397505){\scalebox{0.8333}{$C_{A}$}}
\psline[linewidth=0.06cm](3.657875,0.93702495)(3.493875,1.7340249)
\psline[linewidth=0.06cm](1.758875,0.58202493)(2.205875,0.04102494)
\psline[linewidth=0.06cm](2.798875,0.050024938)(3.352875,0.5870249)
\psline[linewidth=0.06cm](4.244875,0.94302493)(4.922875,0.94302493)
\psline[linewidth=0.06cm](5.510875,0.93102497)(5.685875,1.751025)
\psline[linewidth=0.06cm](5.671875,0.71102494)(6.274875,-0.016975062)
\psline[linewidth=0.06cm](1.476875,-1.383975)(2.041875,-0.9779751)
\psline[linewidth=0.06cm](4.668875,-0.94397503)(5.301875,-1.3619751)
\psline[linewidth=0.06cm](6.255875,-0.25497505)(6.961875,-0.25497505)
\psline[linewidth=0.06cm](7.447875,0.63202494)(6.961875,0.02802494)
\psline[linewidth=0.06cm](7.430875,-1.169975)(8.142875,-1.1639751)
\psline[linewidth=0.06cm](1.307875,-1.3729751)(0.719875,-0.9719751)
\usefont{T1}{ptm}{m}{n}
\rput(1.2735938,1.3849751){\scalebox{0.8333}{$T^*:$}}
\usefont{T1}{ptm}{m}{n}
\rput(6.588875,1.0360249){\scalebox{0.8333}{$C_7$}}
\usefont{T1}{ptm}{m}{n}
\rput(2.588875,1.0360249){\scalebox{0.8333}{$C_1$}}
\end{pspicture} 
}%
    }
    \caption{%
         A bold straight line represents a crest separator
         and a curve together with a straight line represents
         a cycle: either a composed cycle or
         an induced cycle (by the edges of the essential
          boundary of a crest separator).
          If
         such a cycle
         encloses a node being
         an
         $(\mathcal{S},\varphi)$-component $C$,
         this should mean that the cycle  encloses at least
         the crest in $\mathcal H$ that is contained in $C$. 
         Running the algorithm from above, 
          the cycles %
         $Y_R, Y_{C_2}, Y_{C_3}, Y_{C_5}, Y_{C_7}, Y_{C_A},$ and  $Y_{C_B}$
         corresponding to the curves starting in the 
         $(\mathcal{S},\varphi)$-components $R$, $C_2$, $C_3$, $C_5$,
         $C_7$, $C_A$, and $C_B$, respectively,
         are added to ${\mathcal{Y}}$ in Step~3. 
           Afterwards,
          each node of~$\tilde{T}$ is enclosed by
          a cycle in ${\mathcal{Y}}$.
             }
    \label{fig:algPH}
  \end{figure}%

\medskip\medskip

{\bf Analysis of the Algorithm.}  We now prove properties 
of the coast separators constructed by the algorithm above.

\begin{observation}\label{obs:enclosed}
Each crest
in ${\mathcal H}$
is enclosed by a coast separator
of~${\mathcal{Y}}$.
\end{observation}

\begin{lemma}\label{lem:coastSepExist}
For all crests
$H \in {\mathcal H}$ that are
considered in Step~1, there exists a minimal coast
separator in the plane graph $(C',\psi)$ constructed for~$H$.
\end{lemma}

\begin{proof}
Let $C$ be the $(\mathcal{S}',\varphi)$-component that contains $H$.
Since the vertices of the coast have upper height at most
$c_{\mathrm{max}}$ and since $H$ has lower height at least
$k+c_{\mathrm{max}}+1$, 
by Theorem~\ref{the:Sep}, there is
a coast separator for~$H$ in $G$ of weighted length at most~$k$.
Thus, 
there is also such a coast separator $Y$ that is minimal. 
Since no crest separator $X\in\mathcal{S'}$ encloses
$C$, the $(X,\varphi)$-component $\DD$ containing $C$ contains
at least one vertex of the coast of $G$. Hence 
$Y$ and $\DD$ are not
 vertex disjoint. Moreover, by Lemma~\ref{lem:StrongOptimal}, $Y$ can not cross 
$X$. Since this holds for all choices of $X$ in $\mathcal{S}'$,
$Y$ is completely contained in
 $C$. However, $Y$
may contain vertices of the boundary of~$C$ 
(for some appropriate definition of the boundary),
which is the reason
for running the algorithm of Lemma~\ref{lem:generalMaxFlowMinCut} on
the plane graph
$(C',\psi)$. 
\end{proof}

\begin{lemma}\label{lem:mconnected} For a coast separator $Y\in\mathcal{Y}$,
let $\mathcal C$ be the $(S,\varphi)$-components with the crests in
${\mathcal H}$ that are
enclosed by $Y$. Then, $\mathcal C$ 
induces a connected subtree of the mountain connection tree of
$(G,\varphi,c,{\mathcal H},\mathcal{S})$.
\end{lemma}

\begin{proof}
The lemma clearly holds if 
 $Y$ is the essential boundary of a crest separator with an interior
lowpoint. So from now on, we only consider the remaining kinds of coast
 separators.

Let $Y$ be a minimum coast separator constructed for a crest $H$ in Step
1.
By Lemma~\ref{lem:StrongOptimal}, whenever
$Y$
crosses a crest
separator $X$, the part $P'$ of~$Y$ contained in the
$(X,\varphi)$-component not containing $H$ is a pseudo shortcut for
a 
crest-separator path of $X$
contained in the inner graph of~$Y$.

Analogously, let $P$ be a nice pseudo shortcut taken for the construction of 
a composed cycle $Y$ as a coast separator
for a crest $H'$ in Step~2.
Since $P$ is nice,  whenever $P$ crosses a crest
separator $X$, the part $P'$ of~$Y$ contained in the  
$(X,\varphi)$-component not containing $H'$ is a pseudo shortcut.

The lemma now follows from Lemma~\ref{lem:enclosev}
and the fact that each coast separator
enclosing two crests in $\mathcal H$ contained in different $(\mathcal{S},\varphi)$-components
$C_1$ and $C_2$ must cross all crest separators separating two consecutive
$(\mathcal{S},\varphi)$-components on the path from~$C_1$ to $C_2$ in
the mountain connection tree.
\end{proof}

We next want to show that no crest $H \in \mathcal H$ is enclosed by more than one coast
separator in~$\mathcal Y$. We start with two auxiliary lemmas.

\begin{lemma}\label{lem:interact_eins_zwei}
Let $C^*$ be an $(\mathcal{S},\varphi)$-component with a crest
$H^* \in \mathcal H$. If $Y_{C^*}$ is added to $\mathcal Y$
in Step~3 and if $Y_{C^*}$ encloses a crest $H'\in \mathcal H$ in an
$(\mathcal{S},\varphi)$-component $C'$,
then $H'$ is neither part of a pseudo shortcut free
$(\mathcal{S},\varphi)$-component nor
enclosed by a coast separator constructed for a crest
in Step~1.
\end{lemma}

\begin{proof}
The fact that $Y_{C^*}$ encloses $H'$ implies that $H'$ can not be part of a pseudo shortcut free
$(\mathcal{S},\varphi)$-component.
Assume now that a coast separator $Y_C$ constructed for a crest $H\in \mathcal H$ in a
pseudo shortcut free component ${C}$ encloses $H'$ and assume w.l.o.g.~that
$H'$ is chosen as a crest enclosed by $Y_{C^*}$ with this property such
that the distance between the $(\mathcal{S},\varphi)$-components $C'$ and $C^*$ is as small as possible. 
Then, $H^*\neq H'$ since, otherwise, $Y_{C^*}$ would not be added to $\mathcal Y$
in Step~3. 
Moreover, $Y_{C^*}$ can not be the essential boundary of a crest separator with an interior lowpoint. 
Otherwise, 
for the subtree $T'$ of~$T$ with
root $C^*$ that contains $C'$, all its 
$(\mathcal{S},\varphi)$-components are enclosed by a crest
 separator with an interior lowpoint.
Since $Y_C$ enclosing $H'$ but not
$H^*$ implies that $C$ is in $T'$, %
we then obtain
a contradiction to the fact that
$C$ is pseudo shortcut free. 
Thus, $Y_{C^*}$ is a cycle consisting of some nice
pseudo shortcut $P'$ with its crest-separator path. See
Fig.~\ref{fig:enclosed12}. 
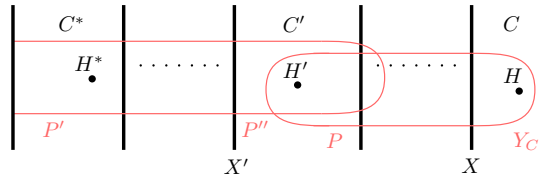
\begin{figure}[b!]
    \begin{center}
       \scalebox{0.8}{
\scalebox{1} 
{
\begin{pspicture}(0,-1.4572693)(7.1095314,1.4453692)
\psset{xunit=1.3cm,yunit=1.0cm,runit=1cm}
\definecolor{color511}{rgb}{1.0,0.4,0.4}
\usefont{T1}{ptm}{m}{n}
\rput(2.8248124,-1.2544568){$X'$}
\usefont{T1}{ptm}{m}{n}
\rput(6.3046875,1.1056681){$C$}
\psline[linewidth=0.056199998cm](5.8,1.4172693)(5.8,-0.9827307)
\psline[linewidth=0.04cm,linestyle=dotted,dotsep=0.16cm](4.6,0.4172693)(5.6,0.4172693)
\usefont{T1}{ptm}{m}{n}
\rput(5.7948126,-1.2544568){$X$}
\psbezier[linewidth=0.02,linecolor=color511](4.0,0.6172693)(3.6,0.6172693)(3.2,0.6172693)(3.2,0.0172693)(3.2,-0.5827307)(3.6,-0.5827307)(4.0,-0.5827307)
\psline[linewidth=0.056199998cm](2.8,1.4172693)(2.8,-0.9827307)
\usefont{T1}{ptm}{m}{n}
\rput(4.0646877,-0.83418304){\color{color511}$P$}
\psdots[dotsize=0.12](6.4,-0.008172693)
\usefont{T1}{ptm}{m}{n}
\rput(6.3246875,0.2456681){$H$}
\psline[linewidth=0.056199998cm](4.4,1.4172693)(4.4,-0.9827307)
\psdots[dotsize=0.12](3.6,0.0922693)
\usefont{T1}{ptm}{m}{n}
\rput(3.5746875,0.3456681){$H'$}
\psbezier[linewidth=0.02,linecolor=color511](5.8,0.6172693)(6.2,0.6172693)(6.6,0.6172693)(6.6,0.0172693)(6.6,-0.5827307)(6.2,-0.5827307)(5.8,-0.5827307)
\psline[linewidth=0.02cm,linecolor=color511](5.8,0.6172693)(3.8,0.6172693)
\psline[linewidth=0.02cm,linecolor=color511](5.8,-0.5827307)(3.8,-0.5827307)
\usefont{T1}{ptm}{m}{n}
\rput(6.4946876,-0.83418304){\color{color511}$Y_C$}
\psline[linewidth=0.056199998cm](1.4,1.4172693)(1.4,-0.9827307)
\psline[linewidth=0.04cm,linestyle=dotted,dotsep=0.16cm](1.6,0.4172693)(2.6,0.4172693)
\usefont{T1}{ptm}{m}{n}
\rput(3.5546875,1.1056681){$C'$}
\usefont{T1}{ptm}{m}{n}
\rput(0.74468756,1.1056681){$C^*$}
\rput(0.7,0){%
\psbezier[linewidth=0.02,linecolor=color511](3.2,0.8172693)(3.6,0.8172693)(4.0,0.8172693)(4.0,0.2172693)(4.0,-0.3827307)(3.6,-0.3827307)(3.2,-0.3827307)
}%
\psdots[dotsize=0.12](1.0,0.1922693)
\usefont{T1}{ptm}{m}{n}
\rput(0.9646875,0.4456681){$H^*$}
\psline[linewidth=0.02cm,linecolor=color511](3.9,0.8172693)(0.0,0.8172693)
\psline[linewidth=0.02cm,linecolor=color511](3.9,-0.3827307)(0.0,-0.3827307)
\psline[linewidth=0.056199998cm](0.0,1.4172693)(0.0,-0.9827307)
\usefont{T1}{ptm}{m}{n}
\rput(0.5146875,-0.634183){\color{color511}$P'$}
\usefont{T1}{ptm}{m}{n}
\rput(3.0646875,-0.634183){\color{color511}$P''$}
\end{pspicture} 
}}
    \end{center}
    \caption{Coast separators constructed in Step~1 and in Step~2 enclosing a crest
    $H'$.}
    \label{fig:enclosed12}
  \end{figure}%
Let $X$ be the crest separator with a top edge in~$C$
disconnecting $H$ and $H'$, and let $D$ be the
$(X,\varphi)$-component containing $C'$. $Y_C$ as a minimal coast separator
has a $D$-pseudo shortcut $P$ of~$X$ for a 
crest-separator path
of~$X$ 
contained 
in the inner graph of %
$Y_C$
as a subpath (Lemma~\ref{lem:StrongOptimal}). Moreover, let
$X'$ be the crest separator with a top edge in $C'$ disconnecting $H^*$ and
$H'$, and let $D'$ be the 
$(X',\varphi)$-component containing $C'$. Note that $X'\neq X$ since
we have chosen $H'$ in such a way that $C'$ and $C^*$ have minimal distance.
Since $P'$ is nice, a subpath $P''$ of~$P'$ is a $D$-pseudo shortcut for~$X'$
for a crest-separator path of~$X'$ contained in the inner graph of~$Y_{C^*}$.
The existence of~$P$  and $P''$ is a contradiction to
Lemma~\ref{lem:SinglePS}. For that also note that
$P$ does not cross $X'$ since $H'$ was chosen such that the
distance between $C'$ and $C^*$ is as small as possible, and
$P''$ does not cross $X$ since $C$ is pseudo shortcut free.
\end{proof}

\begin{lemma}\label{lem:directedEdges}
Let $C$ and $C^*$ be two $(\mathcal{S},\varphi)$-components whose
crests in $\mathcal H$ are not enclosed  
by a coast separator constructed in Step~1. If the coast separator $Y_{C^*}$
for the crest $H^*\in \mathcal H$ in $C^*$
encloses the crest $H\in \mathcal H$ in $C$, then the directed forest $\tilde{F}$ constructed by our algorithm has a
directed path from~$C$ to~$C^*$.
\end{lemma}

\begin{proof}
Assume that $Y_{C^*}$ encloses $H$.
Let $C_0=C,C_1,\ldots,C_r=C^*$ be the path in the mountain connection tree
connecting $C$ and $C^*$. 
For $i\in\{1,\ldots,r\}$, let $X_i \in \mathcal{S}$ be
 the crest separator with a
top edge part of both $C_{i-1}$ and $C_i$. The situation is sketched in
 Fig.~\ref{fig:pathMCT}, but this time, ignore the shown pseudo shortcuts.
Then, the crest~$H_i\in \mathcal H$ in~$C_i$ ($i=\{0,\ldots,r\}$) is enclosed
by $Y_{C^*}$ (Lemma~\ref{lem:mconnected}), but are neither 
part of a pseudo shortcut free
$(\mathcal{S},\varphi)$-component nor
enclosed by a coast separator
constructed 
in Step~1 (Lemma~\ref{lem:interact_eins_zwei}). Thus, Condition 1 in
Step~2 holds for all edges $(C_i,C_{i+1})$ ($i=\{0,\ldots,r-1\}$).
Let $D$ be the $(\mathcal{S},\varphi)$-component of~$X_{C^*}$ containing
$C^*$, where $X_{C^*}$ is the crest separator defined for $C^*$ in Step~2.

We first consider the case where $Y_{C^*}$ is defined in Step~2.A, i.e.,
$Y_{C^*}$ is the essential boundary of
the crest separator $X_{C^*}$. 
Then, $X_{i+1}$ ($i\in \{0,\ldots,r-1\}$) has an interior lowpoint and encloses
$C_i$, but not $C_{i+1}$ (Condition 2 holds). Moreover, the 
$D^*$-pseudo shortcut set for~$X_{i+1}$ is empty by the
definition of pseudo shortcuts, where $D^*$
is the $(X_{i+1} , \varphi)$-component containing $C_{i+1}$ (Condition 3
holds).
As a consequence, our algorithm adds
edges $(C_i,C_{i+1})$ to $\tilde{F}$ for all $i=\{0,\ldots,r-1\}$.

Next, we consider the case where $Y_{C^*}$ is defined in Step~2.B, i.e.,
$Y_{C^*}$ is the cycle consisting of some nice $D$-pseudo shortcut $P'$ with its
crest-separator path of a crest separator $X_{r+1}$ with a top edge in $C_r$.
Note also that either $C_i$ ($i=\{1,\ldots,r\}$) is not enclosed
by $X_i$ or $X_i$ has no interior lowpoint (Condition 2 holds) since, otherwise,
${C^*}$ would be 
also enclosed by $X_i$ and $X_r$ and hence defined in Step~2.A.
For each
 crest
separator $X_{i+1}$
 ($i=\{1,\ldots,r\}$), a subpath $P$ of~$P'$ is a
$D'$-pseudo shortcut 
 where $D'$ is the $(X_{i+1},\varphi)$-component containing $C_i$.
Since $P'$ is a pseudo
 shortcut %
in ${\mathcal Z}$ and since 
 ${\mathcal Z}$ is consistent, the subpath of $P'$ starting and
 ending on $X_{i+1}$ is a pseudo shortcut in %
${\mathcal Z}$. Since 
 ${\mathcal Z}$ is also non-overlapping, 
$X_{i+1}$ is the only crest separator with a
top edge in $C_i$ that has pseudo shortcuts, i.e., Condition 3 is satisfied
for the edge $(C_{i-1},C_{i})$, i.e.,
our algorithm adds
edges $(C_i,C_{i+1})$ to $\tilde{F}$ for all $i=\{0,\ldots,r-1\}$.
\end{proof}

\begin{lemma}\label{lem:mdisjoint}
No crest $H\in \mathcal H$ is enclosed by more than one coast separator in~$\mathcal{Y}$.
\end{lemma}

\begin{proof}
A coast separator  constructed in Step~1 can not enclose a crest 
in ${\mathcal H}$
that is
also enclosed by another coast separator added to $\mathcal{Y}$ in Step~1
(Lemma~\ref{lcor:interactPath2}) or in Step~3
(Lemma~\ref{lem:interact_eins_zwei}).
By Lemma~\ref{lem:directedEdges}, a crest 
in ${\mathcal H}$
can not be enclosed by
two coast
separators in~$\mathcal{Y}$ that are constructed 
for~$(\mathcal{S},\varphi)$-components part of two different trees
of~$\tilde{F}$.

By our choice of 
directing the edges of~$E'$ in Step~2,
each coast separator $Y_C$  for a crest $H\in \mathcal H$ contained
in an $(\mathcal{S},\varphi)$-component~$C$ of
a tree $\tilde{T}$ of~$\tilde{F}$ can only enclose $(S,\varphi)$-components
below $C$ in $\tilde{T}$. By Lemma~\ref{lem:mconnected}
the $(\mathcal{S},\varphi)$-components of~$\tilde{T}$ with the crests 
in~$\mathcal{H}$ that 
are
enclosed by $Y_C$ induce a connected subtree of~$\tilde{T}$.
Therefore, 
two coast separators constructed in Step~2 for the same
tree $\tilde{T}$ and added to $\tilde{\mathcal{Y}}$
in Step~3 cannot enclose the same crest. 
\end{proof}

Finally, we analyze the 
running time. Take $G=(V,E)$. Recall that $G$ is $O(k)$-weighted
outerplanar and that $(G,\varphi,c,{\mathcal H},\mathcal{S})$ is a good
mountain structure for~$(G,\varphi,c)$.

\begin{lemma}
The algorithm for computing the set~$\mathcal Y$ runs in 
$O(|{\mathcal{H}}|k^3 + |V|k )$
time.
\end{lemma}

\begin{proof}
By Lemma~\ref{lem:ShortCutSet}, we can construct 
a consistent $(k-1)$-bounded $\mathcal{M}$-shortcut set
in
$O(|{\mathcal{H}}|k^3 + |V|k )$
time.
Clearly within the same time,
for each $(X,\varphi)$-component $\DD$, we can decide whether 
the $\DD$-pseudo shortcut set is empty and whether
the essential boundary of~$X$ can be used as a coast separator, i.e.,
whether $X$ has
an interior lowpoint.
Hence, again within in the same
time we can determine the subset $\mathcal{S}'$ of
all pseudo shortcut free crest separators
in $\mathcal{S}$, all pseudo shortcut free
$(\mathcal{S},\varphi)$-components 
as well as all $(\mathcal{S}',\varphi)$-components.
A minimal coast separator for a crest $H\in \mathcal H$
of a pseudo shortcut free
$(\mathcal{S},\varphi)$-component contained in an
$(\mathcal{S}',\varphi)$-component~$C$ 
can be computed
in $O(n'\ellk)$ time (Lemma~\ref{lem:generalMaxFlowMinCut})
where
$n'$ is the number of vertices
of~$C$.
In the same time, we can
determine the inner graph of~$Y_C$, all crests in $\mathcal H$ that are part of the inner graph, and the
set of~$(\mathcal{S},\varphi)$-components containing them.
Therefore the running time of Step~1 is bounded by 
$O(|{\mathcal{H}}|k^3 + |V|k )$.

To test whether an $(X,\varphi)$-component has an interior lowpoint
can be done in $O(k)$ time for each of the $O(|V|)$ crest separators.
Since the set $\mathcal{Z}$ can be constructed in $O(|\mathcal{H}|k^2)$
time 
and since $\tilde{F}$ consists of~$O(|V|)$ nodes, it is easy to
run Step~2 in $O(|V|k+|\mathcal{H}|k^2)$ time. 

The construction of the inner graphs of all coast separators $Y_C$ added to
$\mathcal{Y}$ in Step~3 as well as the
removal of all $(S,\varphi)$-components $C$ with
crests in $\mathcal H$ that are part of such an inner graph
can be done by a depth-first search in the inner graph of~$Y_C$ and
the running time can be bounded by the number of edges of the
 $(S,\varphi)$-components removed. Since each removal of 
crest separator removes a disjoint set of 
$(S,\varphi)$-components (Lemma~\ref{lem:mdisjoint}), the
running time of Step~3 is bounded by $O(|V|k)$.
\end{proof}

Note that the coast separators constructed in Step~1 have weighted
size at most~$k$, whereas the coast separators taken
in Step~3 can have a weighted size of at most $3k+4c_{\mathrm{max}}-5$ since the 
down paths
of each crest separator $X\in\mathcal{S}$ can have a weighted length
of at most $k+2c_{\mathrm{max}}-1$, since a pseudo shortcut does not start and end
with a vertex of the coast (which reduces the length of a 
subpath of a down path
that can
be
part of a coast separator by $1$), and since the pseudo shortcuts that we
use have weighted length $\le k-1$.

Let $n$ be the number of vertices of $G^+$. 
Since each crest
$H\in \mathcal H$ of $G$ is obtained from a merge of at least $\lceil
k/c_{\mathrm{max}}\rceil$ vertices
in $G^+$, and since $k\ge c_{\mathrm{max}}$ if $n\ge 2$, we have
$|\mathcal H|\le n \cdot 
c_{\mathrm{max}} /k$ if $n\ge 2$. 

For each $H \in {\mathcal H}$, let us choose one crest of height
$2k+2c_{\mathrm{max}}$ in
$(G^+,\varphi^+,c^+)$ that is merged into $H$. Define ${\mathcal H}^+$ as
the set of chosen crests. Note that with $(G,\varphi,c,\mathcal{H},\mathcal{S})$ being a good mountain structure $(G^+,\varphi^+,c^+,{\mathcal H}^+,\mathcal{S})$ must also be a good mountain structure.
Then, the next corollary summarizes the results of the current section, where
the function $m$ simply maps each coast separator $Y$ to the $(\mathcal{S},\varphi)$-components whose crests are enclosed by $Y$.

\begin{corollary}\label{cor:StrongSCut}
Assume that we are given the integer $k$, the 
$(2k+2c_{\mathrm{max}})$-weighted-outerplanar graph
$(G^+,\varphi^+,c^+)$,
and the set ${\mathcal H}^+$ of crests of height $2k+2c_{\mathrm{max}}$ as
defined before.
Then, in  
$O(n k^2 %
c_{\mathrm{max}} 
)$ time,
one can construct a good mountain structure
$\mathcal{M}=(G^+,\varphi^+,c^+,{\mathcal H}^+,\mathcal{S})$,
the mountain connection tree $T$ for~$\mathcal{M}$,
a set $\mathcal{Y}$ of coast separators
in $G^+$,
for each $Y\in\mathcal{Y}$, the inner graph
$I$ of~$Y$ and the corresponding embedding $\varphi|_I$
such that the properties below hold.%

\begin{itemize}

\item[\rm (i)] For each crest $H$ of
height
$2k+2c_{\mathrm{max}}$ in $G^+$%
---with $H$ not necessarily being contained in $\mathcal{H}^+$---%
there is exactly one
coast separator $Y\in\mathcal{Y}$
for~$H$
of weighted size
at most $3k+4c_{\mathrm{max}}-5$.
\item[{\rm (ii)}] For each pair of crests of height $2k+2c_{\mathrm{max}}$ in
$G^+$, the crests are %
either part of the inner graph of one $Y\in\mathcal{Y}$ or 
there is a crest separator $X\in \mathcal{S}$ strongly going 
between the crests.
\item[{\rm (iii)}]
There is
a function $m$
mapping each $Y\in\mathcal{Y}$ to a non-empty set of
$(\mathcal{S},\varphi^+)$-components
such that
\begin{itemize}
\item the elements of~$m(Y)$ considered as nodes of~$T$
induce a connected subgraph of~$T$,
\item
the subgraph of~$G^+$ obtained from the union of all
$(\mathcal{S},\varphi^+)$-components
in $m(Y)$
contains
the inner graph $I_Y$ of~$Y$ as a subgraph,
\item the set $m(Y)$ does not have any 
$(\mathcal{S},\varphi^+)$-component
as an element
that is also an element in $m(Y')$ for a coast separator $Y'\in\mathcal{Y}$
with $Y'\not=Y$.%
\end{itemize}
\end{itemize}
\end{corollary}

\section{A Tree Decompositions for the Components}\label{sec:bod}

   We first describe an 
   algorithm %
   for
   constructing a tree decomposition of width $3\ell-1$
   for a
   weighted almost triangulated
   $\ell$-outerplanar 
graph.
Then we modify the
   algorithm such that, %
  given a weighted graph $(G,\varphi,c)$, 
    $\mathcal{S} \subseteq \mathcal{S}({G},{\varphi},c)$,
   and
   an $(\mathcal{S},\varphi)$-component $C$,
it
   constructs a tree decomposition
   for~$\mathrm{ext}(C,\mathcal{S})$ of width at most $3\ell-1$
   with the following %
   property:
   For each crest separator $X\in\mathcal{S}$
   with a top edge in $C$, there is a bag containing all vertices of~$X$.

 It is easier to construct such a tree decomposition 
if the following {\it neighborhood property} holds.

\begin{itemize}
\item[(N)]
For each vertex $v$, there is at most one vertex
$u$ with $u\downvertex=v$. %
\end{itemize}

Starting with 
an almost triangulated weighted 
$\ell$-outerplanar graph
$(\tilde{G},\tilde{\varphi},\tilde{c})$
as an intermediate goal, we want to transform it in such a way
into an $\ell$-outerplanar weighted plane graph $(G,\varphi,c)$ 
with the neighborhood property (N)
 that a tree decomposition for $\tilde{G}$ can be easily obtained from
a tree decomposition for $G$.
The idea is to compute $G$ as a {\em reverse
minor} of~$\tilde{G}$, i.e., $\tilde{G}$ can be obtained from~$G$ by
iteratively merging
adjacent vertices and/or by removing vertices.

In order to guarantee property (N), we start  our transformation
by splitting
each vertex $v$ 
into a path %
as it is sketched in
Fig.~\ref{fig:bd}.
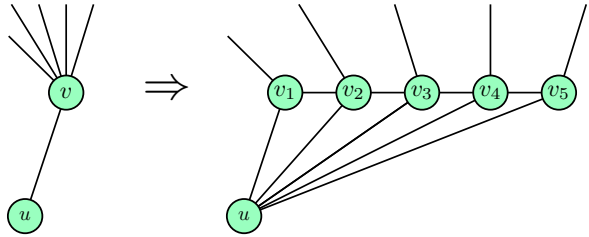
\begin{figure}[b!]
    \begin{center}
       \vspace{1mm}
       \scalebox{0.9}{
\scalebox{1} 
{
\begin{pspicture}(0,-1.844)(9.054375,1.8565)
\definecolor{lightgreen}{rgb}{0.6,1,0.75}
\psset{xunit=1.0cm,yunit=1.3cm,runit=1cm}
\psline[linewidth=0.025cm](4.241875,0.344)(5.241875,0.344)
\psline[linewidth=0.025cm](5.241875,0.344)(6.241875,0.344)
\psline[linewidth=0.025cm](6.241875,0.344)(7.241875,0.344)
\psline[linewidth=0.025cm](7.241875,0.344)(8.241875,0.344)
\psline[linewidth=0.025cm](1.041875,0.344)(0.441875,-1.056)
\psline[linewidth=0.025cm](4.241875,0.344)(3.641875,-1.056)
%
\psline[linewidth=0.025cm](1.041875,0.344)(0.201875,0.9844)
%
\psline[linewidth=0.025cm](4.241875,0.344)(3.401875,0.9844)
\psline[linewidth=0.025cm](1.041875,0.344)(0.241875,1.344)
\psline[linewidth=0.025cm](5.241875,0.344)(4.441875,1.344)
\psline[linewidth=0.025cm](1.041875,0.344)(0.641875,1.344)
\psline[linewidth=0.025cm](6.241875,0.344)(5.841875,1.344)
\psline[linewidth=0.025cm](1.041875,0.344)(1.041875,1.344)
\psline[linewidth=0.025cm](7.241875,0.344)(7.241875,1.344)
\psline[linewidth=0.025cm](8.241875,0.344)(8.641875,1.344)
\psline[linewidth=0.025cm](1.041875,0.344)(1.441875,1.344)
\psline[linewidth=0.025cm](5.241875,0.344)(3.641875,-1.056)
\psline[linewidth=0.025cm](6.241875,0.344)(3.641875,-1.056)
\psline[linewidth=0.025cm](6.241875,0.344)(3.641875,-1.056)
\psline[linewidth=0.025cm](7.241875,0.344)(3.641875,-1.056)
\psline[linewidth=0.025cm](8.241875,0.344)(3.641875,-1.056)
\pscircle[linewidth=0.026300002,dimen=outer,fillstyle=solid,fillcolor=lightgreen](1.041875,0.344){0.2652611}
\usefont{T1}{ptm}{m}{n}
\rput(1.041875,0.344){$v$}
\pscircle[linewidth=0.026300002,dimen=outer,fillstyle=solid,fillcolor=lightgreen](4.241875,0.344){0.2652611}
\usefont{T1}{ptm}{m}{n}
\rput(4.241875,0.344){$v_1$}
\pscircle[linewidth=0.026300002,dimen=outer,fillstyle=solid,fillcolor=lightgreen](5.241875,0.344){0.2652611}
\usefont{T1}{ptm}{m}{n}
\rput(5.241875,0.344){$v_2$}
\pscircle[linewidth=0.026300002,dimen=outer,fillstyle=solid,fillcolor=lightgreen](6.241875,0.344){0.2652611}
\usefont{T1}{ptm}{m}{n}
\rput(6.241875,0.344){$v_3$}
\pscircle[linewidth=0.026300002,dimen=outer,fillstyle=solid,fillcolor=lightgreen](7.241875,0.344){0.2652611}
\usefont{T1}{ptm}{m}{n}
\rput(7.241875,0.344){$v_4$}
\pscircle[linewidth=0.026300002,dimen=outer,fillstyle=solid,fillcolor=lightgreen](8.241875,0.344){0.2652611}
\usefont{T1}{ptm}{m}{n}
\rput(8.241875,0.344){$v_5$}
\pscircle[linewidth=0.026300002,dimen=outer,fillstyle=solid,fillcolor=lightgreen](0.441875,-1.056){0.2652611}
\usefont{T1}{ptm}{m}{n}
\rput(0.441875,-1.056){\scalebox{0.85}{$u$}}
\pscircle[linewidth=0.026300002,dimen=outer,fillstyle=solid,fillcolor=lightgreen](3.641875,-1.056){0.2652611}
\usefont{T1}{ptm}{m}{n}
\rput(3.641875,-1.056){\scalebox{0.85}{$u$}}
\usefont{T1}{ptm}{m}{n}
\rput(2.5135937,0.3665){\scalebox{2}{$\Rightarrow$}}
\end{pspicture} 
}%
       }
    \end{center}
    \caption{
    The replacement of a vertex $v$. 
    Let
    $u=v\downvertex$.
The vertices $v_1,\ldots,v_5$ are
         the copies of~$v$. %
             }
    \label{fig:bd}
  \end{figure}%
More precisely, iterate over the vertices $v$ with non-increasing lower height.
We consider only the more interesting case where $h^-_\varphi(v)\ge 2$.
Let $u=v\downvertex$.
Let $\{v,u\},\{v,u_1\},\ldots,\{v,u_d\}$
be the edges incident to $v$ in the clockwise order in which they appear
around $v$. %
Then $v$ is replaced by a series of vertices $v_1,\ldots,v_{d}$
with $v_1$ being incident to $u_1$ and
$v_2$, whereas $v_{d}$ is incident to $v_{d-1}$ and $u_d$, and
for~$i\in\{2,\ldots,d-1\}$, $v_i$ is incident to $v_{i-1},v_{i+1}$,
and $u_{i}$.
We finally connect all vertices
$v_1,\ldots,v_{d}$ to $u$ by an edge.
The vertices $v_1,\ldots,v_{d}$ are then called the {\em copies of~$v$}.
We %
finally %
define %
the {\cost} of the copies of a vertex $v$ as the {\cost} of~$v$.
Let $(G,\varphi,c)$ be the %
weighted plane graph
obtained. Thus, the copies of $v$ have the same height interval as $v$
itself, and $G$ is weighted $\ell$-outerplanar.
Note that strictly speaking no down vertices and down edges
are defined in $G$ since
down vertices where only defined for almost triangulated graphs.
Thus, if we refer to a down edge $\{u_i,v_j\}$ in $G$ for vertices
$u_i$ and $v_j$ being copies of vertices $u$ and $v$ in $\tilde{G}$, 
we mean that $\{u,v\}$ is a down edge in the original graph. By our construction
every vertex $u_i$ in $G$ being a copy of a vertex $u$ in the original
graph is then incident 
to at most one down edge with another
endpoint of larger height---so that property (N) holds---and exactly one down 
edge with another endpoint  of lower height.
The endpoint of the latter
edge is then defined to be the down vertex $u_i\downvertex$ of $u_i$,
which then is also a copy of $u\downvertex$.      
Moreover, the down path of $u_i$, defined similarly as for almost
triangulated graphs, then consists of copies of the vertices of the
down path from $u$. However, different copies $u_i$ and $u_j$ in $G$ 
of the same vertex $u$ in $\tilde{G}$ have vertex-disjoint down paths
in $G$.  

We next
want to bound the 
number of edges of $G$ with respect to the number 
of edges of $\tilde{G}$. Note that, for each edge incident to a vertex $v$, but
not being the down edge of $v$, the splitting of $v$ introduces a new copy $v_i$
of $v$ and 
up to 
two additional edges into $G$ connecting $v_i$ to the down vertex of
$v_i$ and to a previous copy $v_{i-1}$ of $v_i$, respectively.
One of them namely, the down edge 
of $v_i$ will recursively cause further splittings along the down path of
$v_i$ so that altogether one edge of the original graph will introduce  
up to $2\ell$ new edges.
Consequently, if $\tilde{n}$ is the number of
vertices of the original graph, $G$ has $O(\ell\tilde{n})$ vertices and
edges. 

\begin{lemma}\label{lem:6.1}
For each given almost triangulated weighted $\ell$-outerplanar graph $(\tilde{G},\tilde{\varphi},\tilde{c})$
with $\tilde n$ vertices,
a weighted $\ell$-outerplanar graph $(G,\varphi,c)$ can be found in
$O(\tilde{n}\ell)$ time such that 
\begin{enumerate}
\item $G$ is 
a reverse minor of~$\tilde{G}$,
\item $G$
satisfies the neighborhood property, 
\item for each
each edge $\{u,v\}$,
there is 
an edge connecting a copy $u_i$ of $u$
and a copy $v_i$ of $v$ in $G$, and 
\item
the down paths of $u_i$ and $v_i$
in $G$ can be obtained from the down paths of $u$ and $v$,  
respectively, in $\tilde{G}$ by replacing the vertices of the down paths in $\tilde{G}$ 
by copies of its vertices in $G$.
\end{enumerate}
\end{lemma}

We now describe an algorithm to compute a tree decomposition for
the %
weighted $\ell$-outerplanar
graph
$({G},{\varphi},c)$ constructed above. Let $n$ be the
number of vertices of $G$ and $V$ be the vertex set of $G$. 
W.l.o.g., %
$G$ is biconnected; otherwise,
one can compute a tree
decomposition for each biconnected component independently and finally connect them.
 Take $S=V \setminus \{ u\downvertex\, \mid u \in V
 \ \mathrm{and}\ h^-_{\varphi}(u)\ge 2\}$ as illustrated
 in Fig.~\ref{fig:upSpannT}. %

 \begin{figure}[t!]
     \begin{center}
        \scalebox{0.97}{\input{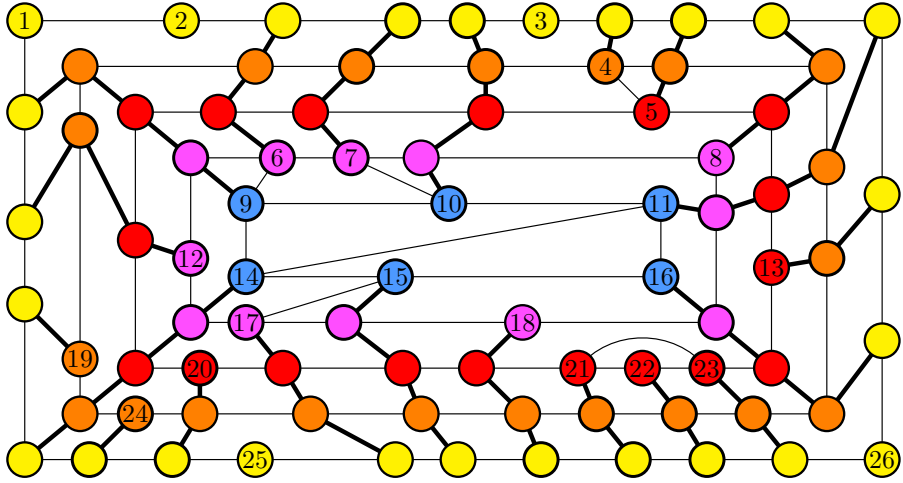}}%
     \end{center}
     \caption{A plane unweighted graph with the neighborhood
     property. Down edges are drawn bold.
      Vertices in the set $S$ are numbered.
              }
     \label{fig:upSpannT}
   \end{figure}
As a first step of our computation, we merge each down path starting at a
vertex $v$ of~$S$ to one vertex $v^*$ and define the weight of~$v^*$ as $h^+_{\varphi}(v)$. 
Let $(G^+,\varphi^+,c^+)$ be the graph obtained.
Note that $h^+_{\varphi}(v)=h^+_{\varphi^+}(v^*)$,
and $G^+$ is 
weighted $\ell$-outerplanar.
Since all vertices of~$(G^+,\varphi^+,c^+)$
are incident to the outer face, the unweighted version $(G^-,\varphi^-)$
of~$(G^+,\varphi^+,c^+)$ is outerplanar and we find a tree decomposition
$(T,B)$ for~$G^-$ of width~2 as shown by Bodlaender~\cite{Bod98}, i.e., in each bag we
have at most~3 vertices. Since each vertex in $G^+$ has a weight of at most
$\ell$, $(T,B)$ is a tree decomposition for~$(G^+,\varphi^+,c^+)$ with 
bags of size $3\ell$; in other words, $(T,B)$ as a tree decomposition
for~$(G^+,\varphi^+,c^+)$ has a
 width of at most $3\ell -1$.
By replacing each vertex $v^*$ by the down path of~$v$, we obtain a tree
decomposition for~$(G,\varphi,c)$ of width $3\ell -1$ in $O(n\ell)$ time.

As we show in the following, the algorithm from above can be slightly
modified to prove the next lemma.

\begin{lemma}\label{lem:Bod} Let %
$(G,\varphi,c)$ be an 
almost triangulated
weighted $\ell$-outerplanar biconnected graph,
and let 
$M=(G,\varphi,c,\mathcal{H},\mathcal{S})$
be a good mountain structure
of~$(G,\varphi,c)$.
Take $C$ as an $(\mathcal{S},\varphi)$-component,
and $C'=\mathrm{ext}(C,\mathcal{S})$.
Then one
can construct a tree decomposition $(T,B)$ of width %
$3\ell-1$
for  $C'$ that, for
each crest separator $X\in\mathcal{S}$ with a top edge in $C$, %
has a bag %
containing all vertices of~$X$.
Moreover, given $C'$ as well as the height intervals
$[h^-_\varphi(v),h^+_\varphi(v)]$
for each vertex $v$ in $C'$, the
construction can be done
in $O(n'\ell^2)$ time, where $n'$ is the number of vertices in $C'$.
\end{lemma}

\begin{proof}
For the time being, let us assume that no
 crest separator %
with a top edge in $C'$
has %
a lowpoint. At the end of our proof we
 show how to handle %
crest separators with a lowpoint.
The heights of the vertices in $C'$ with respect to $\psi=\varphi|_{C'}$
may differ from
the heights %
of these vertices with respect to $\varphi$.
In order to avoid this, %
we
insert additional %
edges into $C'$.
More precisely,
for each crest separator $X=(P_1,P_2)$ with a top edge in 
$C$,
height $r\in \Nat$,
and vertices $u_1,\ldots,u_p$ of~$P_1$ and $v_1,\ldots,v_q$
of~$P_2$ in the order of their appearance on these paths,
and for~$i\in \{1,\ldots,r\}$, define 
$u(i)$ as the vertex in $\{u_1,\ldots,u_p\}$ 
whose height interval contains $i$. 
Analogously, we define $v(i)$ for all $i\le h_\varphi^+(v)$.
We extend $C'$ by inserting vertices $m(1),m(2),\ldots,m(h_\varphi^+(v))$ into the
$(X,\varphi)$-component not containing $C$ as well as 
edges $\{u(i),m(i)\}$ and $\{m(i),v(i)\}$ for all $i$ with $1\le i\le h_\varphi^+(v)$
and edges $\{m(i),m(i+1)\}$ for all $i$ with $1\le i\le h_\varphi^+(v)-1$,
as well as edges between $m(h_\varphi^+(v))$ and each vertex
of $u(h_\varphi^+(v)+1),\ldots u(r)$.
By applying the changes above to all crest separators,
this results in
a weighted plane graph $(\tilde{C},\tilde{\varphi},\tilde{c})$ such that %
$h_{\tilde{\varphi}}(v)=h_\varphi(v)$ holds for all vertices~$v$ of~$C'$.  %
We next transform $(\tilde{C},\tilde{\varphi},\tilde{c})$ into a graph 
$(\hat{C},\hat{\varphi},\hat{c})$ with the neighbor hood property (N) by
applying Lemma~\ref{lem:6.1} and construct a tree decomposition
of width $3\ell-1$ for this new graph as it is described 
after Lemma~\ref{lem:6.1}.

Let us now analyze what happens to a crest separator $X = (P_1,P_2)$.
By Lemma~\ref{lem:6.1}, the 
top edge of $X$ is replaced by an edge in
$(\hat{C},\hat{\varphi},\hat{c})$ connecting copies of the original
endpoints. Moreover, the down paths of these endpoint copies consist of copies
of the vertices of $P_1$ and $P_2$.
After merging each down path 
in $(\hat{C},\hat{\varphi},\hat{c})$
to one vertex---let $G'$ be the graph
obtained---there is an edge connecting the two
vertices that are introduced for~$P_1$ and~$P_2$. Thus, the tree
decomposition for~$G'$ has a bag containing the two vertices. As a
consequence, we obtain a tree decomposition for~$\hat{C}$
 that contains a
bag with %
copies of all
vertices of~$P_1$ and~$P_2$. 
Since $\hat{C}$ is a reverse minor of $\tilde{C}$, we obtain
a tree decomposition for $\tilde{C}$ from a tree decomposition for 
$\hat{C}$ of at most the same width in a standard way by replacing 
split vertices and edges by the original vertices and edges in 
$\tilde{C}$ 
thereby
replacing the copies of the vertices of $P_1$ and $P_2$
by the original vertices of $P_1$ and $P_2$. After removing the
vertices $m(i)$ outside $C'$, we 
obtain a tree decomposition of width at most $3\ell-1$ for $C'$
that, for each crest separator $X\in\mathcal{S}$, has a bag containing
all vertices of $X$.

We next show how to exclude crest separators with
lowpoints by modifying the given graphs. 
For simplicity, our
modifications described below
do not result in an %
almost triangulated graph; 
however the graph can be easily transformed 
into an
almost triangulated graph
by adding into each inner face with more than three edges on its boundary
edges incident to one vertex of smallest upper height on the boundary,
which does not change any height interval.
Let us first consider a crest separator 
$X\in\mathcal{S}$ with a top edge in $C$ that encloses $C$.
Note that in this case every crest separator with a 
top edge in $C$ contains the lowpoint $v$ of $X$.
Let %
$i=h^+_\varphi(v)$. Then
we remove all vertices~$u$ with
$h^+_\varphi(u)\le i$ from $G$, from $C$, and from the crest
separators contained in $\mathcal{S}$, and additionally we
remove from $\mathcal{S}$
all crest separators %
that afterwards have at most one vertex (since, by definition, 
a single vertex is not a crest separator anymore).
For all vertices $u$ with 
$h^-_\varphi(u)\le i+1<h^+_\varphi(u)$,
we define $c'(u)=h^+_\varphi(u)-i$. For the remaining vertices $u$, we define $c'(u)=c(u)$.
We so obtain a new good
mountain structure $(G',\varphi',c',\mathcal{S}')$
from $(G,\varphi,c,\mathcal{S})$,
where $\varphi'$ is %
a weighted $(\ell-i)$-outerplanar embedding.
More precisely,
if the graph $G'$ after removing all vertices of upper height at most
$i$ is not biconnected, we
take
a good mountain
structure for each 
biconnected component.
Note
that, for each crest separator $X'=(P'_1,P'_2)$ of $\mathcal{S}$ with a top edge in $C$, the
subpaths $P_1^*$ and $P_2^*$ of $P'_1$ and $P'_2$, respectively, ending
immediately before $v$ are contained in the same biconnected component.
To see this,
we distinguish between two cases. If we have $h_\varphi^-(u)=h_\varphi^-(v)=i+1$ for both top vertices $u$ and $v$ of $X'$, $P_1^*$ and $P_2^*$ consist
only of these two vertices. Since these vertices are connected by an
edge, they must be part of the same biconnected component.
Otherwise we must have $h_\varphi^-(u)>i+1$ for at least one top vertex $u$ of $X'$.
Hence we can conclude that
there is a cycle of vertices with their height intervals containing $i+1$
that encloses $u$ and all vertices of $P_1^*$ and $P_2^*$ of lower height
at least $i+2$ and that contains the vertices of $P_1^*$ and $P_2^*$ with
lower height $i+1$. The inner graph of this cycle is biconnected
and since it contains the cycle itself it must contain all 
vertices of $P_1^*$ and 
$P_2^*$. Therefore, we separately construct a tree decomposition for each biconnected
component of $G'$ and then connect the tree decompositions of each
biconnected component. 
After the modifications---in particular, the removal of~$v$---no
crest separator in $\mathcal{S}'$ encloses
the new $(\mathcal{S}',\varphi')$-component $C^*$ obtained from $C$
(Lemma~\ref{lemma:enclose}).
The idea is then to use the construction as described for crest separators
with no lowpoints to construct a tree decomposition
of width
$3(\ell-i)-1$ 
for each biconnected component
such that, for each crest separator
$X'\in\mathcal{S}$
with a top edge in $C^*$, 
it 
has a bag containing all
vertices $u$ of~$X'$ with
$h^+_{\varphi}(u)> i$.
Since $C$ is enclosed by $X$, the remaining vertices of~$X'$ are all part
of the down path of~$v$ in $(G,\varphi,c)$. %
We can simply add the %
vertices
of the down path of~$v$ %
into
all bags of the tree decomposition to obtain the desired tree
decomposition for~$C'$ of width $3\ell -1$.

Crest separators with a lowpoint that does not enclose $C$
can be handled in the same way as crest separators without any
lowpoint except that we do need to define vertices $u(i), v(i), m(i)$
for all i smaller or equal than the upper height $h$ of the lowpoint
and that we insert an edge between the lowpoint and $m(h+1)$.

Concerning the running time, it is dominated by the
construction of a tree decomposition for $\hat{C}$. This construction 
takes $O(\hat{n}\ell)$ time, where $\hat{n}$ is the number of
vertices of $\hat{C}$. Since the replacement of~$\tilde{C}$ by $\hat{C}$ may
increase the number of vertices by a factor of~$O(\ell)$, the whole
running time is $O(n'\ell^2)$.
\end{proof}

\section{The Main Algorithm}\label{sec:alg}

In this section we describe our main algorithm.
As mentioned in Section~\ref{sec:ide}, we assume that
we are given an almost triangulated weighted 
graph $(G,\varphi,c)$ with weighted treewidth $k$.
In the case that no embedding $\varphi$ is given, we can compute an
arbitrary planar embedding in linear time~\cite{HopT74}.
Recall that $c_{\mathrm{max}}$ denotes the maximum weight over all vertices.
Let $\ell=2k+2c_{\mathrm{max}}$.
Our algorithm
starts with cutting off all maximal connected
subsets of vertices of whose height interval  contains a value of size
at least $\ell$
by coast separators of size $O(k)$.
More precisely, to find such coast separators, in a first substep we
merge each maximal connected set $M$ of vertices of whose height interval
contains a value of size 
at least 
$\ell$
to one vertex $v_M$ and define $c(v_M)$ to be $\ell+1$ minus
the smallest lower height of a vertex in $M$. This means that
the lower height of $v_M$ is the smallest lower height of a vertex in $M$
and its upper height is $\ell$.
Therefore
the weighted graph $(G',\varphi',c)$ obtained is
weighted $\ell$-outerplanar.
Given a vertex of the coast, this can be done in a time linear in the number of vertices
with a lower height of at most 
$\ell$.
$G'$ is an almost triangulated, biconnected graph since this is true
for~$G$. 
We then can use 
Corollary~\ref{cor:StrongSCut}
to 
construct, for the subset $\mathcal{H}$ of the vertices of 
height $\ell$ in $G'$,
a good mountain
structure $(G',\varphi',c,{\mathcal H},\mathcal{S})$ 
as well as 
to find a set of
coast separators $\mathcal{Y}$
and a function
$m$ that maps the coast separators to $(\mathcal{S},\varphi')$-components
such that the following properties of the lemma hold.
\begin{itemize}
\item Each crest of upper height exactly $\ell$ is enclosed by a 
coast separator $Y\in \mathcal{Y}$
of
 weighted size at most $3k+4c_{\mathrm{max}}-5$.
\item 
For each pair of crests of height $\ell$, the crests are  
either 
part of the inner graph of one $Y\in\mathcal{Y}$ or 
there is a crest separator $X\in \mathcal{S}$ strongly
going between the crests.
\item The $(\mathcal{S},\varphi')$-components in 
$m(Y)$
induce
 a connected subgraph of the mountain connection tree.

\item The inner graph of a
  coast separator $Y\in \mathcal{Y}$
 is a subgraph of the graph
 obtained from the union of the 
 $(\mathcal{S},\varphi')$-components
 in $m(Y)$.
\item For each $(\mathcal{S},\varphi')$-component $C$, there is at most
 one
   $Y\in \mathcal{Y}$
  with
 $C\in m(Y)$.
\end{itemize}

Since $G'$ is %
weighted $\ell$-outerplanar,
we can apply 
Lemma~\ref{lem:Bod}
to each $(\mathcal{S},\varphi')$-component $C$ to
compute a tree decomposition {$(T_C,B_C)$}
of width at most 
$3\ell-1$    
for~$\mathrm{ext}(C,\mathcal{S})$
such that, for each crest separator
$X\in\mathcal{S}$ with a top edge in $C$, {$(T_C,B_C)$} has a node
whose bag contains all
vertices of~$X$.
This node
is then connected 
to a node whose bag also contains all vertices of~$X$ and
that is
constructed for
$\mathrm{ext}(C',\mathcal{S})$ with $C'$ being the
other $(\mathcal{S},\varphi')$-component
$C'$ containing the top edge of~$X$. Since the
set of common vertices of
$\mathrm{ext}(C,\mathcal{S})$ and $\mathrm{ext}(C',\mathcal{S})$
is
a subset of the vertices of~$X$,
after also connecting nodes for all other crest 
separators in $\mathcal{S}$,
we obtain a tree decomposition $(T^*,B^*)$
for~$G'$.

Let us next remove from
$\mathcal{S}$ all crest separators whose top edge
is contained in two $(\mathcal{S},\varphi')$-components
belonging to the same set $m(Y)$ for some $Y\in\mathcal{Y}$.
Afterwards, for the new set $\mathcal{S}'$ of crest separators,
each cycle $Y\in\mathcal{Y}$ is contained in one
$(\mathcal{S}',\varphi')$-component.%
\footnote{We now have found a set of
crest separators and coast separators that guarantee (P1) - (P3)
from page \pageref{page:ideas}. The set $\mathcal{S}'$ is exactly the
set of perfect crest separators.}
For each $(\mathcal{S}',\varphi')$-component $C'$,
let us call the {\em flat component} of~$C'$
to be the subgraph of~$\mathrm{ext}(C',\mathcal{S}')$ obtained by
removing the vertices of the inner graph
of the cycle $Y\in\mathcal{Y}$
with $Y$ contained in $C'$
if such a cycle
$Y$ exists.
Otherwise,
we define the {\em flat component} to be
$\mathrm{ext}(C',\mathcal{S}')$, which then 
contains no vertex of upper height larger than $\ell$. 
From the bags in $(T^*,B^*)$,
we then remove all vertices that
do not belong to a flat component. 
Afterwards, for each cycle $Y\in\mathcal{Y}$ disconnecting
the crests of an $(\mathcal{S}',\varphi')$-component $C'$
from the coast,
we
put
the
vertices of~$Y$ into all bags
of the tree decompositions {$(T_C,B_C)$}
constructed as part of~$(T^*,B^*)$ for the (extended components of the) 
$(\mathcal{S},\varphi')$-components {$C$} contained in $C'$.
This allows us to connect one of these
bags with a bag of a tree decomposition for the inner graph of
$Y$.
Indeed, for each $Y\in\mathcal{Y}$, we recursively construct a
tree decomposition $(T_Y,B_Y)$ for the
inner graph $G_Y$ of~$Y$ with the vertices
of~$Y$ being the coast of~$G_Y$. Into all bags of~$(T_Y,B_Y)$
that are not constructed in further recursive calls, we also put
the vertices of~$Y$.
At the end of the recursions, we obtain a tree decomposition for the
whole graph. Note that
each bag is of weighted size $O(k)$. More precisely, let us consider
a recursive call that constructs a tree decomposition $(T_Y,B_Y)$
for a cycle $Y$ constructed in a previous step. Then the tree decomposition
for the flat component considered in the current recursive call puts
vertices with a total weight 
at most~$3\ell$
into each bag. However, we also have to insert the vertices
of~$Y$ and possibly the vertices of a cycle constructed in the
current recursive call into the bags. Since each of these cycles consists of
vertices with a total weight
at most~$3k+4c_{\mathrm{max}}-5$,
each bag of the final tree decomposition of~$G$ 
has weight
at most~$3\ell+6k+8c_{\max}-10\le 12k+14c_{\max}-10$.
Recall that \mbox{$c_{\max}\le k$.} 
As mentioned in Section~\ref{sec:ide},
it is possible to
replace a non-triangulated weighted graph $H$ of weighted treewidth $k$ by an almost triangulated
weighted supergraph $G$ of $H$
and to run our algorithm from above on $G$
such that
we can obtain a
tree decomposition for~$H$ of width 
$(12+\epsilon)k + 14c_{\max} + O(1)$.

Concerning the running time, it is easy to see that each
recursive call is dominated by the computation of the
cycles being used as coast separators.
This means that each recursive call
runs in $O(\tilde{n}k^2%
c_{\mathrm{max}} 
)$ time, where $\tilde{n}$ is the number
of vertices of the subgraph $G'$ of~$G$ considered in this call.
Some vertices part of one recursive
call are cut off from the current graph
and then are also considered in a further recursive call.
However, since 
the
coast separators 
contain no vertex of the coast,
the coast is
not part of any recursive call. 
Therefore, each vertex is considered in $O(k)$ recursive calls, and
our algorithm finds a tree decomposition for~$G$ of width $O(k)$ in
$O(|V|k^3%
c_{\mathrm{max}} 
)$ time.
If we do not know $k$ in advance, we can use a binary search to
determine
a tree decomposition for~$G$ of width
$O(k)$ in
$O(|V|k^3%
c_{\mathrm{max}} 
\log k)$ time.

\begin{theorem}\label{the:final}
For
a weighted planar graph $(G,c)$ with $n$ vertices and 
weighted treewidth $k$ and any constant $\epsilon>0$, a tree decomposition 
for~$G$ of weighted
width
$(12+\epsilon)k + 14c_{\max} + O(1)$
can be constructed in
$O(nk^3 
c_{\mathrm{max}} 
\log k)$ time, where $c_{\max}$ denotes the maximum weight of a
vertex of~$G$.
\end{theorem}

\begin{corollary}\label{col:final}
For
a planar graph $G$ with $n$ vertices and 
treewidth $k$ and any constant $\epsilon>0$, a tree decomposition 
for~$G$ of 
width
$(12+\epsilon)k + O(1)$
can be constructed in
$O(nk^3\log k)$ time.
\end{corollary}

For a more efficient algorithm, we replace $\ell$ by 
$3k+2c_{\mathrm{max}}$. When considering a weighted $\ell$-outerplanar
graph $(G,\varphi)$ in one recursive call of the algorithm, we
remove all vertices of upper height at most $k$, reduce the weight of
each vertex $v$ with upper height at least $k+1$ and lower height
smaller than $k+1$ by $k$, and search for a coast separator in the
resulting weighted $(2k+2c_{\mathrm{max}})$-outerplanar graph as in our
original algorithm. 
Since now the lower and upper heights of each vertex considered 
in two consecutive recursive steps differ by at least $k$, every
vertex is now considered in at most $O(1)$ recursive calls. 
However, we now have to construct a tree decomposition for
a weighted $(3k+2c_{\mathrm{max}})$-outerplanar flat component in each recursive call.
Thus,
we now construct a tree decomposition 
with a weight of size at most
$3\ell+6k+8c_{\max}-10=15k+14c_{\max}-10$
per bag if $G$ is
almost triangulated,
and $(15+\epsilon)k+14c_{\max}+O(1)$
per bag otherwise.

\begin{theorem}\label{the:finalFast}
For a weighted planar graph $(G,c)$ with $n$ vertices and weighted treewidth $k$
and any constant $\epsilon>0$, a tree decomposition 
for~$G$ of 
weighted width
$(15+\epsilon)k+14c_{\max}+O(1)$ can be constructed in
$O(nk^2 %
c_{\mathrm{max}} 
\log k)$ time, where $c_{\max}$ denotes the maximum weight of a
vertex of~$G$.
\end{theorem}

\begin{corollary}\label{col:finalFast}
For a planar graph $G$ with $n$ vertices and treewidth $k$                    
and any constant $\epsilon>0$, a tree decomposition 
for~$G$ of 
width
$(15+\epsilon)k+O(1)$ can be constructed in
$O(nk^2\log k)$ time. 
\end{corollary}

It is also interesting to compute a grid minor if a graph
has no tree decomposition of size $O(k)$. We can do so, for an unweighted
planar graph $G'$, if we abstain from multiplying the weights of the vertices
of $G'$ by a factor $x$ 
during the transformation of $G'$ into an almost triangulated
graph $G$ by %
inserting a vertex into each face  
and connecting it to all vertices on the boundary.
Kloks et al.~\cite[Theorem 2]{KloLL02} showed that 
the treewidth of $G$ can be bounded by $k=4k'+1$ where $k'$ is the treewidth of $G'$.
If, afterwards, the algorithm fails to construct a tree decomposition
for the graph $G$ without the weight changes of unweighted
tree\-width $k$, 
then this can only happen when we search for a separator with Theorem~\ref{the:Sep}. %
In this case, we have $k$ internally-vertex-disjoint paths $\mathcal P$ that all start and
end at two vertices whose heights differ by more than $k$, i.e., each of these paths
``crosses'' $k-1$ cycles. If we cut the cycles between a 
consecutive pair of paths of
$\mathcal P$ and remove 
the endpoints of the paths, we get a $k \times (k-1)$ minor
in the almost triangulated version of $G$. 
If we remove every second path of $\mathcal P$ and every second cycle, then we can 
replace
the
remaining paths and cycles such that no new vertex added into a face is used and the
paths (the cycles) are pairwise vertex-disjoint. Thus, $G$ has a 
grid of size $\lfloor k/2 \rfloor \times \lfloor(k-1)/2 \rfloor$ as minor.

\begin{theorem}\label{the:grid}
Given a planar graph $G$ with $n$ vertices and $k\in \Nat$, 
there is an algorithm that constructs
either 
a tree decomposition for~$G$ of    
width $O(k)$ or a $\Theta(k)\times \Theta(k)$ grid as a minor of $G$ 
in $O(nk^2)$ time.
\end{theorem}

We finally want to remark that it was not the purpose of this paper to show
the smallest possible approximation ratio. Indeed, it
is not necessary to turn a given graph
into an almost triangulated graph and to compute subsequently a tree decomposition
for the almost triangulated graph.

With more sophisticated techniques,
Kammer~\cite{Kam10} presented an algorithm that can compute a tree decomposition of width 
$9tw(G)+9$
for an
unweighted planar graph $G$ in $O(nk^4\log k)$ time.

\section{Conclusion}\label{sec:concl}

We have shown that a tree decomposition for 
a planar graph 
with its width approximating the treewidth by a constant factor
can be found in a time linear in the number of vertices of the
given graph. Since a tree decomposition for a planar graph with $n$
vertices and treewidth $k$ can be of size $\Theta(nk)$, an interesting open question is
if the running time of Corollary~\ref{col:finalFast} can be improved to $O(nk)$.

To obtain a better approximation ratio, 
Kammer~\cite{Kam10} has shown how
to adapt our algorithm from triangulated planar graphs
to general planar graphs. This makes the algorithm much more
complicated. %
Another more promising approach would be to find
a linear-time triangulation of a planar
(weighted) graph %
without increasing the treewidth of the graph.

We also want to mention that it is a still an open problem whether 
the treewidth on planar graphs can be 
found in
polynomial time or 
whether the problem is NP-hard.

%

\addcontentsline{toc}{section}{References}

\bibliography{ourbib}  %

\end{document}